\documentclass[12pt]{article}
\usepackage{graphicx} 
\usepackage[margin=1in]{geometry}
\usepackage{lscape}
\usepackage{pdflscape}
\usepackage{amsmath,amssymb}
\usepackage{tabularx}
\usepackage{booktabs}
\usepackage{bbm}
\usepackage{ccaption}
\usepackage{dcolumn}
\usepackage{multirow}
\usepackage{nicefrac}
\usepackage{float}
\usepackage{subcaption}
\usepackage{comment}
\usepackage{xcolor}
\usepackage[singlespacing,nodisplayskipstretch]{setspace}
\usepackage{color}
\usepackage{url}
\usepackage{varioref}
\usepackage{eurosym}
\usepackage{mdframed}
\usepackage{rotating}
\usepackage{enumitem}
\usepackage{makecell}
\usepackage[normalem]{ulem} 
\newtheorem{theorem}{Theorem}

\usepackage{algorithm}
\newtheorem{assumption}{Assumption}
\newtheorem{estimand}{Estimand}

\newtheorem{example}{Example}

\newtheorem{conjecture}{Conjecture}
\newtheorem{corollary}[theorem]{Corollary}

\newtheorem{lemma}{Lemma}

\newtheorem{proposition}[theorem]{Proposition}
\newtheorem{remark}{Remark}

\newenvironment{proof}[1][Proof]{\noindent\textbf{#1.} }{\ \rule{0.5em}{0.5em}}
\usepackage{csquotes}
\usepackage[style=apa,uniquename=false,backend=biber,sortcites,sorting=ynt]{biblatex}
\definecolor{navy}{rgb}{0,.3,.7}
\definecolor{myorange}{RGB}{238, 166, 63 }
\definecolor{myblue}{RGB}{29, 130, 231 }
\definecolor{green}{RGB}{49, 178, 54}
\definecolor{darkred}{RGB}{169,5,5}
\definecolor{darkgreen}{RGB}{0,152,0}
\usepackage[pdfpagemode=UseThumbs,%
     colorlinks=true,%
     urlcolor=navy,%
     citecolor=teal,%
     linkcolor=teal]{hyperref}
\usepackage{cleveref}
\graphicspath{{figures/}}
\usepackage[font=bf]{caption} 
\newcolumntype{d}[1]{D{.}{.}{#1}}
\newcolumntype{e}{D{.}{.}{-1}}
\newcolumntype{.}{D{.}{.}{3}}
\newcommand{\ind}{\perp\!\!\!\!\perp }
\newcommand{\indicator}{\mathbbm{1}}
\newcommand{\E}{\mathbb{E}}
\newcommand{\bs}{\boldsymbol}

\newcommand{\attao}{\text{ATT}_{\text{AO}}}
\newcommand{\qttao}{\text{QTT}_{\text{AO}}}

\begin{document}


\title{\large Estimating the Intensive Margin Effect in Panel Data Settings\thanks{A previous version of this paper has been circulated with the title ``Difference-in-Differences and Changes-in-Changes with Sample Selection''. An R package for implementing the proposed methodology is under development. I am thankful to Fabrizia Mealli, Andrea Ichino, and Alessandro Tarozzi for insightful discussions and guidance. I am also grateful to Francesco Drago, Erin Gabriel, Ellen Greaves, Bo Honoré, Guido Imbens, Miguel Nunes, Jörn-Steffen Pischke, and seminar participants at the European University Institute, EuroCIM 2025, ACIC 2025, iCEBDA 2025, and SAEe 2025 for valuable comments and suggestions. All errors are mine.
}
\vspace{.8 cm}
}
\vspace{.8 cm}
\author{Javier Viviens \thanks{Department of Economics, European University Institute. Email: javier.viviens@eui.eu.} }

\date{\vspace{.8 cm} February 4,  2026 \\
}
\maketitle

\begin{abstract}
\noindent
Many policies operate through two different channels: the extensive margin (e.g., the decision to participate) and the intensive margin (e.g., the intensity of the response among participants). This paper develops a novel identification strategy to estimate the intensive margin effect in panel data settings. I adapt the Horowitz-Manski-Lee bounds to the Changes-in-Changes framework to partially identify both the average and quantile intensive margin treatment effects. Additionally, I explore how to leverage multiple sources of sample selection to relax the monotonicity assumption in the original Horowitz-Manski-Lee bounds, which may be of independent interest. Alongside the identification strategy, I present estimators and inference results. I illustrate the relevance of the proposed methodology by analyzing a job training program in Colombia.

\vspace{0.5 cm}
\noindent \textbf{JEL-Code:} C14, C21, C23. 

\vspace{0.2cm}
\noindent \textbf{Keywords:} Intensive Margin, Difference-in-differences, Changes-in-Changes, Sample selection, Panel Data, Repeated Cross-Sections, Principal stratification, Partial identification.
\end{abstract}

\thispagestyle{empty}

\newpage

\setcounter{page}{1}

\begin{spacing}{1.5}

\section{Introduction}
A central concern in economics is understanding the mechanisms through which policies operate. Aggregate policy effects often mask distinct channels through which individuals are affected. 
A prominent instance arises in policies that impact the extensive margin (e.g., the decision to participate) and the intensive margin (e.g., the intensity of the response among participants). 
A canonical example of such policies is job training programs, which affect employment (extensive margin) and also wages for those employed (intensive margin) \parencite{heckman_economics_1999, ham_effect_1996,angrist_estimation_2001}.\footnote{Additional examples include, but are not limited to, education policies that affect enrolment and performance for those enrolled \parencite{krueger_experimental_1999,cornwell_student_2005}; carbon taxes that affect firms' entry and exit and emissions for active firms \parencite{greenstone_effects_2012}; weather shocks that affect the decision to migrate and income for those who do not migrate \parencite{groger_internal_2016}. In the medical literature, this phenomenon is known as truncation due to death \parencite{zhang_estimation_2003,ding_identifiability_2011,rubin_causal_2006}.} 
Distinguishing between these margins is essential for both welfare analysis and policy design. 

The main contribution of this paper is to develop an identification strategy for estimating the intensive margin effect in panel data settings. The existing literature has proposed identification and estimation strategies for intensive margin effects in settings where the treatment can be assumed uncounfounded (e.g., experimental settings).\footnote{\label{fn:literature_lee}This literature was pioneered by \textcite{zhang_estimation_2003}, \textcite{imai_sharp_2008}, \textcite{lee_training_2009}, and \textcite{zhang_likelihood-based_2009}, who introduced a principled approach to partially identify the intensive margin effect. Subsequent research has extended this literature in multiple directions, including sharpening bounds with covariates \parencite{semenova_generalized_2025, long_sharpening_2013, grilli_nonparametric_2008,samii_generalizing_2025, ding_identifiability_2011}, instruments \parencite{mattei_identification_2014}, post-censoring outcomes \parencite{yang_using_2016}, or a structural model for selection \parencite{honore_selection_2020}, incorporating longitudinal data \parencite{comment_survivor_2025, grossi_bayesian_2023}, considering noncompliance \parencite{chen_bounds_2015, blanco_bounds_2020}, targeting different estimands \parencite{huber_sharp_2015,bartalotti_identifying_2023}, and exploring non-binary treatment regimes \parencite{lee_lee_2025}. However, all these papers rely on the assumption that treatment is unconfounded.} This assumption is often implausible in panel data settings, where treatment is confounded with unobservable characteristics \parencite{liu_practical_2024,arkhangelsky_causal_2024,ghanem_selection_2025}. This paper fills this gap by adapting intensive-margin estimands and identification strategies to panel data settings in which treatment is not randomly assigned. 

I formalize the intensive margin effect using principal stratification analysis \parencite{frangakis_principal_2002}. This framework 
characterizes each unit by its potential participation decision under both treatment arms, thereby partitioning the sample into four distinct subpopulations or principal strata. For example, in the job training case, those who are always employed, those employed only if trained, those employed only if not trained, and those who are always unemployed. The extensive margin effect refers to the treatment effect on participation (employment in the job training example). 
The intensive margin effect is defined as the treatment effect on the response 
for the subset of units that would participate under both treatment arms, and hence for whom there is no effect on participation, i.e., no extensive margin effect. In the job training case, these are the workers who are employed regardless of whether they receive the training or not. For these workers, once participation is defined in terms of employment, any remaining treatment effect operates through the intensive margin, such as changes in wages.

The identification strategy developed in this paper focuses on the Always-Observed principal stratum, i.e., the units with observed outcomes regardless of their treatment status.\footnote{\label{fn:bad_control}The methodology developed in this paper can also be applied in settings where outcomes are observed for all units, but we are interested in isolating the intensive margin effect. Consider the `bad control' example in \textcite{angrist_mostly_2009}, where we are interested in the effect of education on earnings. Part of this effect operates through occupational choice. In this example, we may be interested in the effect of education on earnings net of occupational changes. In this case, we would estimate the effect on the subpopulation of units that always have a white-collar occupation, regardless of their education level.} The effect on these units can be interpreted as the policy's intensive margin effect. To achieve identification, I extend 
the Howorowitz-Manski-Lee bounds \parencite{horowitz_identification_1995, lee_training_2009} to the Changes-in-Changes (CiC) model \parencite{athey_identification_2006}. The identification strategy relies on 
the assumption 
that all the unobserved confounders affecting the outcome for the Always-Observed units can be captured in a single latent index (for example, ability in the case of wages). The distribution of this index may differ for the Always-Observed units in the treatment and control groups, but it is assumed to be constant over time within each group (that is, the distribution of ability for workers who are always employed can be different in the treatment and control groups as long as this distribution does not vary over time). Combined with information on the principal strata proportions, this assumption delivers partial identification of two causal estimands: (i) the Average Treatment Effect on the Treated for the Always-Observed stratum ($\attao$) and (ii) the Quantile Treatment Effect on the Treated for the same stratum ($\qttao$). 

Another contribution of this paper is the estimation of the principal strata proportions. I also employ the CiC approach to model participation. Combined with a monotonicity assumption, as in the original Horowitz-Manski-Lee bounds, this method delivers point identification of the principal strata proportions. The required monotonicity assumption implies that treatment can affect participation only in `one direction' for all the units, ruling out the existence of one of the strata (e.g., units observed only if treated).\footnote{This assumption is analogous to the monotonicity assumption in IV settings \parencite{angrist_identification_1996} that rules out the existence of defiers.} In the job training context, this assumption implies that all workers employed without training would also be employed if trained. I relax this assumption when the data include multiple sources of participation. For instance, wages may be unobserved for both unemployed workers and workers who decide to migrate. I assume each source is monotonic while allowing the direction of monotonicity to differ across sources. For example, the training may reduce the probability of unemployment but increase the likelihood of migration for all units. This extension enables all four principal strata to exist while still achieving point identification of their proportions. This novel extension of the original Horowitz-Manski-Lee bounds could also be of interest in settings where treatment is unconfounded. 

I illustrate the proposed methodology by analyzing a job training program in Colombia. Using data from \textcite{attanasio_subsidizing_2011}, I estimate the intensive margin effect of the training on salaried earnings. I contrast these findings with those obtained using a `naive' approach, 
in which the CiC estimator is applied to units with observed wages. 
The `naive' approach suggests that the training increased wages by 13\% for trained units. In contrast, using the methodology presented in this paper, I fail to reject the null that the training's intensive margin effect was zero. When examining the $\qttao$, I find that the estimated bounds vary across the outcome distribution, suggesting heterogeneous intensive margin effects with potential distributional implications. 

To date, the literature on panel data 
has focused on correcting selection bias and targeting the ATT \parencite{heckman_sample_1979,hausman_attrition_1979,fitzgerald_analysis_1998,ghanem_correcting_2024,bhattacharya_inference_2008, semykina_estimating_2010,carlson_addressing_2024}. While these approaches recover causal effects under sample selection, their estimand cannot be interpreted as the intensive margin effect.\footnote{This estimand identifies the effect for all treated units and thus integrates over all the principal strata, mixing intensive and extensive margin effects.} More
recently, there has been growing attention to estimating the intensive margin effect in panel data settings. In parallel with this work, \textcite{shin_difference--differences_2024, rathnayake_difference--differences_2024} extend Horowitz-Manski-Lee bounds to a Difference-in-Differences (DiD) setting and target one of the estimands proposed in this paper (the ATT on the Always-Observed principal stratum). The methodology proposed by this paper can also be applied to the DiD design, as illustrated in Appendix \ref{ap:ext_did}.
A key distinction is that both \textcite{shin_difference--differences_2024, rathnayake_difference--differences_2024} identify the principal strata proportions under a parallel trends assumption in participation, whereas I identify them using the Changes-in-Changes model.
An advantage of my approach is that it guarantees estimated probabilities that lie in the unit interval. 
Moreover, this paper is the first to establish partial identification of distributional effects, the $\qttao$. 
This estimand is particularly relevant when policymakers are interested not only in the average impact of a policy but also in its distributional effects. By shifting the focus from means to distributions, the CiC framework enables the exploration of heterogeneity of treatment effects along the outcome distribution. Finally, I consider several practical extensions—including covariates, binary outcomes, and repeated cross sections—which further enhance the applicability of the proposed framework.

The remainder of this paper is organized as follows. Section \ref{sec:bi} introduces the notation and set-up. Section \ref{sec:hiru} introduces the identification strategy to estimate causal effects for the Always-Observed units. Section \ref{sec:lau} derives the estimation and inference results. Section \ref{sec:bost} applies the methodology to data from a job training program from \textcite{attanasio_subsidizing_2011}. Section \ref{sec:sei} concludes. The Appendix sections describe the extensions of the main results (\ref{ap:extensions}), prove them (\ref{ap:proofs}) and provide supplementary material (\ref{ap:supmat}).

\section{Notation and set up} \label{sec:bi}
For simplicity, I consider the two-groups two-periods case. Consider a random sample of $N$ units observed during two periods, $t \in \{1,2\}$. Let $w_{it}$ denote the treatment of unit $i$ at time $t$. No one is treated in the pre-treatment period, $w_{i1} = 0$ $\forall i$. Some units are treated in the post-treatment period, $w_{i2}\in\{0,1\}$. Let $G$ denote a treatment group indicator, with $G_{i} = 1$ if unit $i$ is treated in the second period and zero otherwise, $G_{i} \equiv w_{i2}$. Let $Y_{it}$ be the outcome of unit $i$ at time $t$. I adopt the potential outcome framework and assume the Stable Unit Treatment Value Assumption (SUTVA) \parencite{rubin_randomization_1980, rubin_application_1990}. This assumption implies that there is no interference between units and there are no hidden versions of treatment, allowing me to relate the observed $Y_{it}$ to the potential outcomes of unit $i$, $Y_{it}(w_{i1},w_{i2})$:
\begin{equation*}
        Y_{it} = Y_{it}(0,1)G_{i} + Y_{it}(0,0)(1 - G_{i}).
\end{equation*}

The outcome for some sampled units may not be observed. Let $S_{it}$ be a binary selection indicator equal to 1 if the outcome of unit $i$ is observed at time $t$ and zero otherwise. A crucial aspect of this setup is the recognition that the selection indicator, $S$, is also a post-treatment variable. Consequently, I must also relate the observed selection outcome, $S_{it}$, to the potential selection outcomes:
\begin{equation*}
    S_{it} = S_{it}(0,1)G_{i} + S_{it}(0,0)(1 - G_{i}).
\end{equation*}

Throughout the paper, I will maintain the assumption of no treatment anticipation for both the outcome and the selection, a standard assumption in the literature \parencite{arkhangelsky_causal_2024}. 
\begin{assumption} \label{as:no_anti}
    No treatment anticipation
    \begin{align*}
        S_{i1}(0,w_{i2}) = S_{i1}(0,w_{i2}^{\prime}) =S_{i1} \quad \forall w_{i2},w_{i2}^{\prime}\forall i ,\\
        Y_{i1}(0,w_{i2}) = Y_{i1}(0,w_{i2}^{\prime}) = Y_{i1} \quad \forall w_{i2},w_{i2}^{\prime}\forall i .
    \end{align*}
\end{assumption}
Assumption \ref{as:no_anti} implies that the treatment does not affect the outcome and the selection in the pre-treatment period. Since no unit is treated in the first period, it is sufficient to index the second-period potential outcomes with $w_{i2}$:
\begin{align*}
    S_{i2} = S_{i2}(1)G_{i} + S_{i2}(0)(1-G_{i}), \\
    Y_{i2} = Y_{i2}(1) G_{i} + Y_{i2}(0)(1-G_{i}).
\end{align*}

Since selection is a post-treatment variable, the population can be partitioned into subpopulations defined by joint values of the potential selection indicator in the post-treatment period, i.e., $(S_{i2}(1), S_{i2}(0))$. This approach is known as principal stratification \parencite{frangakis_principal_2002}. Every unit belongs to one of the four principal strata. Table \ref{tab:strata} presents the principal strata, where $V_{i}$ denotes the principal stratum of unit $i$. Because potential outcomes are not affected by the treatment, the principal strata exist prior to the treatment assignment. While the observed selection status is a post-treatment outcome, the principal strata act as pre-treatment covariates.\footnote{In the example discussed in footnote \ref{fn:bad_control}, occupational choice is a post-treatment variable and deemed a bad control \parencite{angrist_mostly_2009, rosenbaum_consequences_1984}. Nevertheless, principal strata can be defined in terms of the potential occupational choices under each treatment arm. The treatment determines which of these potential choices is realized, but does not determine to which principal stratum the unit belongs.} 
\begin{table}[H] \centering
\caption{Principal Strata.}
\begin{tabular}{|cc|c|c|}
\hline
$\bs{S_{i2}(0)}$ & $\bs{S_{i2}(1)}$  & \textbf{Description} & $\bs{V_{i}}$ \\ \hline 
   1&1 & Always Observed  & AO \\
    0&0 & Never Observed  & NO \\
    1&0 & Observed only in Control & OC \\
    0&1 & Observed only in Treatment & OT    \\ \hline
\end{tabular}
\label{tab:strata}
\begin{minipage}{0.81\linewidth  \setstretch{0.75}}
{\scriptsize  Notes: Principal Strata defined by the joint value of the potential selection indicators, $S(0)$ and $S(1)$. For example, consider the case where the treatment is a job training program aiming to increase workers' wages. The AO stratum comprises workers who are always employed, both with and without training. The NO stratum consists of workers who are always unemployed, regardless of whether they are treated or not. The OC stratum refers to workers who would find a job without the training but would be unemployed if they were treated. Conversely, the OT stratum includes workers who are unemployed in the absence of treatment but find a job if they receive it. }
 \end{minipage}
\end{table}

To facilitate the exposition of the proposed methodology, I impose the following assumption on missingness:
\begin{assumption} \label{as:abstate}
    Missingness as an absorbing state
    \begin{equation*}
        S_{i1} = 0 \implies S_{i2} = 0.
    \end{equation*}
\end{assumption}
Assumption \ref{as:abstate} implies that once a unit’s outcome is unobserved, it remains unobserved in all subsequent periods. Thus, all the units with $S_{i1} = 0$ will belong to the Never-Observed stratum at $t=2$. In contrast, units with $S_{i1} = 1$ can transition to any stratum in the second period. Beyond this, Assumptions \ref{as:no_anti} and \ref{as:abstate} do not impose additional constraints on the proportion of the strata in the population, which may vary across groups, nor on the distribution of potential outcomes $Y(0)$ and $Y(1)$ across strata.

Assumption \ref{as:abstate} is plausible in settings with censoring due to death or similar events, such as clinical trials where patients cannot be observed after death, or educational contexts where students who drop out or graduate do not subsequently re-enroll. 
In other settings, however, researchers may be reluctant to impose this assumption, as outcomes that are unobserved in the pre-treatment period may be observed at a later date.
This assumption imposes testable restrictions on the data, since $S_{i1}$ and $S_{i2}$ are observed for all units.
I impose Assumption \ref{as:abstate} to simplify the exposition of the methodology, but it is not required for the identification of the causal effects considered in this paper. Appendix \ref{ap:ext_as2} presents the identification results of this paper without this assumption.

It is also important to note that this assumption does not correspond to restricting the sample to those observed in the first period. Allowing for missingness in the pre-treatment period is crucial for imputing the missing potential selection outcomes, which is central to the methodology developed in this paper. Restricting the sample to units with observed outcomes at baseline would discard valuable information contained in baseline selection. At the same time, as discussed in Appendix \ref{ap:ext_as2}, under Assumption \ref{as:abstate}, the identification strategy in the main text targets causal effects only for units with observed outcomes at baseline. Nonetheless, this strategy can be extended to target causal effects regardless of whether the units were observed in the pre-treatment period, as shown in Appendix \ref{ap:ext_as2}. 

Finally, I formalize the random sampling assumption in Assumption \ref{as:rand_samp}.
\begin{assumption} \label{as:rand_samp}
    Random Sampling. $\{Y_{i1},Y_{i2}, S_{i1}, S_{i2},G_{i}\}_{i=1}^{N}$ are independent and identically distributed with joint distribution $P_{Y,G,S}$, and supp$(Y_{it}) \subset \mathbb{R}\cup\{*\}$, supp$(S_{it}) = \{0,1\}$, and supp$(G_{i}) = \{0,1\}$.
\end{assumption}
Assumption \ref{as:rand_samp} implies that the units are sampled from a larger population. Even when the outcome $Y_{it}$ is not observed for some units, the researcher knows that those units were included in the sample. I use the convention that $Y_{it} = *$ for units with $S_{it} = 0$. For example, wages are only observed for employed workers, but researchers also observe which workers are unemployed; quality of life is only observed for survivors, but researchers know which patients died; and academic performance is only observed for enrolled students, but researchers know which students dropped out. Accordingly, the expectations throughout the paper are meant to capture the sampling-based uncertainty \parencite{abadie_sampling-based_2020}.

\section{Identification of intensive margin effects}\label{sec:hiru}
This section presents a novel identification strategy to estimate the intensive margin effect in longitudinal settings. 
Unlike previous literature on intensive margin effects,\footnote{See footnote \ref{fn:literature_lee}.} I do not assume treatment unconfoundedness. Instead, I allow the treatment to be confounded with unobservable unit-level characteristics. However, I impose restrictions on these unobservables' distribution over time. In what follows, I propose two causal estimands and discuss their identification. 

\begin{estimand} \label{es:attao}
    Average Treatment Effect on the Treated Always-Observed units ($\text{ATT}_{\text{AO}}$)
    \begin{equation*}
        \attao = \E[Y_{i2}(1) - Y_{i2}(0) \mid G_{i} = 1, V_{i} = AO]
    \end{equation*}
\end{estimand}
Estimand \ref{es:attao} is the ATT for the subpopulation of Always-Observed units. This estimand captures the average effect for treated units whose outcome is observed under both treatment arms, and therefore can be interpreted as the policy's intensive margin effect. This estimand is especially relevant in settings where the researcher is interested in the causal effect net of any compositional changes. 

In many applications where treatment effects are likely heterogeneous, policymakers may be interested not only in the average treatment effect but also in its distributional impact. For instance, consider the case of a job training program. A positive $\attao$ indicates that, on average, the training increases wages for the workers employed both with and without the program. However, this average positive effect might be driven primarily by workers at the lower end of the wage distribution or, conversely, by those at the top. These contrasting scenarios carry significant policy implications, as the distributional effects of the training differ substantially. To fill this gap, I propose a new estimand that targets the treatment effect at various quantiles of the outcome distribution for the Always-Observed treated units.

\begin{estimand} \label{es:qttao}
    Quantile Treatment Effect on the Treated Always-Observed units ($\qttao$)
    \begin{align*}
        \qttao(q) &= Q_{Y_{2}(1) \mid G = 1, V_{} = AO}^{}(q)-Q_{Y_{2}(0) \mid G = 1, V_{} = AO}^{}(q), \\
        Q_{Y}^{}(q) &= \inf\{y : F_{Y}(y) \geq q\}, \\
        & q \in[0,1]
    \end{align*}
    where $F_{Y_{t}(w) \mid G = 1, V_{} = AO}$ denotes the cumulative distribution of the Potential Outcome $Y_{t}(w)$ at period $t$ for units that belong to the the treatment group, $G=1$, and to the Always-Observed stratum, $V = AO \equiv S_{2}(1) =S_{2}(0) = 1$.
\end{estimand}
Estimand \ref{es:qttao} contrasts the quantile $q$ of the treated potential outcome distribution for AO units in the treatment group with the same quantile of the control potential outcome distribution for these units. For example, $\qttao(0.5)$ captures the treatment effect at the median of the outcome distribution. This estimand captures a distributional effect: it does not require any rank-invariance assumption across potential outcomes and should not be confused with the median treatment effect \parencite{imbens_causal_2015}. 

Both Estimands \ref{es:attao} and \ref{es:qttao} are principal causal effects, in that they define causal effects for a subpopulation characterized by a principal stratum. Because principal stratum membership depends on joint potential outcomes, which are never simultaneously observed, this subpopulation is not identifiable without additional assumptions. Analogously to the IV setting, where individual compliance types cannot be directly observed, it may not be possible to know which units belong to the Always-Observed stratum. This feature may be viewed as a limitation of the proposed methodology. Nevertheless, because these estimands admit a clear interpretation as intensive margin effects, they are policy-relevant and practically meaningful for evaluating interventions and informing scale-up decisions. Finally, as in every principal stratification analysis, it is advisable to characterize the units belonging to the Always-Observed principal stratum \parencite{mealli_refreshing_2012}.

An alternative to targeting principal causal effects is given by bias-correction methods, that attempt to recover the ATT in the presence of sample selection. This estimand can be interpreted as ``the treatment effect if all the treated units were observed'' and therefore cannot be interpreted as the intensive-margin effect without further assumptions. Moreover, bias-correction methods have several limitations. First, they assume that outcomes are always well-defined, even if they are not directly observed \parencite{zhang_likelihood-based_2009}. In the job training example, bias-correction methods assume the existence of a `wage offer', allowing us to consider what the wage of an unemployed worker is. Second, they believe that the selection status is manipulable \parencite{mealli_refreshing_2012}, meaning that employment status could, in principle, be directly manipulated by the researcher.\footnote{When selection status is manipulable, one can define potential outcomes that can never be observed in the data and are `a priori counterfactuals', such that the potential wage without training for a worker who is only employed if trained. Another example of this type of potential outcome is given by the IV setting, where we can hypothesize about the potential outcome under treatment for a never-taker. See \textcite{mealli_refreshing_2012} for a discussion.} Third, bias-correction methods often rely on exclusion restriction assumptions, requiring variables that affect employment but not wages \parencite{lee_training_2009} and place restrictions on how potential outcomes are distributed across principal strata.

Next, I outline the assumptions required for partial identification of these estimands. First, I discuss the partial identification of both the $\attao$ and the $\qttao(q)$ in a Changes-in-Changes setting. I then examine the identification of the principal strata proportions, which constitutes a key component of the proposed methodology, under alternative sets of assumptions.

\subsection{Identification of causal effects}
The existing literature estimating the effects for the Always-Observed units has relied on the assumption that treatment is unconfounded. However, this assumption is often unrealistic in causal models for panel data, where treatment assignment is typically non-random and likely correlated with unobserved factors \parencite{ghanem_selection_2025, liu_practical_2024}. To address this limitation, I reformulate the Horowitz-Manski-Lee bounds \parencite{horowitz_identification_1995,lee_training_2009} by replacing the unconfoundedness assumption with the Changes-in-Changes (CiC) model as in \textcite{athey_identification_2006}, allowing the treatment assignment to be confounded with unobservable characteristics that affect the potential outcomes. 

The choice of the CiC outcome modeling is meant to address several limitations in the canonical DiD research design. The canonical DiD may be sensitive to functional form \parencite{roth_when_2023} and implicitly relies on additivity and linearity assumptions on the potential outcomes. To address these issues, \textcite{athey_identification_2006} generalize the DiD identification strategy to a nonlinear and scale-free model, Changes-in-Changes. In the CiC model, individual and time `fixed effects' can influence the outcome in more flexible ways beyond the linear and additive fashion typically imposed in DiD. The CiC setup also offers additional advantages. First, shifting the focus from means to distributions enables the identification of distributional treatment effects, such as the $\qttao$. Second, it avoids extrapolation outside the outcome's support, ensuring that the counterfactual outcomes $Y(0)$ for treated units lie within the support of $Y$. On the other hand, because the method reconstructs the full distribution of $Y(0)$ for treated units, it requires assumptions that are not needed in the DiD framework, where identification concerns only the first moment of this distribution. As a result, the canonical DiD is not nested as a particular case of the CiC research design (see \textcite[Remark 4]{roth_when_2023} for an example). Nevertheless, the methodology developed in this paper can be adapted to the canonical DiD setting, as shown in Appendix \ref{ap:ext_did}.

\begin{assumption} \label{as:outcome}
    Outcome model for Always-Observed units \\
    \begin{equation*}
        (Y_{it}(0) \mid V_{i} = AO )\quad =\quad  m(\mathcal{U}_{it}, t),
    \end{equation*}
    where $m(u,t)$ is strictly increasing in $u$ $\forall t$, and $\mathcal{U}_{it}$ is an unobservable scalar for unit $i$ at time $t$ with constant distribution over time within groups,
    \begin{equation*}
        \mathcal{U}_{i1} \mid G_{i}, V_{i} = AO  \sim \mathcal{U}_{i2} \mid G_{i}, V_{i} = AO.
    \end{equation*}
\end{assumption}
Assumption \ref{as:outcome} is the principal counterpart to assumptions 3.1-3.4 in \textcite{athey_identification_2006}. It models the untreated potential outcome as a function of an unobservable individual characteristic, $\mathcal{U}_{it}$. This model embeds several restrictions on the unobservable, its distribution, and the outcome function $m(\cdot)$. These assumptions do not restrict the data in any way and, thus, are not testable.

This model assumes that all the individual characteristics affecting the outcome can be captured in a single index, $\mathcal{U}_{it}$. The function $m(\cdot)$ being strictly increasing implies that units with higher values of the unobservable also have higher outcomes. This assumption is intuitive when the unobservable is interpreted as an individual characteristic such as ability or productivity. It holds by construction in additively separable models, such as the two-way fixed effects model. However, it also accommodates a broad range of non-additive nonlinear outcome functions. 

The distribution of the unobservable may differ between the treatment and control groups. However, it remains constant over time within each group. This assumption is crucial, as it allows interpreting changes in the outcome in the control group as changes in the function $m(\cdot)$. Once the trend in $m(\cdot)$ is estimated, the missing potential outcome for the treatment group can be identified.

\begin{example} \label{ex:twfe}
    (Two-way fixed effects model). Assume a separable additive model. Let $m(\mathcal{U}_{it},t) = \mathcal{U}_{it} + \lambda_{t}$, and assume $\mathcal{U}_{it} = \alpha_{i} + \varepsilon_{it}$, where $\varepsilon_{it}$ is random noise. Then, the untreated potential outcome for AO units can be expressed as:
    \begin{equation*}
        Y_{it}(0) = \alpha_{i} + \lambda_{t} + \varepsilon_{it}.
    \end{equation*}
    This example illustrates that the standard two-way fixed effect model can be seen as a special case of the CiC outcome model.
\end{example}

\begin{proposition} \label{prop:bounds_qttao}
    Let $Y_{it}(1)$ and $Y_{it}(0)$ be continuous with compact 
    support. Furthermore, let the support of $Y_{it}(0)$ for the treatment group be contained in the support of $Y_{it}(0)$ for the control group. Then, if Assumptions \ref{as:no_anti}, \ref{as:abstate}, \ref{as:rand_samp} and \ref{as:outcome} hold, then $\Lambda^{LB}(q)$ and $\Lambda^{UB}(q)$ are 
    lower and upper bounds for the Quantile Treatment Effect on the Treated Always-Observed units ($\qttao(q)$), where:
    \begin{align*}
        \Lambda^{LB}(q) &= Q_{Y_2 \mid G=1, S_{2} = 1}(q\pi_1) - Q_{Y_{2} \mid G=0, S_2 = 1}\left(F_{Y_{1}\mid G=0, S_2 = 1}\left(Q_{Y_{1}\mid G=1,S_2 = 1}(q\pi_1 + 1 - \pi_1)\right) + 1 - \pi_0 \right), \\
         \Lambda^{UB}(q) &= Q_{Y_2 \mid G=1, S_{2} = 1}(q\pi_1 + 1 - \pi_1) -  Q_{Y_{2} \mid G=0,S_2 = 1}\left(F_{Y_{1}\mid G=0, S_2 = 1}\left(Q_{Y_{1}\mid G=1,S=1}(q\pi_1)\right)-(1-\pi_0)\right), \\
        \pi_{1} & = Pr(S_{i2}(0) = 1 \mid G_{i} = 1, S_{i2}(1) = 1),\\
       \pi_{0} & = Pr(S_{i2}(1) = 1 \mid G_{i} = 0, S_{i2}(0) = 1), \\
       F_{Y}(y) & := Pr(Y \leq y) \\
       Q_{Y}(q) & := \inf\{y: F_{Y}(y) \geq q\}
    \end{align*}
    provided that 
    \begin{align*}
        F_{Y_{1}\mid G=0, S_2 = 1}\left(Q_{Y_{1}\mid G=1,S_2 = 1}(q\pi_1 + 1 - \pi_1)\right) &\leq \pi_0 \\
        F_{Y_1 \mid G=0, S_2 = 1}\left(Q_{Y_1\mid G=1, S_2 = 0} (q\pi_1)\right) &\geq 1 - \pi_0.
    \end{align*}
    If $F_{Y_{1}\mid G=0, S_2 = 1}\left(Q_{Y_{1}\mid G=1,S_2 = 1}(q\pi_1 + 1 - \pi_1)\right) > \pi_0$, then 
    \begin{equation*}
         \Lambda^{LB}(q) = Q_{Y_2 \mid G=1, S_{2} = 1}(q\pi_1) - Q_{Y_{2} \mid G=0, S_2 = 1}\left(1\right).
    \end{equation*}
    If $F_{Y_1 \mid G=0, S_2 = 1}\left(Q_{Y_1\mid G=1, S_2 = 0} (q\pi_1)\right) < 1 - \pi_0$, then 
    \begin{equation*}
        \Lambda^{UB}(q) = Q_{Y_2 \mid G=1, S_{2} = 1}(q\pi_1 + 1 - \pi_1) -  Q_{Y_{2} \mid G=0,S_2 = 1}\left(0\right)
    \end{equation*}
    Proof: See Appendix \ref{proof:bounds_qttao}.
\end{proposition}
\begin{remark}
    The two conditions, $$F_{Y_{1}\mid G=0, S_2 = 1}\left(Q_{Y_{1}\mid G=1,S_2 = 1}(q\pi_1 + 1 - \pi_1)\right) \leq \pi_0$$ and $$F_{Y_1 \mid G=0, S_2 = 1}\left(Q_{Y_1\mid G=1, S_2 = 0} (q\pi_1)\right) \geq 1 - \pi_0,$$ ensure that extrapolation is not required outside the support of $Y_{2} \mid G=0$. When these conditions do not hold, the lower bound for $Q_{Y_2(0) \mid G=1, V=AO}(q)$ is given by the infimum of the support of $Y_2 \mid G=0$, and the upper bound by its supremum. This implies that, for extreme quantiles, the bounds on $\qttao(q)$ may be uninformative. Moreover, the smaller the proportion of Always-Observed units, $\pi_1$ and $\pi_0$, the larger the region where these bounds may be uninformative. Note that when $\pi_0 = 1$, both conditions are always satisfied. 
\end{remark}
\begin{remark} \label{prop:remark4}
    Once the distributions of $Y_{2}(1)_{ \mid G = 1, V_{} = AO}$ and $Y_{2}(0)_{ \mid G = 1, V_{} = AO}$ are partially identified, the bounds for the $\attao$ follow immediately:
    $$\int_{0}^{1}  \Lambda^{LB}(q)dq \leq \attao \leq \int_{0}^{1}  \Lambda^{UB}(q)dq.$$
\end{remark}
\begin{remark}
    Proposition \ref{prop:bounds_qttao} builds on continuity and support assumptions. These results can be extended to cases where the support of $Y_{it} \mid G_i=1$ does not fully overlap with that of $Y_{it}\mid G_{i} = 0$; see \textcite[Corollary 3.1]{athey_identification_2006} for a discussion. Appendix \ref{ap:ext_discrete} explores the identification for discrete outcomes. However, the identification is unsuitable when $Y$ is mixed or has non-compact support.
\end{remark}
\begin{remark} \label{prop:remark3}
    When the OT and OC strata are empty, it follows that $\pi_{1} = \pi_{0} = 1$, $\Lambda^{LB} = \Lambda^{UB}$, and the $\qttao(q)$ is point identified. The intuition behind this result is that all the units observed in the post-treatment period belong to the AO stratum. In this scenario, $S$ becomes an indicator of membership in the AO stratum.
\end{remark}

Proposition \ref{prop:bounds_qttao} describes a trimming procedure to bound the $\qttao(q)$. It builds on \textcite{horowitz_identification_1995}, who derive bounds for the components of a mixture distribution when the mixture proportions are known—here, the proportions of Always-Observed units in each group. I combine this with results from \textcite{athey_identification_2006}, who derive the counterfactual distribution of $Y_{2}(0)$ for treated units. Once the distributions of $Y_{2}(1)$ and $Y_{2}(0)$ for the treated AO units are bounded, the identification of the $\qttao(q)$ and the $\attao$ follows. 

The intuition behind Proposition \ref{prop:bounds_qttao} is as follows. Suppose that the proportion of AO units among treated units with observed outcomes, $\pi_1$, is known. For example, if $\pi_1 = 0.9$, the observed distribution for treated units trimmed at the 90\% percentile constitutes a lower bound for the real latent distribution for these 90\% AO treated units. Similarly, the top 90\% of the observed distribution constitutes an upper bound. Analogously, the distribution of outcomes for control AO units can be partially identified using the corresponding proportion, $\pi_0$.
Once the outcome distributions for AO units are partially identified, they can be combined with the results from  \textcite{athey_identification_2006} to partially identify the $\qttao(q)$. Under the CiC model for the AO units, any change in the control group’s outcome distribution must come from changes in the outcome mapping $m(\cdot)$, as the distribution of $\mathcal{U}_{it}$ is assumed to be stable over time within groups. This lets me use the control group to recover how $m(\cdot)$ shifts over time. I can then apply the same shift to the pre-treatment distribution of $\mathcal{U}_{it}$ in the treatment group, inferred from their pre-treatment distribution of the outcome, to construct the counterfactual post-treatment distribution. Figure \ref{fig:ex_inv_dist} provides graphical intuition behind Proposition \ref{prop:bounds_qttao}.

\begin{figure}[h] 
    \caption{Graphical intuition on Proposition \ref{prop:bounds_qttao}.}
    \centering
       \vskip\baselineskip
     \raisebox{0.1cm}{$\Lambda^{LB}$}
    \fbox{
        \begin{minipage}{0.98\textwidth}
            \centering
            \begin{minipage}{0.48\textwidth}
                \centering
                \includegraphics[width=\textwidth]{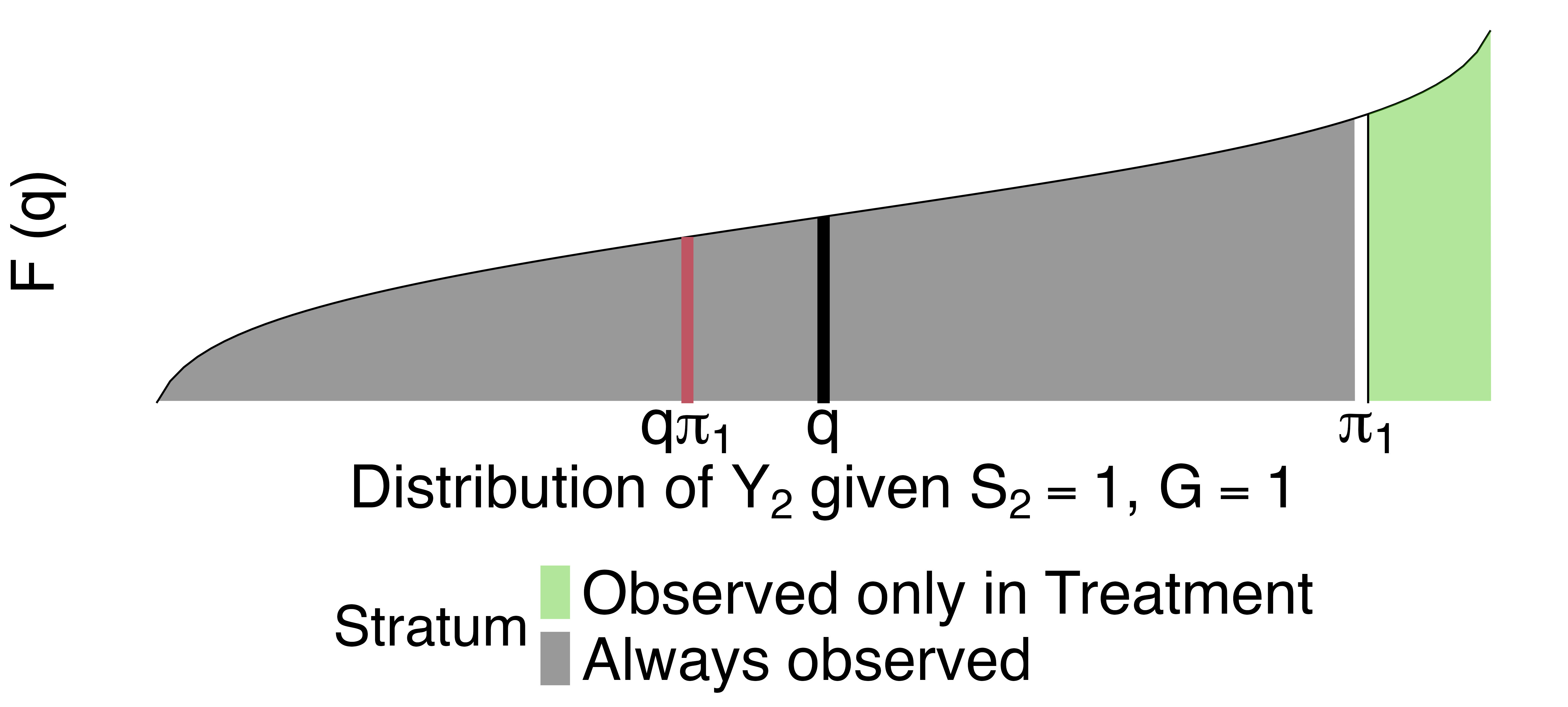}
            \end{minipage}%
            \hfill
            \begin{minipage}{0.48\textwidth}
                \centering
                \includegraphics[width=\textwidth]{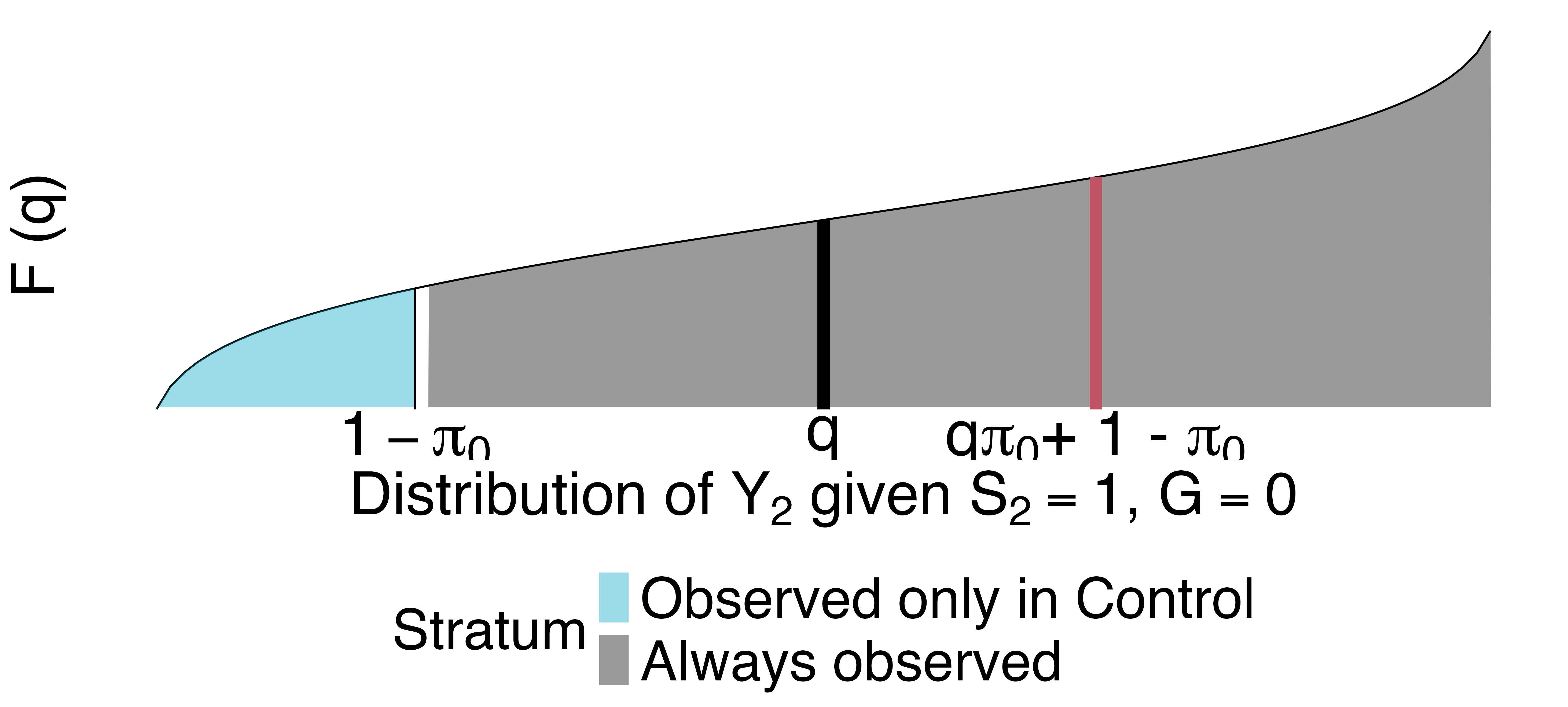}
            \end{minipage}
        \end{minipage}  
    }

    \vskip\baselineskip

     \raisebox{0.1cm}{$\Lambda^{UB}$}  
    \fbox{
        \begin{minipage}{0.98\textwidth}
            \centering
            \begin{minipage}{0.48\textwidth}
                \centering
                \includegraphics[width=\textwidth]{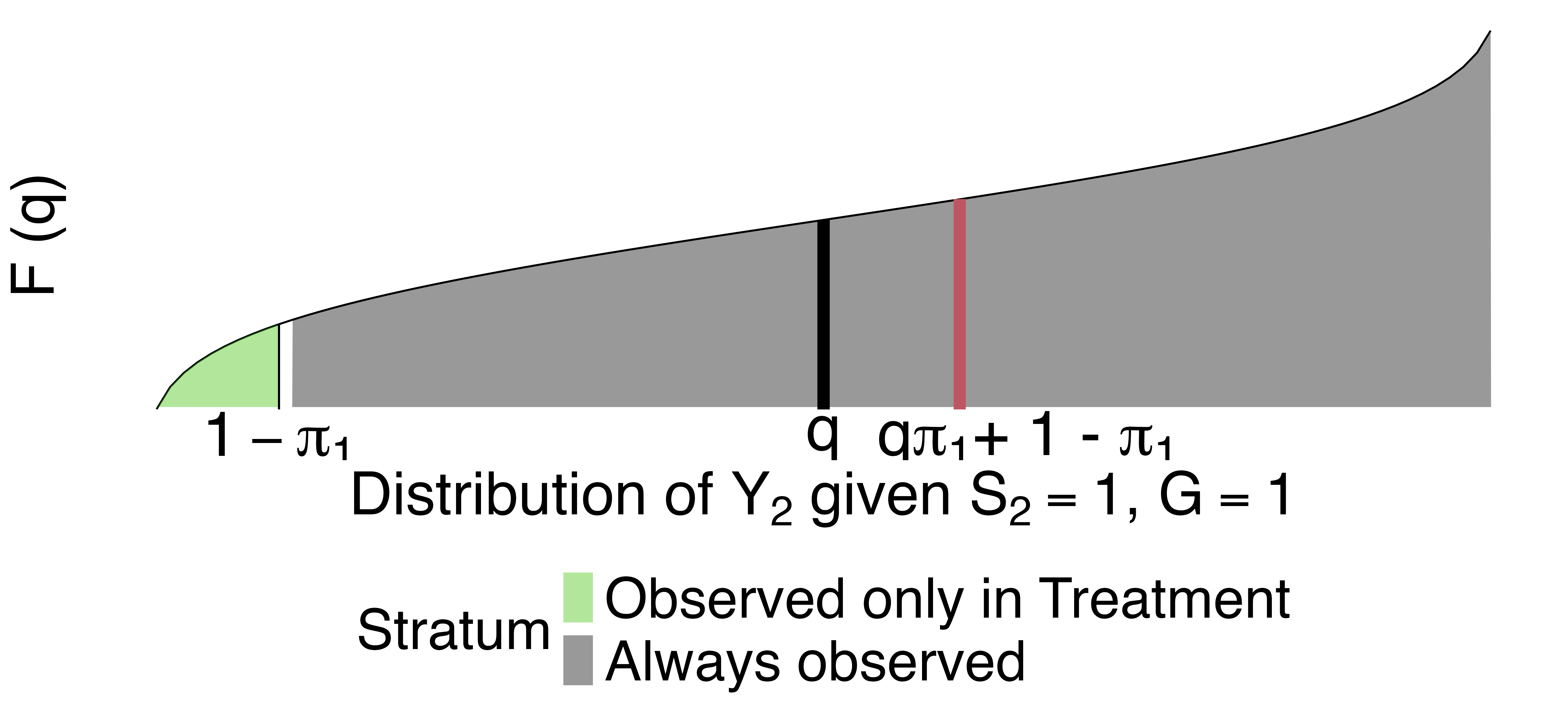}
            \end{minipage}%
            \hfill
            \begin{minipage}{0.48\textwidth}
                \centering
                \includegraphics[width=\textwidth]{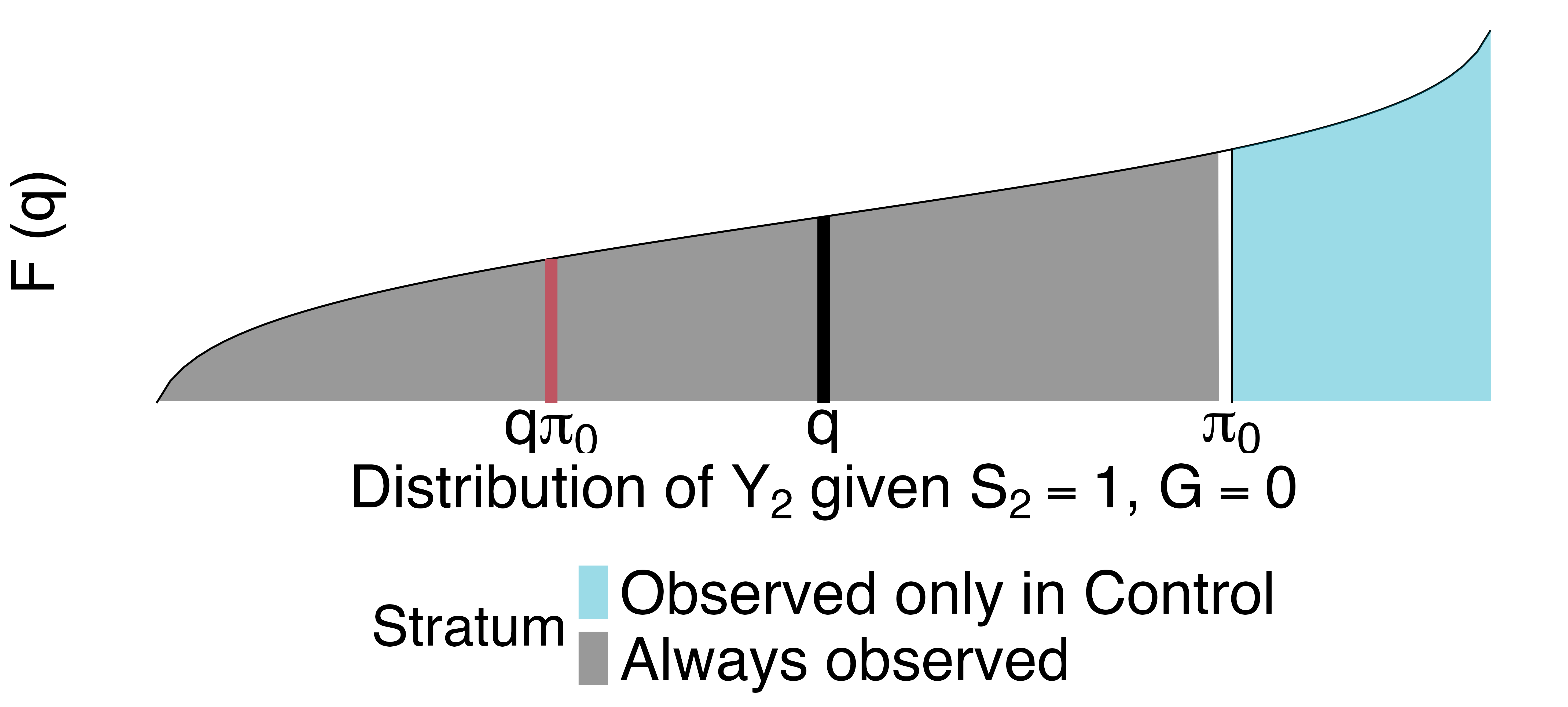}
            \end{minipage}
        \end{minipage}  
    }
    \vskip\baselineskip
\vspace{-0.2cm}
\begin{minipage}{0.99\linewidth  \setstretch{0.75}}
{\scriptsize Notes: This figure plots a hypothetical distribution of $Y_2 \mid S_{2} =1$. The distributions on the left correspond to the Treated group ($G=1$) and the one on the right to the Control group ($G=0$). The shaded areas correspond to the proportion of strata in the given group: black corresponds to the AO units ($\pi_{1}$ on the left and $\pi_{0}$ on the right),  green corresponds to the OT stratum and blue to the OC. The vertical black lines denote where a given quantile $q$ lies in the observed distribution. The vertical red lines denotes where the same quantile $q$ lies in the trimmed distributions, used in Proposition \ref{prop:bounds_qttao}. This figure also illustrates Remark \ref{prop:remark3}: as $\pi_{1}$ and $\pi_{0}$ go to 1, the black shadowed areas expand, and the trimmed (red) quantiles converge to those of the entire distribution (black).
}
 \end{minipage}
 \label{fig:ex_inv_dist}
\end{figure}

\subsection{Identification of principal strata proportions} \label{sec:hiru.prop}
A pivotal factor of Proposition \ref{prop:bounds_qttao} is the proportion of AO units in each group, $\pi_{1}$ and $\pi_{0}$. Next, I impose additional assumptions on the selection mechanism that enable the identification of these proportions.

\begin{assumption} \label{as:mono}
    Monotonicity
    \begin{align*}
        S_{i2}(1) \geq S_{i2}(0) \quad & \forall i \quad \text{Positive Monotonicity} \\
       & \text{or} \\
        S_{i2}(1) \leq S_{i2}(0) \quad & \forall i \quad \text{Negative Monotonicity}         
    \end{align*}
\end{assumption}
Assumption \ref{as:mono} implies that treatment can affect selection only in ‘one direction’ for all the units and rules out one of the four principal strata. Positive monotonicity excludes the OC stratum. This implies that the control units observed in both periods consist solely of Always-Observed units: $S_{i2}(1) \geq S_{i2}(0) \implies Pr(S_{i2}(1) = 0\mid S_{i2}(0) = 1) = 0 \implies \pi_{0} = 1$. Conversely, negative monotonicity rules out the OT stratum, assuming that all the treated units observed in both periods belong to the AO stratum. While restrictive, this assumption is standard in the literature and fundamental to ensure point identification of principal strata proportions. It is analogous to the monotonicity assumption that rules out the existence of defiers in IV settings. In section \ref{sec:hiru.mono}, I propose a novel methodology to relax this assumption when different sources of selection are available in the data.

The preference for positive or negative monotonicity depends on the context. For example, consider a researcher estimating the effect of job training on wages. It is reasonable to assume that a worker employed without training would also be employed with it. This assumption makes positive monotonicity the preferred assumption. Conversely, take the example of a researcher studying the effect of increased tuition fees on university students' performance. In this case, sample selection arises as some students drop out. Under the assumption that higher education costs increase dropouts, it is reasonable to conclude that if a student remains enrolled despite the fee increase, they would also stay enrolled without it. In this case, negative monotonicity is preferred.

When positive (negative) monotonicity is assumed, $\pi_{0}$ ($\pi_{1}$) is equal to one, and only the treatment (control) group needs to be trimmed. I adopt the CiC \parencite{athey_identification_2006} methodology to estimate the proportion of AO units in this group. This identification strategy assumes that an unknown function of individual unobservables determines the potential selection outcomes. This approach has some key advantages over the canonical DiD framework when applied to selection: it allows the identification of the missing potential selection outcomes for both groups. This is crucial when trimming the control group under negative monotonicity. Furthermore, it ensures that all the estimated probabilities lie in the unit interval.

\begin{assumption} \label{as:sele}
    Selection Model \\
    Under Positive Monotonicity:
    \begin{equation*}
    \begin{aligned}
        S_{it}(0) = h^{0}(U_{it},t) .
    \end{aligned}
    \end{equation*}
    Under Negative Monotonicity:
    \begin{equation*}
        \\
        S_{it}(1) = h^{1}(U_{it},t),
    \end{equation*}
    where $U_{it}$ is an unobservable scalar for unit $i$ at time $t$ and $h^{w}(u,t)$ is non decreasing function in $u$ $\forall$ $t \in\{1,2\}$, $w = \{0,1\}$. \\
    The unobservable $U_{}$ is continuously distributed and has the same compact support in both groups. Its distribution is constant over time within groups
    \begin{equation*}
        U_{i1} \mid G_{i} \sim U_{i2} \mid G_{i} .
    \end{equation*}
    Additionally, given the realized selection outcome, the distribution of $U$ is independent of the group in a given time period.
    \begin{equation*}
        U_{it} \perp \! \! \! \perp G_{i} \mid S_{it}
    \end{equation*}
\end{assumption}
Assumption \ref{as:sele} models the selection mechanism as a function of an unobservable individual characteristic, $U_{it}$. Depending on the monotonicity direction in Assumption \ref{as:mono}, a model for $S_{it}(0)$ or $S_{it}(1)$ needs to be assumed. Under Positive (Negative) Monotonicity, the control (treatment) group is composed solely of AO units, and the treatment (control) group is a mixture of AO and OT (OC) units, and therefore $\E[S_{i2}(0)\mid G_{i} =1, S_2 = 1]$ $\left(\E[S_{i2}(1)\mid G_{i} = 0, S_2 = 1]\right)$ needs to be imputed. This selection model shares significant similarities with the outcome model in Assumption \ref{as:outcome}. The main components of both models are the unobservable scalars, $U$ in the case of selection, and $\mathcal{U}$ in the case of the outcome. While they may be independent, the potential correlation between the unobserved factors influencing selection and those determining the outcome is central to the non-ignorable sample selection problem.

There are three noteworthy differences between the outcome model in Assumption \ref{as:outcome} and the selection model in Assumption \ref{as:sele}, albeit both share the general intuition behind a CiC model. First, depending on the direction of the monotonicity assumption, imputing the missing potential selection outcomes may be necessary for the treated or control units. This requires specifying functions for both potential selection outcomes, $S(1)$ in case negative monotonicity holds, and $S(0)$ in case positive monotonicity holds. In contrast, only the missing potential outcome $Y(0)$ for treated units needs to be imputed. Thus, modeling $Y_{it}(1)$ is not required. Second, the unknown function $m(\cdot)$ is assumed to be strictly increasing in the unobservable. Since $S$ is binary, the equivalent assumption cannot be made for $h^{0}(\cdot)$ and $h^{1}(\cdot)$. Nevertheless, the non-decreasing assumption ensures that observed units will not have lower values of the unobservable than the units that leave the sample. Third, Assumption \ref{as:sele} states that $ U_{it} \perp \! \! \! \perp G_{i} \mid S_{it}$. While the corresponding assumption, $ \mathcal{U}_{it} \perp \! \! \! \perp G_{i} \mid Y_{i}, V_{i} = AO$, is not explicitly made, it is embedded in Assumption \ref{as:outcome}. Since the function $m(u,t)$ is strictly increasing in $u$, it follows that $\mathcal{U}_{it} = m^{-1}(Y_{it}, t)$. Thus, when conditioning on $Y$ and $t$, the distribution of $\mathcal{U}$ becomes degenerate and, therefore, identical for both groups. As a result, the conditional independence assumption is trivially satisfied.

\begin{lemma} \label{prop:selection}
    If Assumptions \ref{as:no_anti}, \ref{as:abstate}, \ref{as:rand_samp}, \ref{as:mono}, and \ref{as:sele} hold, then the missing potential selection outcomes are given by:

    Under Positive Monotonicity,
    \begin{align*}
        \E[S_{i2}(0) \mid G_{i} = 1] = \E[S_{i1} \mid G_{i} = 1]\frac{\E[S_{i2} \mid G_{i} =0]}{\E[S_{i1} \mid G_{i} = 0]}.
    \end{align*}
    
    Under Negative Monotonicity,
    \begin{equation*}
        \E[S_{i2}(1) \mid G_{i} = 0] = \E[S_{i1} \mid G_{i} = 0] \frac{\E[S_{i2} \mid G_{i} = 1]}{\E[S_{i1}\mid G_{i} = 1]}.
    \end{equation*}
\end{lemma}
\begin{remark}
    Both $ \E[S_{i2}(0) \mid G_{i} = 1]$ and $\E[S_{i2}(1) \mid G_{i} = 0]$ will always lie in the unit interval. They will always be well-defined, except for the case when $\E[S_{i1} \mid G=g]=0$, which would mean that no unit is observed in the pre-treatment period in group $g$.
\end{remark}
Lemma \ref{prop:selection} is equivalent to Theorem 4.2 in \textcite{athey_identification_2006}, and the proof can be found in \textcite[458]{athey_identification_2006}. Under the CiC model for selection, the missing selection outcome in the post-treatment period for group $g$  is given by their baseline selection in the pre-treatment period scaled by the proportional change for the other group. Table \ref{tab:example_cic_selection} provides a numerical example for illustration purposes. 
\begin{table}[H] \centering
\caption{Numerical example of Lemma \ref{prop:selection}}
\begin{tabular}{|c|c|c|c|}
\hline
Group ($G$) & $\E[S_{1}]$ & $\E[S_{2}(G)]$ & $\E[S_{2}(1-G)]$  \\ \hline 
   Treatment $(G=1)$ & 0.9 & 0.55 & 0.45 \\
   Control $(G = 0)$ & 0.7 & 0.35 & 0.43 \\ \hline
\end{tabular}
\label{tab:example_cic_selection}
\begin{minipage}{0.81\linewidth  \setstretch{0.75}}
{\scriptsize  Notes: This table illustrates Lemma \ref{prop:selection} with a numerical example. Columns 2 and 3 display the observed mean in the selection indicator in each of the groups. Column 2 corresponds to the pre-treatment period and column 3 to the post-treatment period. The fourth column presents the missing potential selection outcome for each of the groups, imputed using Lemma \ref{prop:selection}.}
 \end{minipage}
\end{table}

\begin{remark} \label{prop:remark5}
    The treatment effect on selection, $\E[S_{i2}(1) - S_{i2}(0) \mid G_{i} = g]$ will have the same sign for both groups $g \in \{0,1\}$. This is consistent with the monotonicity assumption. Furthermore, these two effects will be the same if the expected selection in the pre-treatment period is the same across groups, i.e., $\E[S_{i1} \mid G_{i}= 1] = \E[S_{i1} \mid G_{i} = 0]$. This aligns with the selection model in Assumption \ref{as:sele} as
    \begin{equation*}
        \E[S_{i1} \mid G_{i}= 1] = \E[S_{i1} \mid G_{i} = 0] \iff U_{it} \ind G_{i} \implies \E[S_{i2}(1) - S_{i2}(0) ] \ind G_{i}
    \end{equation*}
    Proof: See Appendix \ref{proof:remark5}
\end{remark}

\begin{proposition} \label{prop:proportions}
    Under Assumptions \ref{as:no_anti}, \ref{as:abstate}, \ref{as:rand_samp}, \ref{as:mono}, and \ref{as:sele}, the proportion of Always-Observed units in the treatment group ($\pi_{1}$) and control group ($\pi_{0}$) are identified as follows:
    \begin{itemize}
        \item[] If Positive Monotonicity holds:
        \begin{align*}
            \pi_{0} &= 1 \\
            \pi_{1} & = \frac{\E[S_{i2}(0) \mid G_{i} = 1]}{\E[S_{i2} \mid G_{i} = 1]} = \frac{\E[S_{i1} \mid G_{i} = 1]}{\E[S_{i2} \mid G_{i} = 1]}\frac{\E[S_{i2} \mid G_{i} = 0]}{\E[S_{i1} \mid G_{i} = 0]} \in [0  ,1]
        \end{align*} 
        \item[] If Negative Monotonicity holds:
        \begin{align*}
            \pi_{0} &= \frac{\E[S_{i2}(1) \mid G_{i} = 0]}{\E[S_{i2} \mid G_{i} = 0]} = \frac{\E[S_{i1}\mid G_{i} = 0]}{\E[S_{i2} \mid G_{i} = 0 ]}\frac{\E[S_{i2} \mid G_{i} = 1]}{\E[S_{i1}\mid G_{i} = 1]} \in [0,1] \\
            \pi_{1} &= 1
        \end{align*}
    \end{itemize}
    Proof: See Appendix \ref{proof:proportions}.
\end{proposition}

\subsection{Relaxing the monotonicity assumption} \label{sec:hiru.mono}
Assumption \ref{as:mono} can be restrictive, as it requires the treatment to affect selection in only one direction for all the units, even when selection may respond to treatment through multiple channels. Similarly, Assumption \ref{as:sele} implies that all unobservables affecting selection can be summarized in a single scalar, ruling out different unobservables affecting selection in distinct ways. When multiple sources of selection are observed, the monotonicity assumption can be relaxed by allowing each source to exhibit monotonicity independently, potentially with different signs. For instance, consider a scholarship program designed to improve academic performance. Such a scholarship may affect selection through two distinct mechanisms: reducing dropout rates (positive monotonicity) while increasing graduation rates (negative monotonicity).

Suppose that there are $J$ different sources of sample selection, where $s_{it}^{j}$ is a binary variable equal to 0 if unit $i$'s outcome is not observed because of source $j$ and 1 otherwise. In the previous example, $J=2$, with $s_{it}^{1}$ equal to 0 if the student $i$ drops out at period $t$, and $s_{it}^{2} = 0$ if the student graduates. By definition, these sources are mutually exclusive, meaning that when a unit leaves the sample, it is due to a specific source $j$. For instance, students may leave the sample either because they graduated or because they dropped out before completing the degree, but both events cannot happen simultaneously. This framework allows the expression of the generic selection indicator as the product of the different sources of selection: $S_{it} = \prod_{j=1}^{J}s_{it}^{j}$, where the mutual exclusivity of the sources implies that $\sum_{j = 1}^{J}s_{it}^{j} \in \{J-1, J\}$ $ \forall i,t$. This notation extends to potential selection indicators as well: $S_{it}(w) = \prod_{j=1}^{J}s_{it}^{j}(w)$. 

The four strata defined in Table \ref{tab:strata} exist for each source $j$. I can define analogously a variable $V_{i}^{j}$ indicating the principal stratum of unit $i$ defined by the joint values of the potential selection indicator of the source $j$, $(s_{i2}^{j}(1), s_{i2}^{j}(0))$.
\begin{assumption} \label{as:cond_mono}
    Source-specific monotonicity \\
    For any source $j\in\{1,..., J \}$:
        \begin{align*}
        s_{i2}^{j}(1) \geq s_{i2}^{j}(0) \quad & \forall i \quad \text{Positive Monotonicity} \\
       & \text{or} \\
        s_{i2}^{j}(1) \leq s_{i2}^{j}(0) \quad & \forall i \quad \text{Negative Monotonicity}         
    \end{align*}
\end{assumption}
Assumption \ref{as:cond_mono} is analogous to Assumption \ref{as:mono} defined for each specific source of selection. It rules out the existence of one stratum for each source $j$ ($V_{i}^{j} = OC$ if positive monotonicity is assumed and $V_{i}^{j} = OT$ if negative monotonicity is assumed). Similarly, I can formulate an analogous, source-specific Assumption \ref{as:sele}:
\begin{assumption} \label{as:sele_ss}
    Source-specific Selection Model \\
    For any source $j\in\{1,...,J\}$:\\
    Under Positive Monotonicity,
    \begin{equation*}
    \begin{aligned}
        s_{it}^{j}(0) = h_j^{0}(u_{it}^{j},t) .
    \end{aligned}
    \end{equation*}
    Under Negative Monotonicity,
    \begin{equation*}
        s_{it}^{j}(1) = h_j^{1}(u_{it}^{j},t),
    \end{equation*}
    where $u_{it}^{j}$ is an unobservable scalar for unit $i$ at time $t$ and $h_j^{w}(u,t)$ is non decreasing function in $u$ $\forall$ $t \in\{1,2\}$, $w = \{0,1\}$. \\
    The unobservable $u_{}^{j}$ is continuously distributed and has the same compact support in both groups. Its distribution is constant over time within groups
    \begin{equation*}
        u_{i1}^{j} \mid G_{i} \sim u_{i2}^{j} \mid G_{i} .
    \end{equation*}
    Additionally, given the realized selection outcome, the distribution of $u^{j}$ is independent of the group in a given time period.
    \begin{equation*}
        u^{j}_{it} \perp \! \! \! \perp G_{i} \mid s_{it}^{j}
    \end{equation*}
\end{assumption}
Assumption \ref{as:sele} states that unobservables affecting selection can be summarized by a single scalar that affects the eventual selection in a monotonic way. Assumption \ref{as:sele_ss} relaxes this assumption by allowing for multiple unobservable scalars that may affect selection in different directions. Going back to the college example, one can think of two distinct unobservables affecting selection: ability and motivation. Higher ability increases the likelihood of graduation, whereas lower motivation increases the likelihood of dropping out. Because students with high ability and high motivation may leave the sample through graduation, while those with low ability and low motivation may leave through dropout, it is difficult to summarize these distinct selection forces using a single index.

Let $\mathcal{X}_{j}$ denote the set of units that belong to stratum $X$ according to source $j$. For instance, $\mathcal{AO}_{j} = \{i : V_{i}^{j} = AO\} \equiv \{i:s_{i2}^{j}(1) = 1, s_{i2}^{j}(0) = 1\}$.
\begin{assumption} \label{as:intersec}
    No intersection of OC and OT strata. \\
    \begin{equation*}
       \mathcal{OT}_{j} \cap \mathcal{OC}_{k} = \varnothing \quad \forall j,k \in \{1,...,J\}
    \end{equation*}
\end{assumption}
Assumption \ref{as:intersec} rules out the scenario where a unit belongs to the OT stratum defined by source $j$ and to the OC stratum defined by a different source. Consequently, units that belong to the OC or OT strata according to source $j$ can only belong to the AO stratum according to all the other sources. Formally, $s_{i2}^{j}(1) \neq s_{i2}^{j}(0) \implies s_{i2}^{k}(1) = s_{i2}^{k}(0) = 1$ $\forall k \neq j$.

For instance, consider a student who drops out in the control group but remains enrolled if treated ($s_{i2}^{1}(0) = 0,s_{i2}^{1}(1) = 1$). Since sources are mutually exclusive and $s_{i2}^{1}(0) = 0$, it must be true that this unit does not graduate in the control group, $s_{i2}^{2}(0) = 1$. Assumption \ref{as:intersec} rules out the possibility that this student, who drops out in the control group, graduates in the treatment group, i.e., assumes that $s_{i2}^{2}(1) = 1$. Hence, given that the student belongs to the OT stratum defined by the joint values of the dropout source, $(s_{i2}^{1}(0) , s_{i2}^{1}(1) ) = (0, 1)$, they must belong to the AO stratum defined by the joint values of the graduation source, $(s_{i2}^{2}(1),s_{i2}^{2}(0)) = (1,1)$.

\begin{remark} \label{prop:remark6}
 If sources are mutually exclusive, and Assumption \ref{as:mono} holds, then Assumptions \ref{as:cond_mono} and \ref{as:intersec} also hold, with all the sources of selection having the same sign of monotonicity. On the other hand, when Assumptions \ref{as:cond_mono} and \ref{as:intersec} hold, then:
 \begin{itemize}
     \item If all sources exhibit the same sign of monotonicity, then Assumption \ref{as:mono} also holds.
     \item If some sources exhibit positive monotonicity while others exhibit negative monotonicity, then Assumption \ref{as:mono} is violated.
 \end{itemize}
\end{remark}

While I acknowledge that Assumptions \ref{as:cond_mono} and \ref{as:intersec} are restrictive, they relax Assumption \ref{as:mono} as noted in Remark \ref{prop:remark6}. This relaxation,  which allows the existence of the four principal strata, comes at no cost, as the proportion of AO units in each group remains point-identified. This section can therefore be viewed as a generalization of Section \ref{sec:hiru.prop}: when Assumptions \ref{as:mono} and \ref{as:sele} hold, then Assumptions \ref{as:cond_mono}, \ref{as:sele_ss}, and \ref{as:intersec} also hold; when Assumptions \ref{as:mono} and \ref{as:sele} fail, then Assumptions \ref{as:cond_mono}, \ref{as:sele_ss}, and \ref{as:intersec} may still hold. A significant limitation, though, is that different sources must exist and be identifiable within the data, which is not always guaranteed. Intuitively, if selection is driven by multiple channels operating in different directions, the assumptions in Section~\ref{sec:hiru.prop} are violated. If these channels manifest themselves in distinct, observable sources of selection, the assumptions in this section may still be satisfied. If not, these assumptions are violated as well.

\begin{lemma} \label{prop:strata_multiple_sources}
     If assumptions \ref{as:cond_mono} and \ref{as:intersec} hold, the four principal strata defined by the joint values of $(S_{i2}(1),S_{i2}(0))$ that partition the population are given by:
     \begin{align*}
         \mathcal{AO} = \bigcap\limits_{j} \mathcal{AO}_{j}  \quad ; \quad
         \mathcal{NO} = \bigcup\limits_{j} \mathcal{NO}_{j}  \quad ; \quad
         \mathcal{OT} = \bigcup\limits_{j} \mathcal{OT}_{j}  \quad ; \quad
         \mathcal{OC} = \bigcup\limits_{j} \mathcal{OC}_{j}
     \end{align*}
     Proof:  See Appendix \ref{proof:strata_multiple_sources}.
\end{lemma}
Lemma \ref{prop:strata_multiple_sources} states that to belong to the AO stratum, a unit must belong to the AO stratum defined by the potential selection indicators of all sources. Conversely, if a unit belongs to the NO, OT, or OC stratum for any source, it will also belong to the corresponding principal stratum defined by the overall selection indicator $S$. Table \ref{tab:example_mult_sourc} provides an example with two different sources of selection.

\begin{proposition} \label{prop:proportions_multiple_sources}
    Let $J^{+}$ denote the set of sources for which positive monotonicity holds and $J^{-}$ the set of sources for which negative monotonicity holds. Under Assumptions \ref{as:no_anti}, \ref{as:abstate}, \ref{as:rand_samp}, \ref{as:cond_mono}, \ref{as:sele_ss}, and \ref{as:intersec}, the proportion of Always-Observed units in the control group ($\pi_{0}$) and the treatment group ($\pi_{1}$) are identified as follows:
        \begin{align*}
            \pi_{0} = \frac{1}{\E[S_{i2} \mid G_{i} = 0]}\left(1 - \sum_{j \in J^{+}}\left(1 - \E[s_{i2}^{j} \mid G_{i} = 0]\right) - \sum_{j \in J^{-}}\left(1 - \E[s_{i1}^{j} \mid G_{i} = 0]\frac{\E[s_{i2}^{j} \mid G_{i} = 1]}{\E[s_{i1}^{j} \mid G_{i} = 1]}\right)\right) \\
             \pi_{1} = \frac{1}{\E[S_{i2} \mid G_{i} = 1]}\left(1 - \sum_{j \in J^{+}}\left(1 - \E[s_{i1}^{j} \mid G_{i} = 1]\frac{\E[s_{i2}^{j} \mid G_{i} = 0]}{\E[s_{i1}^{j} \mid G_{i} = 0]} \right) - \sum_{j \in J^{-}}\left(1 - \E[s_{i2}^{j} \mid G_{i} = 1]\right)\right)
        \end{align*} 
        Proof:  See Appendix \ref{proof:proportions_multiple_sources}.
\end{proposition}
Proposition \ref{prop:proportions_multiple_sources} identifies the proportion of always observed in both the treatment and control groups. By Assumption \ref{as:cond_mono}, some sources exhibit positive monotonicity, while others exhibit negative monotonicity. Following the same logic as in Proposition \ref{prop:proportions}, the former can identify the AO units in the treatment group, and the latter can identify the AO units in the control group. Appendix \ref{proof:equivalence} shows how Proposition \ref{prop:proportions} is a specific case of Proposition \ref{prop:proportions_multiple_sources} with only one source of selection, $J = 1$.

\section{Estimation and inference} \label{sec:lau}
This section proposes estimators for the bounds of the $\attao$ and $\qttao(q)$. I show that these estimators are consistent and asymptotically normal and discuss how to construct confidence intervals for the estimated bounds. The estimators are constructed using the sample analogs of the objects in Propositions \ref{prop:bounds_qttao} and \ref{prop:proportions_multiple_sources}. For simplicity and without loss of generality, this section considers the case with two different sources of sample selection. For the first one, $j =1$, positive monotonicity holds. For the second one, $j = 2$, negative monotonicity holds. All the summations in this section are over the entire sample.

\subsection{Estimation}
Let $N$ denote the sample size, $N_{0}$ the number of units in the control group, and $N_{1}$ the number of units in the treatment group. The estimators of the proportions of Always-Observed units defined in Proposition \ref{prop:proportions_multiple_sources} are given by:
\begin{align}
    \hat \pi_{0} &= \frac{1}{\frac{\sum_{i}S_{i2}(1-G_{i})}{N_{0}}}\left(1 - \left(1 - \frac{\sum_{i}s_{i2}^{1}(1 - G_{i})}{N_{0}}\right) - \left(1 - \frac{\sum_{i}s_{i1}^{2}(1-G_{i})}{N_{0}}\frac{\sum_{i}s_{i2}^{2}G_{i}}{\sum_{i}s_{i1}^{2}G_{i}}\right)\right) \label{eq:pi0}\\
    \hat \pi_{1} & = \frac{1}{\frac{\sum_{i}S_{i2}G_{i}}{N_{1}}}\left(1 - \left(1 - \frac{\sum_{i}s_{i1}^{1}G_{i}}{N_{1}}\frac{\sum_{i}s_{i2}^{1}(1-G_{i})}{\sum_{i}s_{i1}^{1}(1-G_{i})}\right) - \left(1 - \frac{\sum_{i}s_{i2}^{2}G_{i}}{N_{1}}\right)\right) \label{eq:pi1}
\end{align}
where I have substituted the expectations in Proposition \ref{prop:proportions_multiple_sources} with their sample analogs.

Given the estimated proportions, $\hat \pi_{0}$ and $\hat\pi_{1}$, the estimators for the bounds of the $\qttao(q)$ defined in Proposition \ref{prop:bounds_qttao} are given by:
\begin{align} \label{eq:est_qttao_bat}
    \widehat{ \Lambda^{LB}}(q) &= \widehat Q_{Y_{2} \mid G = 1, S_{2} = 1}^{}(q\hat\pi_{1}) -  \widehat  Q_{Y_{2} \mid G = 0, S_{2} = 1}^{}\left(\hat q_{LB}^{*}\right),  \\
    \hat{ q}_{LB}^{*} &= \min\left\{\widehat F_{Y_{1} \mid G = 0, S_{2} = 1}^{}\left(\widehat Q_{Y_{1} \mid G=1, S_{2} = 1}^{}(q\hat\pi_{1}+1 - \hat\pi_{1}) \right) + 1 - \hat\pi_{0} ,1\right\}, \\
      \widehat{\Lambda^{UB}}(q)& = \widehat Q_{Y_{2} \mid G = 1 , S_{2} = 1}^{}(q\hat\pi_{1}+1-\hat\pi_{1}) -\widehat Q_{Y_{2} \mid G = 0, S_{2} = 1}^{}\left(  \hat q_{UB}^{*}\right), \\
      \hat q_{UB}^{*} &= \max\left\{\widehat F_{Y_{1} \mid G = 0, S_{2} = 1}^{}\left(\widehat Q_{Y_{1} \mid G=1, S_{2} = 1}^{}(q\hat \pi_{1}) \right)-(1 - \hat \pi_{0}), 0\right\}, \label{eq:est_qttao_lau}
\end{align}
where $\widehat F_{Y}$ denotes the empirical distribution of $Y$. For instance,
\begin{align*}
    \widehat F_{Y_{1} \mid G = 0, S_{2} = 1}(y) &= \frac{\sum_{i}S_{i2}(1-G_{i})\indicator(Y_{i1}\leq y)}{\sum_{i}S_{i2}(1-G_{i})} \\
    \widehat Q_{Y_{2} \mid G=0,  S_{2} = 1}^{}(q) &= \inf\left\{y: \widehat F_{Y_{2} \mid G=0, S_{2} = 1}(y) \geq q\right\}.
\end{align*}

The bounds for the $\attao$ can be estimated under the CiC framework using the estimated bounds for the $\qttao$, as stated in Remark \ref{prop:remark4}:

\begin{equation}
    \widehat{ATT}_{AO} \in \left[\int_{0}^{1} \widehat{ \Lambda^{LB}}(q) dq \quad , \quad \int_{0}^{1}\widehat{\Lambda^{UB}}(q) dq \right]
\end{equation}

\subsection{Asymptotic Normality} \label{sec:lau.norm}
The next propositions establish the consistency and asymptotic normality of the proposed estimators. These results are indispensable for conducting valid inference on the estimated bounds.

\begin{proposition} \label{prop:normal_qttao}
    Asymptotic normality of the bounds of the $\qttao$
    \begin{align}
        \sqrt{n}(\widehat{\Lambda^{LB}}(q) - {\Lambda^{LB}}(q)) \xrightarrow{d} \mathcal{N}\left(0,\varsigma_{LB}^{2}\right) \\
        \sqrt{n}(\widehat{\Lambda^{UB}}(q) - {\Lambda^{UB}}(q)) \xrightarrow{d} \mathcal{N}\left(0,\varsigma_{UB}^{2}\right) 
    \end{align}
    Provided that $$q \in (0,1),$$ $$F_{Y_{1}\mid G=0, S_2 = 1}\left(Q_{Y_{1}\mid G=1,S_2 = 1}(q\pi_1 + 1 - \pi_1)\right) +1 - \pi_{0} < 1,$$ and $$F_{Y_1 \mid G=0, S_2 = 1}\left(Q_{Y_1\mid G=1, S_2 = 0} (q\pi_1)\right) -( 1 - \pi_0) > 0.$$ \\
    Proof: See Appendix \ref{proof:normal_qttao}
\end{proposition}
Under Positive Monotonicity, $\pi_{0} = 1 $ and the bounds are asymptotically normal for any $q \in (0,1)$. When $\pi_{0} \neq 1$, there might be some extreme quantiles in which the estimators in equations (\ref{eq:est_qttao_bat}-\ref{eq:est_qttao_lau}) are the sample maximum/minimum of $Y_{2} \mid G=0$. In these cases, the estimator for $Q_{Y_{2}(0) \mid G_{i} = 1, S_2 = 1}(q)$ follows a extreme value distribution.
\subsection{Confidence Intervals}
Provided that the estimators of the bounds are asymptotically normal, \textcite{imbens_confidence_2004} provide an expression to construct confidence intervals around a partially identified parameter:
\begin{align} \label{eq:ci}
     CI_{\alpha} = \left[\widehat{\Lambda^{LB}} - Z_{\alpha}  \frac{\hat \sigma_{LB}}{\sqrt{n}} \quad ,\quad  \widehat{ \Lambda^{UB}} + Z_{\alpha}  \frac{\hat\sigma_{UB}}{\sqrt{n}}\right]
\end{align}
 where $Z_{\alpha}$ is such that:
    \begin{equation*}
        \Phi\left(Z_{\alpha} + \sqrt{n}\frac{\widehat{\Lambda^{UB}} - \widehat{\Lambda^{LB}} }{\max\{\hat\sigma_{UB},\hat\sigma_{LB}\}}\right) - \Phi\left(-Z_{\alpha}\right) = \alpha,
    \end{equation*}
where $\Phi(\cdot)$ is the c.d.f. of the standard normal distribution. The interval described in equation (\ref{eq:ci}) contains the set $[\Lambda^{LB}, \Lambda^{UB}]$ at least $\alpha\%$ of the times. Intuitively, when the bounds are far away from each other and/or the standard errors of the bounds are small, the confidence interval is constructed by adding or subtracting the one-side critical value associated with a confidence level $\alpha$ (e.g., 1.645 for a 95\% confidence interval) times the standard error of the bound. This critical value increases when the bounds are close to each other or have large standard errors. In the extreme case where $\widehat{ \Lambda^{LB}} = \widehat{\Lambda^{UB}}$ and the $\attao$ is point identified, $Z_{0.95} \approx 1.96$ and equation (\ref{eq:ci}) becomes the standard confidence interval constructed around a point-identified parameter.

\section{Empirical application} \label{sec:bost}
This section illustrates the methodology using a job training program, \textit{Jóvenes en Acción,} implemented in Colombia between 2002 and 2005. The program provided three months of in-classroom training and three months of on-the-job training to disadvantaged youth to improve their labor market outcomes. The program has been previously studied by \textcite{attanasio_subsidizing_2011}, who found a positive effect on wages and the probability of having paid employment for women. For men, they find a positive effect on job quality but no significant impact on earnings. In a follow-up study, \textcite{attanasio_vocational_2017} examined the program's long-term effects, finding similar and persistent impacts on earnings and job quality for both men and women. More recently, \textcite{possebom_probability_2024} revisited these results using principal stratification analysis. They partially identify the effect of the training on the probability of being employed in the formal sector for the women who are employed, regardless of their treatment status. For this stratum, they cannot reject the null that the program had no impact on formal employment.

In this paper, I use data from \textcite{attanasio_subsidizing_2011}, which consists of survey responses from 3,955 individuals. The participants were interviewed twice: initially in January 2005, before the training, and then between August and October 2006, after the training. Using the baseline data, I construct a panel dataset. 
Further details on data collection and program implementation are available in \textcite{attanasio_subsidizing_2011}.

I analyze the effect of the training on salaried wages. This outcome is observed only among those with paid employment, as unemployed and self-employed workers do not receive salaried wages. Of the 3,955 individuals in the sample, 1453 have paid employment in the pre-treatment period. Among these, 20\% were no longer employed in salaried positions during the post-treatment period, and 19\% could not be reached in the post-treatment survey. As a result, I observe salaried wages in both periods for 888 individuals. Table \ref{tab:samplesize} presents these numbers separately for each treatment group.

I can define two different sources of sample selection in my sample: transitioning to unemployment/self-employment and survey non-response. For those with missing salaried earnings, I observe the specific source of selection in the data. I assume negative monotonicity for unemployment/self-employment and positive monotonicity for survey non-response.\footnote{
The negative monotonicity assumption implies that all workers who did not transition into unemployment or self-employment after receiving training would also have retained their salaried positions in the absence of treatment. In the context of disadvantaged youth in Colombia, many salaried jobs are informal. As a result, training may reduce the probability of holding a salaried position by increasing treated workers’ reservation wages, enhancing their entrepreneurial skills and likelihood of self-employment, or raising aspirations that lead to longer job-search durations in the short run following program completion. By contrast, the positive monotonicity assumption implies that all control units that were reachable in the post-treatment period would also have responded to the survey had they been treated. This assumption is plausible if program participation increases respondents’ commitment to survey participation and if attrition is correlated with migration, which the treatment may reduce.
} 
Under these assumptions, I estimate the proportion of units that have paid employment in both treatment arms among my final sample of 888 workers. That is, I estimate the proportions of Always-Observed units, $\pi_{1}$ and $\pi_{0}$ as defined in Proposition \ref{prop:bounds_attao}, using the estimators described in equations (\ref{eq:pi0}) and (\ref{eq:pi1}). I estimate $\hat \pi_{1} = 0.93$ and $\hat \pi_{0} = 0.96$.  In other words, of the 501 units observed in both periods in the treatment group and the 387 units in the control group, 93\% of the treated units and 96\% of the control group units would have also been observed in the opposite group.

Using these estimated proportions, I partially identify the $\attao$. Table \ref{tab:attao_cic} reports the estimates of the bounds with and without the inclusion of covariates.\footnote{The covariate adjustment is described in Appendix \ref{ap:ext_covariates}.} For comparison, columns 3 and 4 present results from a complete-case analysis, which discards observations with missing outcomes and applies the Changes-in-Changes estimator to the subsample of units with observed outcomes.
\begin{table}[H]
    \caption{Estimates of the bounds for the $\attao$}
    \label{tab:attao_cic}
\vspace{-0.1cm} 
\begin{center}
\begin{tabular}{r*{8}{c}}
\toprule

Outcome & \multicolumn{8}{c}{Log of salaried earnings} \\ 

\cmidrule(lr){2-9} 

 Estimand & \multicolumn{4}{c}{$\attao$ } & \multicolumn{4}{c}{Complete Case} \\
 
\cmidrule(lr){2-5}  \cmidrule(lr){6-9} 

Estimate &  [-0.11  ,&  \hspace{-0.3cm} 0.429] & [-0.095 ,& \hspace{-0.3cm} 0.319] & \multicolumn{2}{c}{0.129} & \multicolumn{2}{c}{0.123} \\ 

95\% CI & \footnotesize(-0.169  ,& \footnotesize \hspace{-0.3cm} 0.527) &\footnotesize (-0.158 ,& \footnotesize\hspace{-0.3cm} 0.416) & \footnotesize(0.037 ,&\footnotesize  \hspace{-0.3cm} 0.22) &\footnotesize (0.031 ,& \footnotesize\hspace{-0.3cm} 0.215)  \\

Covariates & \multicolumn{2}{c}{No} &\multicolumn{2}{c}{Yes} & \multicolumn{2}{c}{No} & \multicolumn{2}{c}{Yes} \\

N & \multicolumn{2}{c}{888}&\multicolumn{2}{c}{888}&\multicolumn{2}{c}{888}&\multicolumn{2}{c}{888} \\

\bottomrule
\end{tabular}
 
\end{center}
\vspace{-0.2cm}
\begin{minipage}{1\linewidth \setstretch{0.75} } 
{\scriptsize Notes: Columns 1 and 2 present the estimates of the bounds for the $\attao$ using the estimated proportions $\hat \pi_{1} = 0.93$ and $\hat \pi_{0} = 0.96$. Column 1 reports the bounds estimated without covariates with the estimators presented in Section \ref{sec:lau}, while Column 2 includes education as a covariate as explained in Appendix \ref{ap:ext_covariates}. Columns 3 and 4 show estimates using the complete case analysis without and with covariates respectively. The 95\% confidence intervals are computed using equation (\ref{eq:ci}).
}
 \end{minipage}
\end{table}

Table \ref{tab:attao_cic} illustrates the importance of the methodology developed in this paper. Disregarding the problem of sample selection and estimating the ATT using units with observed outcomes suggests that the training had a positive and significant effect on salaried earnings, increasing them by 12\% for treated units. However, this increase may be driven by differences in the distribution of outcomes across strata. Instead, I estimate the intensive-margin effect of the training by targeting treated units with paid employment regardless of their treatment status. The bounds for this effect include 0. Accordingly, I cannot reject the null hypothesis that the program had no impact on wages at any significance level. 

In this context, it is relevant to examine the distributional effects. Suppose the treatment effects are heterogeneous along the outcome distribution. In that case, the bounds for the average effect may include zero even if the effect is positive and significant in some parts of the distribution. For instance, consider two different training programs. The first only benefits workers at the top of the salaried earnings distribution. The other only benefits workers at the bottom. These two policies may generate the same bounds for the $\attao$. However, their policy implications change drastically: while the first program increases wage dispersion, the second helps reduce earnings inequality. Figure \ref{fig:qttao} plots the $\qttao(q)$ for different quantiles $q$.
\begin{figure}[H]
    \caption{$\qttao(q)$. Outcome: log of salaried earnings}
    \label{fig:qttao}
    \centering
\includegraphics[width = \linewidth]{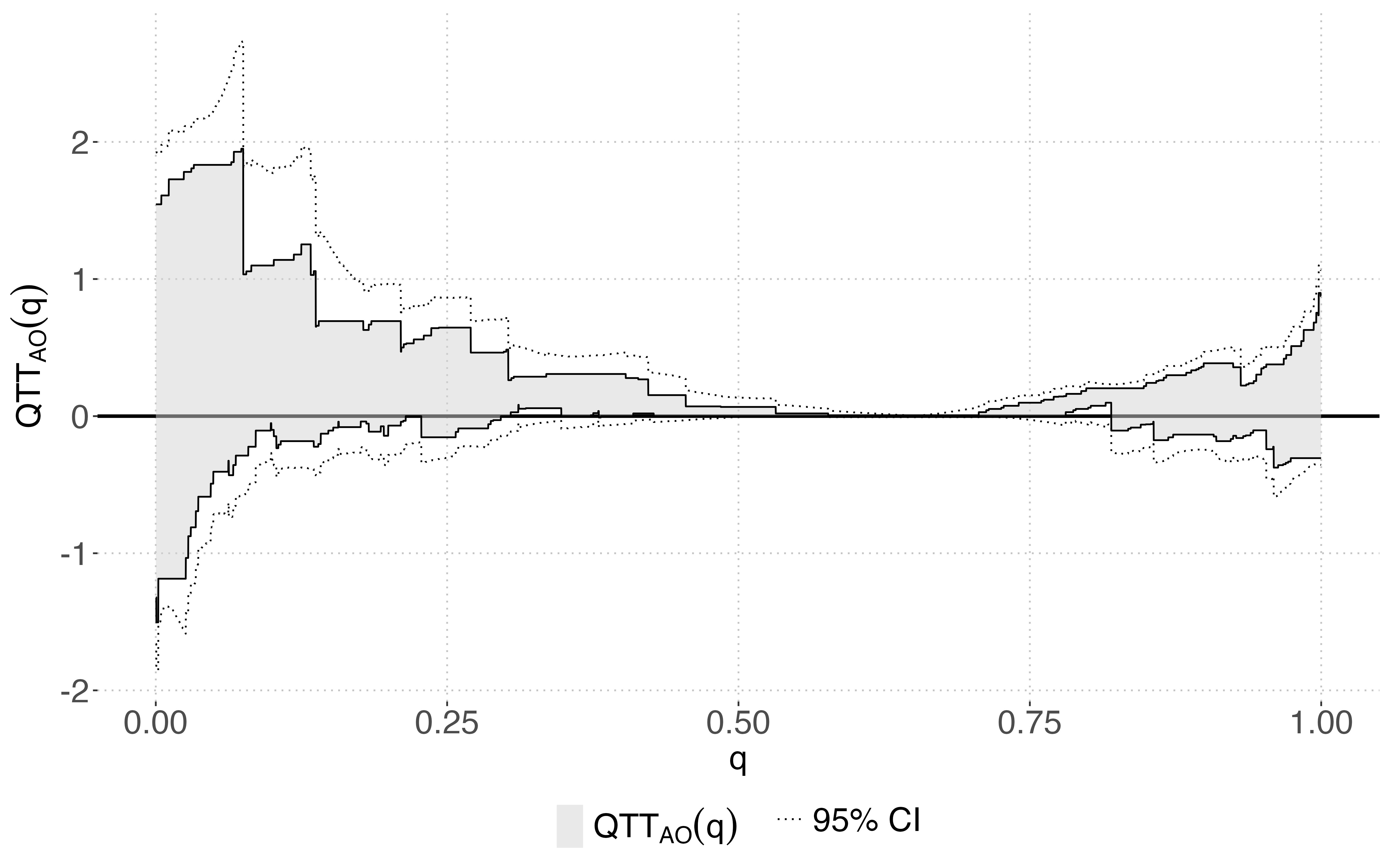}    
\vspace{-0.2cm}
\begin{minipage}{0.99\linewidth  \setstretch{0.75}}
{\scriptsize Notes: Quantile Treatment effect on the Treated Always Observed units ($\qttao(q)$) for different values of $q$. The number of units used to compute these bounds is 888.}
 \end{minipage}
\end{figure}

Figure \ref{fig:qttao} shows that the $\qttao(q)$ bounds vary across different quantiles, reflecting heterogeneity in the treatment effects. The 95\% confidence interval includes zero for all the quantiles. However, their widths and ranges differ substantially across the outcome distribution. The bounds are vast and not very informative at the lower quantiles. However, as $q$ approaches the median, the bounds tighten and the lower bound approaches or even exceeds zero. A precise zero effect is estimated between the median and the third quartile. Finally, a moderate but positive effect is estimated around the eighth decile, before the bounds widen again towards the upper tail of the distribution.
\section{Conclusion and extensions} \label{sec:sei}
This paper develops a novel identification strategy for estimating the intensive margin effect in panel data settings. I propose two new estimands, the $\attao$ and the $\qttao(q)$, which can be interpreted as the average and the quantile intensive-margin effects, respectively. I explore partial identification of these causal quantities within a CiC framework, which allows the treatment to be confounded with unobservable unit-level characteristics. 
Appendix \ref{ap:extensions} extends the methodology presented in the main text in different directions. Appendix \ref{ap:ext_did} replaces the Changes-in-Changes outcome model with the canonical Difference-in-Differences framework. Appendix \ref{ap:ext_as2} revisits the methodology without Assumption \ref{as:abstate}, which treats missingness as an absorbing state. Appendix \ref{ap:ext_covariates} studies how to incorporate covariates into the identification and estimation strategies. Appendix \ref{ap:ext_rcs} extends the identification to repeated cross-sectional data. Appendix \ref{ap:ext_discrete} explores the identification when the outcome is discrete. Yet several extensions remain for future research.

The intensive margin effects are instrumental for policy evaluation. Beyond estimating the intensive margin effect, I encourage practitioners to examine the extensive margin effect by characterizing all the principal strata, including those in the Observed-if-Control and Observed-if-Treated strata. Furthermore, the two estimands proposed in this paper focus on the treatment group. Since the treatment is not unconfounded, its effect will likely differ across the treatment and control groups. However, the CiC framework can be extended to impute the missing potential outcome for control units, thereby enabling the identification of the average treatment effect for the control group. As a result, the proposed methodology can accommodate the identification of two additional estimands, the $\text{ATE}_{\text{AO}}$ and the $\text{QTE}_{\text{AO}}(q)$, enhancing the external validity of the conclusions drawn from the analysis.

This paper focuses on the two-group, two-period case. However, many applications involve multiple pre- and post-treatment periods. Incorporating this information is essential for both assessing the identification assumptions and gaining a deeper understanding of the policy's impacts. Moreover, many policies are implemented in a staggered fashion. While the proposed methodology can be smoothly extended to the multiple groups, multiple periods setting through the estimation and aggregation of pairwise comparisons \parencite{athey_identification_2006,sun_estimating_2021}, a more ambitious and valuable direction for research would be to integrate it with doubly robust estimators that combine outcome modeling with propensity score weighting \parencite{arkhangelsky_doubly_2022, santanna_doubly_2020}. Another promising direction for future research is the extension to time-varying estimands, as proposed by \textcite{comment_survivor_2025,lin_longitudinal_2008}, which may be particularly policy-relevant in dynamic settings.

One limitation of the proposed methodology is that the derived bounds may be uninformative. In general, lower proportions of Always-Observed units lead to wider bounds, and even when this proportion is close to one, certain outcome distributions may still generate wide identification regions. Adopting a Bayesian perspective may help infer where the target estimand is most likely to lie within the bounds. Another approach is to tighten the bounds by incorporating covariates, as discussed in Appendix~\ref{ap:ext_covariates}. Under treatment unconfoundedness, covariate-adjusted bounds are weakly tighter than unadjusted bounds \parencite{grilli_nonparametric_2008, mealli_using_2013, lee_training_2009, long_sharpening_2013}. However, in the setting studied here, where treatment is confounded and covariate distributions can differ substantially across groups, there is no guarantee that covariate adjustment will tighten the bounds. Determining whether, how, and which covariates should be included in panel data settings remains an important topic for future research.

Finally, this paper relaxes the monotonicity assumptions while maintaining the point identification of principal strata proportions. This extension may also be of interest in settings under unconfoundedness. Nevertheless, it relies on multiple sources of sample selection. Furthermore, it maintains the monotonicity assumption for each available source.  If the researcher does not observe multiple sources of selection, believes that these multiple sources do not capture all the unobservables that may affect selection in different directions, and/or does not want to rely on monotonicity of any form, \textcite{shin_difference--differences_2024,rathnayake_difference--differences_2024} provide identification of the $\attao$ without monotonicity for the DiD case. However, they replace the monotonicity assumption with the assumption that the average treatment effect on selection is the same in the treatment and control groups\footnote{In the case of \textcite{rathnayake_difference--differences_2024}, conditional on the pre-treatment selection outcome.}, as in settings where treatment is unconfounded. An interesting approach to address the monotonicity assumption is the automated method proposed by \textcite{duarte_automated_2024}, which can be used to assess the plausibility of the assumption and derive sharp bounds that relax or even remove it. 

\end{spacing}

\begin{spacing}{1}
{ \small
\nocite{*}
\newrefcontext[sorting=nyt]
\printbibliography
}
\end{spacing}

\newpage

\appendix

\renewcommand{\theequation}{A\arabic{equation}}
\renewcommand{\thetable}{A\arabic{table}}
\renewcommand{\thefigure}{A\arabic{figure}}
\renewcommand{\thelemma}{A\arabic{lemma}}
\renewcommand{\theproposition}{A\arabic{proposition}}
\renewcommand{\thecorollary}{A\arabic{corollary}}
\setcounter{equation}{0}
\setcounter{table}{0}
\setcounter{figure}{0}
\setcounter{lemma}{0}
\linespread{1}

\newpage
\vspace{-0.4 cm}

\begin{spacing}{1.5}

\section{Extensions and additional results} \label{ap:extensions}
\subsection{Identification in the DiD setup} \label{ap:ext_did}
\begin{assumption} \label{as:ppt} 
    Principal Parallel Trends (PPT)
    \begin{equation*}
        \E[Y_{i2}(0) - Y_{i1} \mid G_{i} = 1, V_{i} = AO] = \E[Y_{i2}(0) - Y_{i1} \mid G_{i} = 0, V_{i} = AO]
    \end{equation*}
\end{assumption}
Assumption \ref{as:ppt} states that, in the absence of treatment, the difference in the outcome between groups remains constant over time for the AO units. It is a conditional version of the canonical parallel trends assumption. Like other conditional parallel trends assumptions, it is neither a sufficient nor a necessary condition for the unconditional parallel trend assumption. Under Assumptions \ref{as:no_anti}, \ref{as:abstate} and \ref{as:ppt}, the missing potential outcome for treated Always-Observed units is given by:
\begin{equation} \label{eq:missing_po_ao}
     \E[Y_{i2}(0)\mid G_{i} = 1, V_{i} = AO] =\E[ Y_{i1} \mid G_{i} = 1, V_{i} = AO] + \E[Y_{i2}(0) - Y_{i1} \mid G_{i} = 0, V_{i} = AO]
\end{equation}
Then, I could plug in the term in equation (\ref{eq:missing_po_ao}) into the Estimand \ref{es:attao} to obtain:
\begin{align} 
    \text{ATT}_{\text{AO}} &=  \E[Y_{i2}(1) - Y_{i1} \mid G_{i} = 1, V_{i} = AO] - \E[Y_{i2}(0) - Y_{i1} \mid G_{i} = 0, V_{i} = AO] \notag \\
&= \E[Y_{i2} - Y_{i1} \mid G_{i} = 1, V_{i} = AO] - \E[Y_{i2} - Y_{i1} \mid G_{i} = 0, V_{i} = AO].   \label{eq:princ_did}
\end{align}

\begin{proposition} \label{prop:bounds_attao}
    Let $Y_{it}(1)$ and $Y_{it}(0)$ be continuous. If Assumptions \ref{as:no_anti}, \ref{as:abstate}, \ref{as:rand_samp} and \ref{as:ppt} hold, then $\Delta^{LB}$ and $\Delta^{UB}$ are lower and upper bounds for the Average Treatment Effect on the Treated Always-Observed units ($\attao$), where
    \begin{align*}
       \Delta^{LB} & =  \E[\ddot Y_{i} \mid G_{i} = 1, S_{i2} = 1, \ddot Y_{i} \leq \ddot y_{\pi_{1}}^{1} ] - \E[\ddot Y_{i} \mid G_{i} = 0, S_{i2} = 1, \ddot Y_{i} \geq \ddot y_{1 - \pi_{0}}^{0} ] \\
       \Delta^{UB} & =  \E[\ddot Y_{i} \mid G_{i} = 1, S_{i2} = 1, \ddot Y_{i} \geq \ddot y_{1 -\pi_{1}}^{1} ] - \E[\ddot Y_{i} \mid G_{i} = 0, S_{i2} = 1, \ddot Y_{i} \leq \ddot y_{\pi_{0}}^{0} ] \\
       \ddot Y_{i} & = Y_{i2} - Y_{i1} \\
       \ddot y_{q}^{g} &= \inf \{\ddot y: F(\ddot y) \geq q\}, \text{ with }F \text{ the c.d.f. of }\ddot Y\text{ conditional on } S_{2} = 1 \text{ and } G = g \\
       \pi_{1} & = Pr(S_{i2}(0) = 1 \mid G_{i} = 1, S_{i2}(1) = 1)\\
       \pi_{0} & = Pr(S_{i2}(1) = 1 \mid G_{i} = 0, S_{i2}(0) = 1).
    \end{align*}
    Proof: See Appendix \ref{proof:bounds_attao}.
\end{proposition}

Proposition \ref{prop:bounds_attao} describes a trimming procedure to bound the $\attao$, as proposed by \textcite{lee_training_2009}. If the proportion of Always-Observed units in each group is known, one can use them to trim the observed distribution of $\ddot Y = Y_{2} - Y_{1}$ and partially identify $\E[\ddot Y_{i} \mid G_{i} = g, V_{i} = AO]$. For instance, suppose 90\% of the units in the treatment group belong to the AO stratum. In this case, the expected value of $\ddot Y$ using only the bottom 90\% of the distribution of $\ddot Y$ provides a lower bound of the expected value of $\ddot Y$ for these units. Similarly, using the top 90\% of the distribution provides an upper bound for the expected value of $\ddot Y$ for the AO units in the treatment group.

Figure \ref{fig:ex_dist} provides graphical intuition for Proposition \ref{prop:bounds_attao}. The distributions on the left correspond to the treatment group and are trimmed using $\pi_{1}$. The ones on the right correspond to the control group and are trimmed using $\pi_{0}$. The lower bound is computed using the bottom $\pi_{1}$ units from the treatment group and the top $\pi_{0}$ units from the control group. Likewise, the upper bound is obtained using the top $\pi_{1}$ units from the treatment group and the bottom $\pi_{0}$ units from the control group. 

\begin{figure}[h] 
    \caption{Graphical intuition on Proposition \ref{prop:bounds_attao}.}
    \centering
        \vskip\baselineskip
     \raisebox{0.1cm}{$\Delta^{LB}$}
    \fbox{
        \begin{minipage}{0.98\textwidth}
            \centering
            \begin{minipage}{0.48\textwidth}
                \centering
                \includegraphics[width=\textwidth]{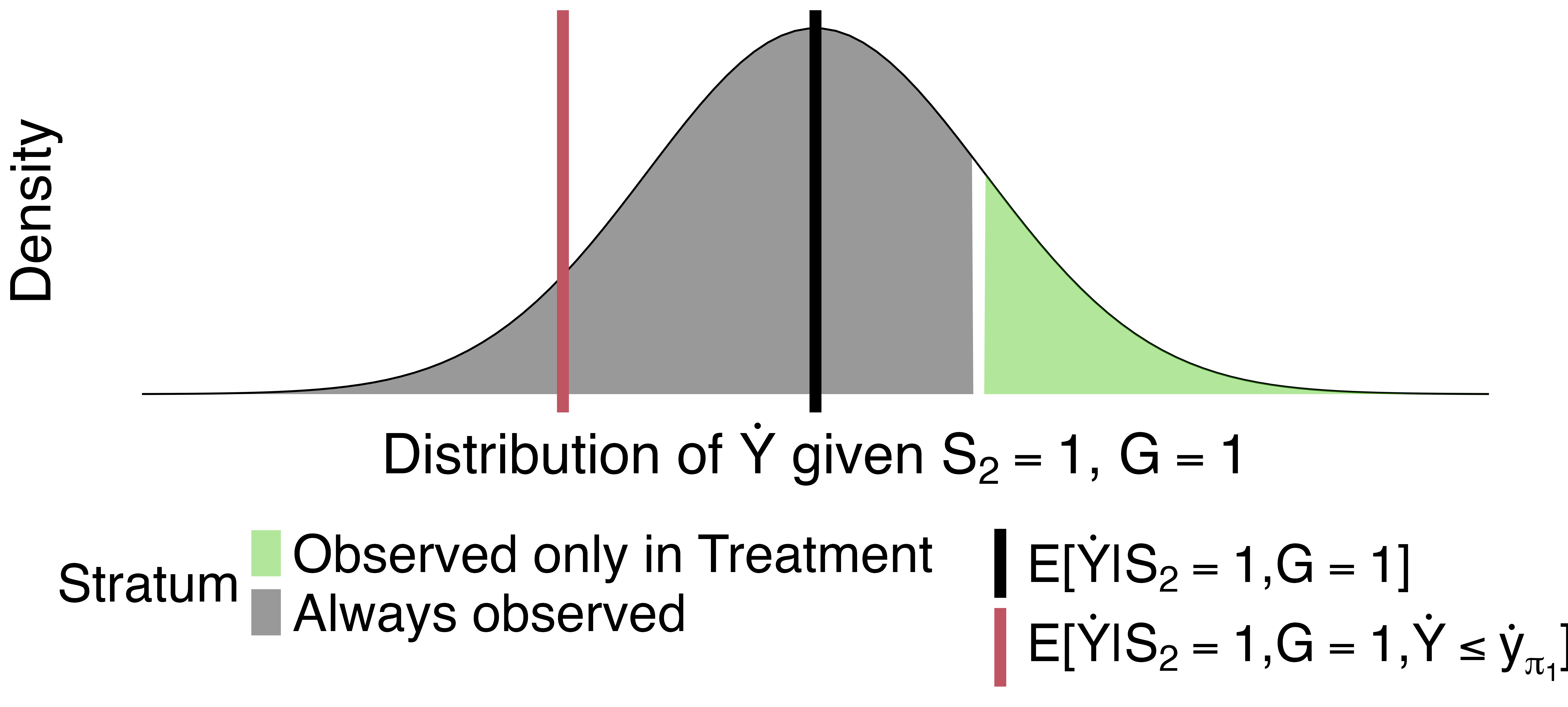}
            \end{minipage}%
            \hfill
            \begin{minipage}{0.48\textwidth}
                \centering
                \includegraphics[width=\textwidth]{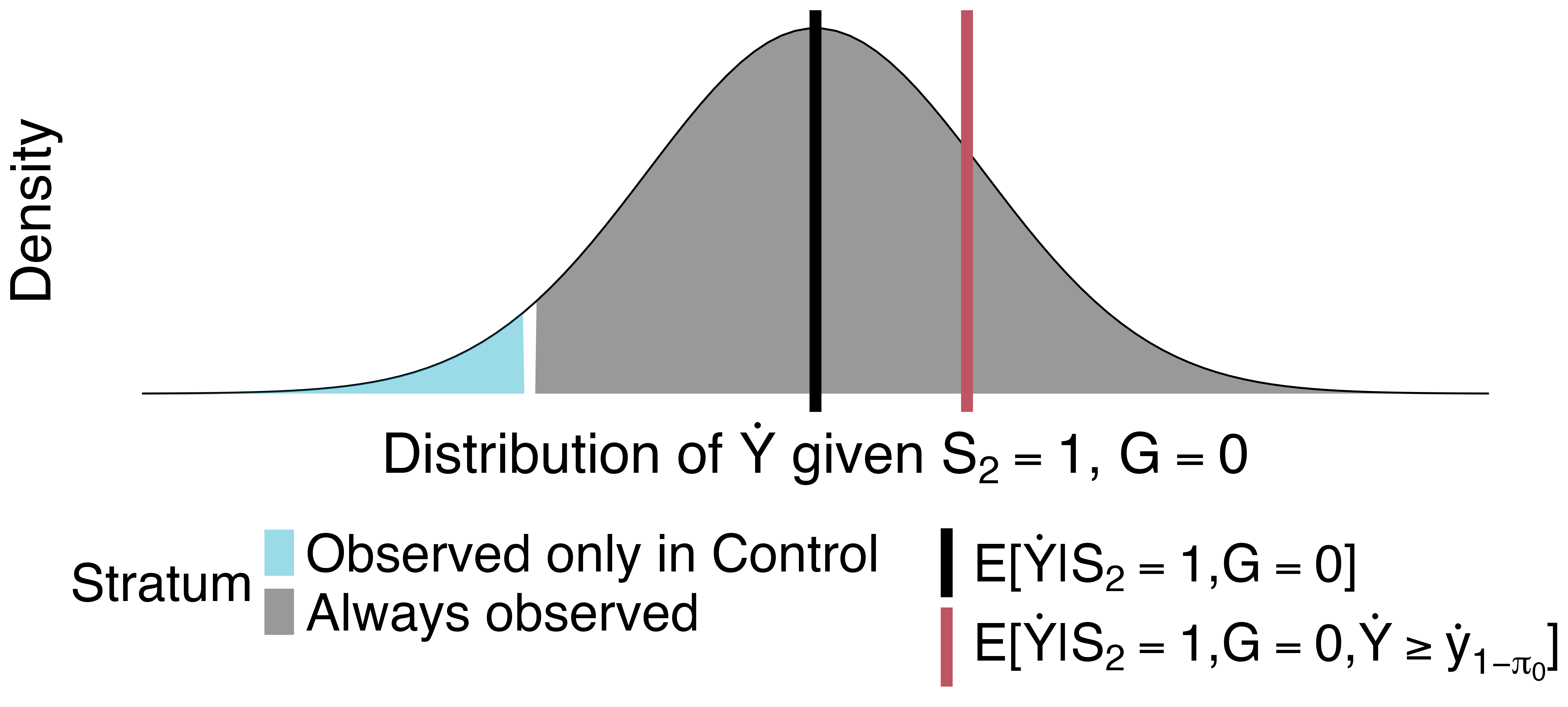}
            \end{minipage}
        \end{minipage}  
    }

    \vskip\baselineskip

     \raisebox{0.1cm}{$\Delta^{UB}$}  
    \fbox{
        \begin{minipage}{0.98\textwidth}
            \centering
            \begin{minipage}{0.48\textwidth}
                \centering
                \includegraphics[width=\textwidth]{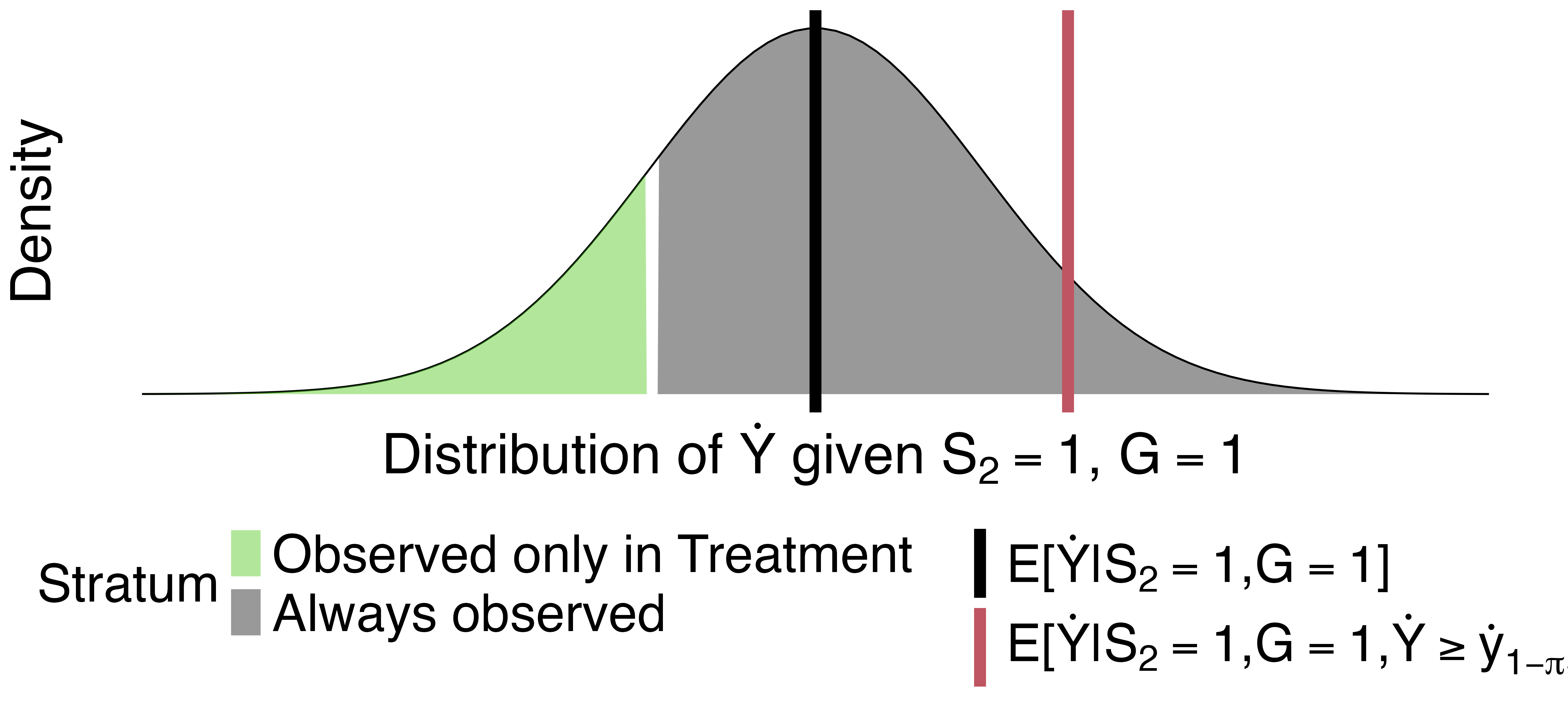}
            \end{minipage}%
            \hfill
            \begin{minipage}{0.48\textwidth}
                \centering
                \includegraphics[width=\textwidth]{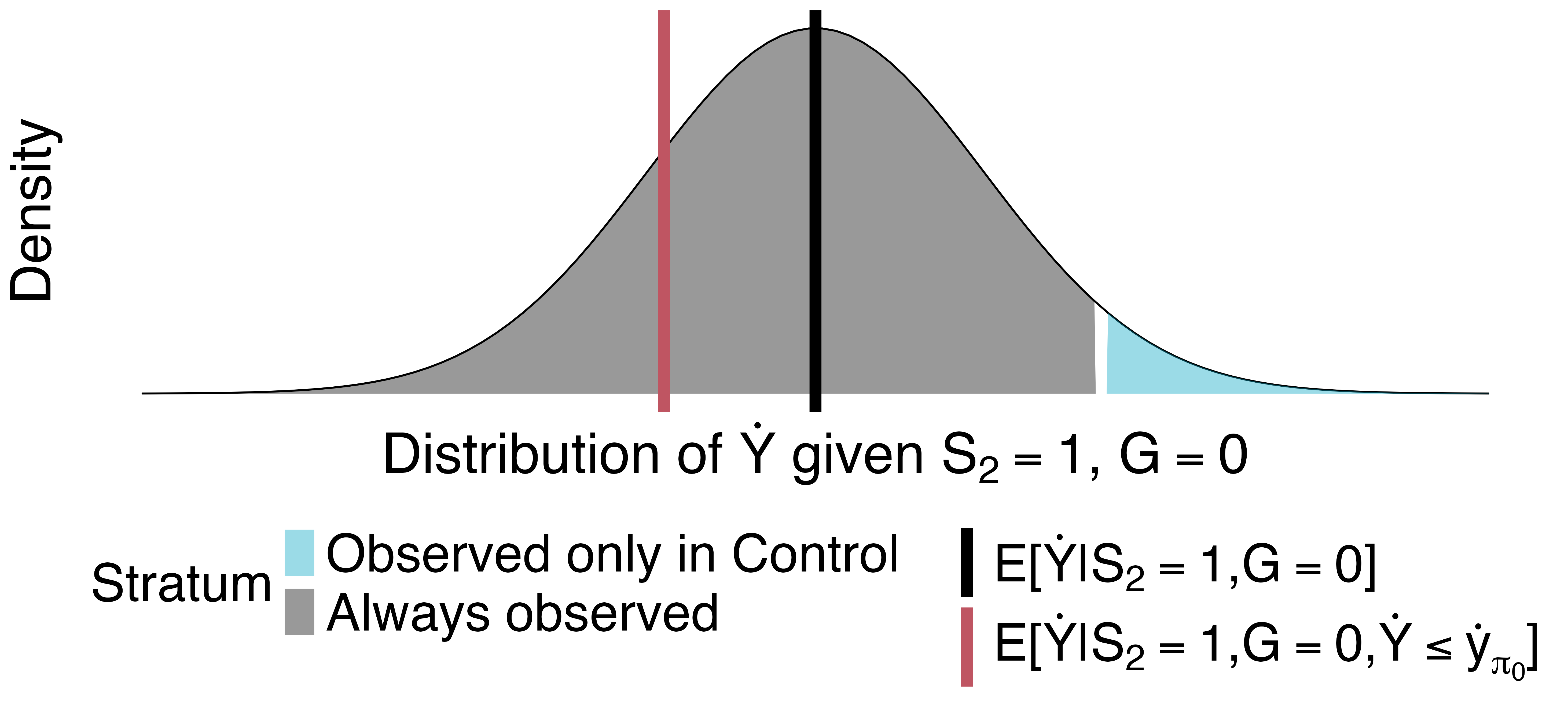}
            \end{minipage}
        \end{minipage}  
    }
\vspace{-0.2cm}
\begin{minipage}{0.99\linewidth  \setstretch{0.75}}
{\scriptsize Notes: This figure plots a hypothetical distribution of $\ddot Y$. The distributions on the left correspond to the Treated group ($G=1$) and the one on the right to the Control group ($G=0$). The shaded areas correspond to the proportion of strata in the given group: black corresponds to the AO units ($\pi_{1}$ on the left and $\pi_{0}$ on the right),  green corresponds to the OT stratum and blue to the OC. The vertical black lines are the expectation of the plotted distribution. The vertical red lines are the trimmed distributions used in Proposition \ref{prop:bounds_attao}. This figure also illustrates Remark \ref{prop:remark3}: as $\pi_{1}$ and $\pi_{0}$ go to 1, the black shadowed areas expand, and the trimmed (red) expectations converge to those of the entire distribution (black).}
 \end{minipage}
 \label{fig:ex_dist}
\end{figure}

\subsubsection{Estimation and inference in the DiD setup} \label{ap:ext_did_est}
The estimators for the bounds of the $\attao$ defined in Proposition \ref{prop:bounds_attao} are given by:
\begin{align}
    \hat{\ddot y}_{q}^{1} & = \min\left\{\ddot y: \frac{\sum_{i}S_{i2}G_{i}\indicator(\ddot Y_{i} \leq \ddot y)}{\sum_{i}S_{i2}G_{i}} \geq q \right\} \label{est:trimt}\\
    \hat{\ddot y}_{q}^{0} & = \min\left\{\ddot y: \frac{\sum_{i}S_{i2}(1-G_{i})\indicator(\ddot Y_{i} \leq \ddot y)}{\sum_{i}S_{i2}(1-G_{i})} \geq q \right\}. \label{est:trimc} \\
    \widehat{ \Delta^{LB}} &= \frac{\sum_{i} S_{i2}G_{i} \indicator(\ddot Y_{i} \leq \hat{ \ddot y}_{\hat \pi_{1}}^{1})\ddot Y_{i} }{\sum_{i}S_{i2}G_{i}\indicator(\ddot Y_{i} \leq \hat{ \ddot y}_{\hat \pi_{1}}^{1})} - \frac{\sum_{i}S_{i2}(1-G_{i})\indicator(\ddot Y_{i} \geq \hat{\ddot y}_{1 - \hat \pi_{0}}^{0})\ddot Y_{i}}{\sum_{i}S_{i2}(1-G_{i})\indicator(\ddot Y_{i} \geq \hat{\ddot y}_{1 - \hat \pi_{0}}^{0})} \label{est:lb}\\
    \widehat{\Delta^{UB}} &= \frac{\sum_{i} S_{i2}G_{i} \indicator(\ddot Y_{i} \geq \hat{ \ddot y}_{1 - \hat \pi_{1}}^{1})\ddot Y_{i} }{\sum_{i}S_{i2}G_{i}\indicator(\ddot Y_{i} \geq \hat{ \ddot y}_{1 - \hat \pi_{1}}^{1})} - \frac{\sum_{i}S_{i2}(1-G_{i})\indicator(\ddot Y_{i} \leq \hat{\ddot y}_{\hat \pi_{0}}^{0})\ddot Y_{i}}{\sum_{i}S_{i2}(1-G_{i})\indicator(\ddot Y_{i} \leq \hat{\ddot y}_{ \hat \pi_{0}}^{0})} \label{est:ub}
\end{align}
Equations (\ref{est:trimt}) and (\ref{est:trimc}) estimate the trimming threshold for a given trimming proportion. Equations (\ref{est:lb}) and (\ref{est:ub}) estimate the expectations with their sample analogs, taking the trimming threshold and proportions as given. 

\begin{proposition} \label{prop:normal_attao}
    Asymptotic normality of the bounds of the $\attao$
    \begin{align}
        \sqrt{n}(\widehat{\Delta^{LB}} - {\Delta^{LB}}) \xrightarrow{d} \mathcal{N}\left(0,\sigma_{LB}^{2}\right) \\
        \sqrt{n}(\widehat{\Delta^{UB}} - {\Delta^{UB}}) \xrightarrow{d} \mathcal{N}\left(0,\sigma_{UB}^{2}\right) ,
    \end{align}
    Proof: See Appendix \ref{proof:normal_attao}
\end{proposition}

\subsubsection{Empirical Application}
\begin{table}[H]
    \caption{Estimates of the bounds for the $\attao$}
    \label{tab:attao_did}
\vspace{-0.1cm} 
\begin{center}
\begin{tabular}{r*{4}{c}}
\toprule

Outcome & \multicolumn{4}{c}{Log of salaried earnings} \\ 

\cmidrule(lr){2-5} 

 Estimand & \multicolumn{2}{c}{$\attao$ } & \multicolumn{2}{c}{Complete Case} \\
 
\cmidrule(lr){2-3}  \cmidrule(lr){4-5} 

Estimate &  [-0.013  ,&  \hspace{-0.3cm} 0.322] &  \multicolumn{2}{c}{0.158} \\ 

95\% CI & \footnotesize(-0.082  ,& \footnotesize \hspace{-0.3cm} 0.39) &\footnotesize (0.075 ,& \footnotesize\hspace{-0.3cm} 0.241)  \\

Covariates & \multicolumn{2}{c}{No} & \multicolumn{2}{c}{No}  \\

N & \multicolumn{2}{c}{888}&\multicolumn{2}{c}{888} \\

\bottomrule
\end{tabular}
 
\end{center}
\vspace{-0.2cm}
\begin{minipage}{1\linewidth \setstretch{0.75} } 
{\scriptsize Notes: Columns 1 presents the estimates of the bounds for the $\attao$ using the estimated proportions $\hat \pi_{1} = 0.93$ and $\hat \pi_{0} = 0.96$, with the estimators presented in Section \ref{ap:ext_did_est}, while Column 2 shows estimates using the complete case analysis. The 95\% confidence intervals are computed using equation (\ref{eq:ci}).
}
 \end{minipage}
\end{table}
\subsection{Identification without Assumption \ref{as:abstate}, Missingness as an absorbing state.} \label{ap:ext_as2}

Assumption \ref{as:abstate} states that once a unit’s outcome is unobserved, it remains unobserved in
all subsequent periods. There are settings in which this assumption is not plausible. For instance, a job-training program in which a worker is unemployed during the pre-treatment period but later finds a job. In this case, there would be some $i$ in the data such that $S_{i1} = 0$ but $S_{i2} = 1$, violating Assumption \ref{as:abstate}. Next, I discuss the identification for settings in which missigness is not an absorbing state.

Consider the Estimand \ref{es:attao}, $\attao$. Using the Law of Total Probability, I can write it as:
\begin{equation}
\begin{aligned}
    \attao =& \E[Y_{i2}(1) - Y_{i2}(0) \mid G_{i} = 1, V_{i} = AO, S_{i1}= 1] Pr(S_{i1} = 1 \mid G_{i} = 1, V_{i} = AO) + \\
    &\E[Y_{i2}(1) - Y_{i2}(0) \mid G_{i} = 1, V_{i} = AO, S_{i1}= 0] Pr(S_{i1} = 0 \mid G_{i} = 1, V_{i} = AO)
\end{aligned}
\end{equation}
Under Assumption \ref{as:abstate},  $$V_{i} = AO \implies S_{i2} = 1 \implies S_{i1} \neq0$$
and therefore,
$Pr(S_{i1} = 0 \mid G_{i} = 1, V_{i} = AO) = 0$. Hence, under Assumption \ref{as:abstate}, the following equality is true:
\begin{equation}
    \attao = \E[Y_{i2}(1) - Y_{i2}(0) \mid G_{i} = 1, V_{i} = AO, S_{i1}= 1]. \label{eq:attao_observed_t1}
\end{equation}

Sections \ref{sec:hiru} and \ref{sec:lau} discuss identification, estimation, and inference for the RHS of equation \eqref{eq:attao_observed_t1}. The conditioning on $S_{1} = 1$ is not being made explicit as under Assumption \ref{as:abstate} it is implied by the conditioning on $V_{i} = AO$.

When missigness is not an absorbing state, some units may have observed outcomes only in the post-treatment period, i.e., $S_{i1} = 0$ and $S_{i2} = 1$. Such units may belong to the Always-Observed stratum. In that case, $Pr(S_{i1} = 0 \mid V_{i} = AO)$ is no longer 0 and the equality in equation \eqref{eq:attao_observed_t1} is no longer true.\footnote{Unless one is willing to assume that pre-treatment selection is ignorable, that is, $Y_{i2}(w_{i2}) \mid G_{i } = 1, V_i = AO \ind S_{i1} \quad \quad \forall \quad w_{i2}  $.} Importantly, without Assumption \ref{as:abstate}, the identification strategy developed in Sections \ref{sec:hiru} and \ref{sec:lau} still identifies a causal effect, since $S_{1}$ is not affected by the treatment.\footnote{\label{fn:ap1}Albeit the conditioning on $S_{i1} = 1$ must be made explicit in all the assumptions, estimands, and identification results.} However, the resulting estimand is no longer the $\attao$ but rather the $\attao$ for units observed in the first period, as in the RHS of equation \eqref{eq:attao_observed_t1}.

In some contexts, researchers may want to estimate the causal effect for all the treated units in the AO stratum, regardless of whether they were observed in the pre-treatment period. Consider the case of the $\attao$:
\begin{align}
    \attao &= \E[Y_{i2}(1) - Y_{i2}(0) \mid G_{i} = 1, V_{i} = AO] \\
    &=\E[Y_{i2}(1) \mid G_{i} = 1, V_{i} = AO]-\E[Y_{i2}(0) \mid G_{i} = 1, V_{i} = AO] \label{eq:full_attao}
\end{align}
The first term in equation \eqref{eq:full_attao} can be partially identified using the identification results from Section \ref{sec:hiru} using all the units observed at $t = 2$ instead of conditioning on those with $S_{i1} = 1$. However, because now $\E[S_{i2}]$ may be larger than $\E[S_{i1}]$, Proposition \ref{prop:proportions} needs to be more general:
\begin{proposition} \label{prop:proportions_general}
    Under Assumptions \ref{as:no_anti}, \ref{as:rand_samp}, \ref{as:mono}, and \ref{as:sele}, the proportion of Always-Observed units in the treatment group ($ \pi_{1}$) and control group ($\pi_{0}$) are identified as follows:
    \begin{itemize}
        \item[] If Positive Monotonicity holds:
        {\small
        \begin{align*} 
            \pi_{0} &= 1 \\
            \pi_{1} & = \frac{\E[S_{i2}(0) \mid G_{i} = 1]}{\E[S_{i2} \mid G_{i} = 1]} = \frac{\E[S_{i1} \mid G_{i} = 1]}{\E[S_{i2} \mid G_{i} = 1]}\frac{\E[S_{i2} \mid G_{i} = 0]}{\E[S_{i1} \mid G_{i} = 0]}  \quad  &\text{if}\quad \E[S_{i2}\mid G_{i} = 0 ] \leq \E[S_{i1}\mid G_{i}= 0] \\
            \pi_{1} & = \frac{\E[S_{i2}(0) \mid G_{i} = 1]}{\E[S_{i2} \mid G_{i} = 1]} = \frac{1-\left(1-\E[S_{i1} \mid G_{i} = 1]\right)\frac{1-\E[S_{i2} \mid G_{i} = 0]}{1-\E[S_{i1} \mid G_{i} = 0]}}{\E[S_{i2} \mid G_{i} = 1]}  \quad  &\text{if}\quad \E[S_{i2}\mid G_{i} = 0 ] > \E[S_{i1}\mid G_{i}= 0] \\
        \end{align*} 
        }%
        \item[] If Negative Monotonicity holds:
        {\small
        \begin{align*}
            \pi_{0} &= \frac{\E[S_{i2}(1) \mid G_{i} = 0]}{\E[S_{i2} \mid G_{i} = 0]} = \frac{\E[S_{i1}\mid G_{i} = 0]}{\E[S_{i2} \mid G_{i} = 0 ]}\frac{\E[S_{i2} \mid G_{i} = 1]}{\E[S_{i1}\mid G_{i} = 1]}  \quad \quad &\text{if}\quad \E[S_{i2}\mid G_{i} = 1 ] \leq \E[S_{i1}\mid G_{i}= 1] \\
            \pi_{0} &= \frac{\E[S_{i2}(1) \mid G_{i} = 0]}{\E[S_{i2} \mid G_{i} = 0]} = \frac{1-\left(1-\E[S_{i1}\mid G_{i} = 0]\right)\frac{1-\E[S_{i2} \mid G_{i} = 1]}{1-\E[S_{i1}\mid G_{i} = 1]}}{\E[S_{i2} \mid G_{i} = 0 ]}  \quad \quad &\text{if}\quad \E[S_{i2}\mid G_{i} = 1 ] > \E[S_{i1}\mid G_{i}= 1] \\
            \pi_{1} &= 1
        \end{align*}
        }%
    \end{itemize}
    Proof: The proof is analogous to the proof of Proposition \ref{prop:proportions} (Appendix \ref{proof:proportions}), complementing Lemma \ref{prop:selection} with the full results in Theorem 4.2 in \textcite{athey_identification_2006} when $\E[S_{i2} \mid G_{i} = g] \geq \E[S_{i1}\mid G_{i} = g]$.
\end{proposition}

Notice that these proportions are not the ones used to partially identify $\E[Y_{i2}(1) - Y_{i2}(0) \mid G_{i} = 1, V_{i} = AO, S_{i1}= 1]$. In that case, as indicated in footnote \ref{fn:ap1}, we would need to condition on $S_{i1} = 1$ in Proposition \ref{prop:proportions}, that is, 
\begin{remark} \label{prop:proportions_general_obs}
    Under Assumptions \ref{as:no_anti}, \ref{as:rand_samp}, \ref{as:mono}, and \ref{as:sele}, the proportion of Always-Observed units that were also observed in the pre-treatment period in the treatment group ($\tilde \pi_{1}$) and control group ($\tilde \pi_{0}$) are identified as follows:

\begin{itemize}
        \item[] If Positive Monotonicity holds:
        \begin{align*}
            \tilde\pi_{0} &= 1 \\
            \tilde\pi_{1} & = \frac{\E[S_{i2}(0) \mid G_{i} = 1, S_{i1} = 1]}{\E[S_{i2} \mid G_{i} = 1, S_{i1} = 1]} = \frac{\E[S_{i2} \mid G_{i} = 0, S_{i1} = 1]}{\E[S_{i2} \mid G_{i} = 1, S_{i1} = 1]}
        \end{align*} 
        \item[] If Negative Monotonicity holds:
        \begin{align*}
            \tilde\pi_{0} &= \frac{\E[S_{i2}(1) \mid G_{i} = 0, S_{i1} = 1]}{\E[S_{i2} \mid G_{i} = 0, S_{i1} = 1]} = \frac{\E[S_{i2} \mid G_{i} = 1, S_{i1} = 1]}{\E[S_{i2} \mid G_{i} = 0 , S_{i1} = 1]}\\
           \tilde \pi_{1} &= 1
        \end{align*}
\end{itemize}
\end{remark}
\begin{remark}
    Under Assumption \ref{as:abstate}, $\E[S_{i2}\mid G_{i} = g ] \leq \E[S_{i1}\mid G_{i}= g]$ and $\pi_{g} = \tilde \pi_{g}$ $\forall g$.
\end{remark}

The second term in equation \eqref{eq:full_attao} can be decomposed as follows using the Law of Total Expectation:
\begin{equation}
\begin{aligned}
    \E[Y_{2}(0)\mid G_{i} = 1, V_{i} = AO] =&  \E[Y_{2}(0)\mid G_{i} = 1, V_{i} = AO, S_{i1} = 1]Pr(S_{i1} = 1 \mid G_i = 1, V_i = AO) +\\
    &\E[Y_{2}(0)\mid G_{i} = 1, V_{i} = AO, S_{i1} = 0]Pr(S_{i1} = 0 \mid G_i = 1, V_i = AO), 
\end{aligned}
\end{equation}
where the partial identification of $\E[Y_{2}(0)\mid G_{i} = 1, V_{i} = AO, S_{i1} = 1]$ is discussed in Sections \ref{sec:hiru} and \ref{sec:lau}, $Pr(S_{i1} = 1 \mid G_i = 1, V_i = AO)$ is identified from Baye's rule\footnote{
\begin{align*}\tilde \pi_1 &= Pr(S_{2}(0) = 1 \mid G_{i} = 1, S_{2}(1) = 1, S_{1} =1) = Pr(V_{i} = AO \mid G_{i} = 1, S_{i2}= 1,S_{i1} =1) \\
&= \frac{Pr(S_{i1} =1 \mid V_{i} =AO, G_{i}=1, S_{2} =1)Pr(V_{i} = AO \mid G_{i} =1, S_{i2}=1)}{Pr(S_{i1} =1 \mid G_{i} =1, S_{i2}=1)} =  \frac{Pr(S_{i1} =1 \mid V_{i} =AO, G_{i}=1, S_{2} =1)\pi_{1}}{Pr(S_{i1} =1 \mid G_{i} =1, S_{i2}=1)}
\end{align*}
} by 
$$Pr(S_{i1} = 1 \mid G_i = 1, V_i = AO) = \frac{\tilde \pi_{1}\E[S_{i1}\mid G_{i}= 1, S_{i2} = 1]}{\pi_{1}},$$
and $E[Y_{2}(0)\mid G_{i} = 1, V_{i} = AO, S_{i1} = 0]$ can be bounded using the infimum and supremum of the distribution of $Y$ \parencite{manski_nonparametric_1990}, provided once again that $Y$ has compact support.

All the discussion in this section directly applies to the Estimand \ref{es:qttao} ($\qttao(q)$) as well.
\subsection{Covariates} \label{ap:ext_covariates}
Up to this point, this paper abstracts from the role of covariates. However, they have substantial potential to strengthen the identification strategy 
\parencites{grilli_nonparametric_2008}{lee_training_2009}{mealli_using_2013}{long_sharpening_2013}{semenova_generalized_2025}. This section incorporates covariates into the analysis presented in the main text.

Introducing covariates can improve inference in two different ways. First, conditioning on covariates can make identifying assumptions more credible by allowing them to hold locally rather than globally. For example, the monotonicity assumption can be restrictive, as it requires treatment to affect selection in the same direction for all units. Conditioning monotonicity on covariates relaxes this restriction by allowing treatment to affect selection in one direction for some covariate cells and in the opposite direction for others.

Second, covariates can be used to tighten the bounds, making them more informative. As proposed by \textcite{grilli_nonparametric_2008}, researchers can partition the covariate space, $X$, and estimate the Conditional $\attao$ in each cell. Then, the unconditional $\attao$ could be recovered by aggregating the conditional $\attao$, marginalizing over the distribution of $X$. \textcite{lee_training_2009} and \textcite{mealli_using_2013} show that the resulting bounds will never be wider than the unconditional ones. \textcite{long_sharpening_2013} characterize the conditions under which they will be strictly narrower. Whether these results extend to settings without treatment unconfoundedness remains an open question for future research.

To be explicit, I next introduce the conditional analogs of estimands, assumptions, and identification results for the identification of the $\attao$ in the DiD framework. Identical reasoning applies to the CiC framework and the $\qttao(q)$ estimand. Let $X\subseteq \mathbb{R}^{d}$ be the covariate space, and $X_{i}\in X$ be the vector of covariates of unit $i$. Let $\{x^{(k)}\}_{k = 1}^{K}$ be a partition of the covariate space, so that $X_{i} \in x \implies X_{i}\notin \{x^{(k)}\}_{k = 1}^{K}$\textbackslash $\{x\}$.

\begin{estimand} \label{ce:attaox}
    Conditional Average Treatment Effect on the Treated ($C\attao(x)$) 
    $$\E[Y_{i2}(1) - Y_{i2}(0) \mid G_{i}=1, V_{i} = AO, X_{i} = x]$$
\end{estimand}

\begin{assumption} \label{as:pptx} 
    Conditional Principal Parallel Trends (CPPT)
    \begin{equation*}
        \E[Y_{i2}(0) - Y_{i1} \mid G_{i} = 1, V_{i} = AO, X_{i} = x] = \E[Y_{i2}(0) - Y_{i1} \mid G_{i} = 0, V_{i} = AO, X_{i} = x]
    \end{equation*}
\end{assumption}

\begin{assumption} \label{as:mono_conditional}
    Conditional Monotonicity
    \begin{align*}
        S_{i2}(1) \geq S_{i2}(0) \quad & \forall i:X_{i} = x \quad \text{Positive Monotonicity} \\
       & \text{or} \\
        S_{i2}(1) \leq S_{i2}(0) \quad & \forall i: X_{i} = x \quad \text{Negative Monotonicity}         
    \end{align*}
\end{assumption}

\begin{assumption} \label{as:sele_covariates}
    Conditional Selection Model \\
    Under Positive Monotonicity:
    \begin{equation*}
    \begin{aligned}
       ( S_{it}(0) \mid X_{i} = x) = h_x^{0}(U^{x}_{it},t) .
    \end{aligned}
    \end{equation*}
    Under Negative Monotonicity:
    \begin{equation*}
        \\
        (S_{it}(1)\mid X_{i}  =x ) = h_x^{1}(U^{x}_{it},t),
    \end{equation*}
    where $U^{x}_{it}$ is an unobservable scalar for unit $i$ at time $t$ and $h_x^{w}(u,t)$ is non decreasing function in $u$ $\forall$ $t \in\{1,2\}$, $w = \{0,1\}$. \\
    The unobservable $U_{}^{x}$ is continuously distributed and has the same compact support in both groups. Its distribution is constant over time within groups
    \begin{equation*}
        (U_{i1}^{x} \mid G_{i}, X_{i} = x) \sim (U_{i2}^{x} \mid G_{i}, X_{i} = x ).
    \end{equation*}
    Additionally, given the realized selection outcome, the distribution of $U$ is independent of the group in a given time period.
    \begin{equation*}
        U_{it}^{x} \perp \! \! \! \perp G_{i} \mid S_{it}, X_{i} = x
    \end{equation*}
\end{assumption}

\begin{remark} \label{prop:bounds_attao_covariates}
    Let $Y_{it}(1)$ and $Y_{it}(0)$ be continuous. If Assumptions \ref{as:no_anti}, \ref{as:abstate}, \ref{as:rand_samp}, \ref{as:pptx}, \ref{as:mono_conditional} and \ref{as:sele_covariates} hold, then $\Delta^{LB}(x)$ and $\Delta^{UB}(x)$ are lower and upper bounds for the Conditional Average Treatment Effect on the Treated Always-Observed units ($\attao(x)$), where
    \begin{align*}
       \Delta^{LB}(x) & =  \E[\ddot Y_{i} \mid G_{i} = 1, S_{i2} = 1, \ddot Y_{i} \leq \ddot y_{\pi_{1}}^{1}, X_{i} = x ] - \E[\ddot Y_{i} \mid G_{i} = 0, S_{i2} = 1, \ddot Y_{i} \geq \ddot y_{1 - \pi_{0}}^{0}, X_{i} = x ] \\
       \Delta^{UB}(x) & =  \E[\ddot Y_{i} \mid G_{i} = 1, S_{i2} = 1, \ddot Y_{i} \geq \ddot y_{1 -\pi_{1}}^{1} , X_{i} = x] - \E[\ddot Y_{i} \mid G_{i} = 0, S_{i2} = 1, \ddot Y_{i} \leq \ddot y_{\pi_{0}}^{0}, X_{i} = x ] \\
       \ddot Y_{i} & = Y_{i2} - Y_{i1} \\
       \ddot y_{q}^{g} &= \inf \{\ddot y: F(\ddot y) \geq q\}, \text{ with }F \text{ the c.d.f. of }\ddot Y\text{ conditional on } S_{2} = 1, X_{i}= x \text{ and } G = g 
    \end{align*}
     \begin{itemize}
        \item[] If Conditional Positive Monotonicity holds:
        \begin{align*}
            \pi_{0} &= 1 \\
            \pi_{1} & = \frac{\E[S_{i1} \mid G_{i} = 1, X_{i} = x]}{\E[S_{i2} \mid G_{i} = 1, X_{i} = x]}\frac{\E[S_{i2} \mid G_{i} = 0, X_{i} = x]}{\E[S_{i1} \mid G_{i} = 0, X_{i} = x]}
        \end{align*} 
        \item[] If Conditional Negative Monotonicity holds:
        \begin{align*}
            \pi_{0} & = \frac{\E[S_{i1}\mid G_{i} = 0, X_{i} = x]}{\E[S_{i2} \mid G_{i} = 0 , X_{i} = x]}\frac{\E[S_{i2} \mid G_{i} = 1, X_{i} = x]}{\E[S_{i1}\mid G_{i} = 1, X_{i} = x]} \\
            \pi_{1} &= 1
        \end{align*}
    \end{itemize}
    The proof follows verbatim by conditioning throughout and using Propositions \ref{prop:proportions} and \ref{prop:bounds_attao}.
\end{remark}

The adjusted $\overline \attao$ is equal to the weighted average of $C \attao(x)$, marginalizing over the distribution of $X$ for the treated Always-Observed units: $$\overline \attao \in \left[\overline{ \Delta^{LB}} , \overline{\Delta^{UB}}\right],$$ where
\begin{align*}
    \overline{ \Delta^{LB}} & = \sum_{x \in X}\Delta^{LB}(x) p(x), \\
    \overline{\Delta^{UB}} & = \sum_{x \in X}\Delta^{UB}(x) p(x), \\
    p(x) &= Pr(X_{i} = x \mid G_{i} = 1, V_{i} = AO)
\end{align*}
\begin{conjecture}
    The adjusted $\overline \attao$ will always lie inside the unadjusted $\attao$ $$\left[\overline{ \Delta^{LB}} , \overline{\Delta^{UB}}\right] \subseteq \left[\Delta^{LB}, \Delta^{UB}\right]$$
    This statement is proved in \textcite[Proposition 1]{long_sharpening_2013} under the assumption of treatment unconfoundedness.
\end{conjecture}

When it comes to estimation, each of the $\text{C}\attao(x)$ can be estimated as described in Section \ref{sec:lau} in each of the cells that partition the covariate space. These estimates can be aggregated as expressed in \eqref{eq:agg_cattao}, where each of the $\text{C}\attao(x)$ is weighted by the number of treated units and the estimated proportion of AO units in this cell.

\begin{equation}
    \label{eq:agg_cattao}
    \begin{aligned}
        \widehat{\Delta^{LB}} &= \sum_{x\in X} \omega_{x}\widehat{\Delta^{LB}} (x) \\
         \widehat{\Delta^{UB}} &=\sum_{x\in X} \omega_{x}\widehat{\Delta^{UB}} (x) \\
         \omega_{x} &= \frac{\widehat \pi_{1}^{x}}{\Pi_{1}}\frac{\sum_{i}\indicator{(G_{i}= 1, S_{i2}= 1,X_{i} = x)}}{N_{1}} \\
         \Pi_{1} &= \sum_{x\in X}\widehat \pi_{1}^{x}
    \end{aligned}
\end{equation}

The discretized, cell-based procedure just described can help make the identification strategy more credible and informative. This procedure goes in the direction of \textcite{semenova_generalized_2025} and could be generalized and improved by adapting the Generalized Lee Bounds in \textcite{semenova_generalized_2025} to the panel data setting considered in this paper.

\subsection{Repeated Cross-Sections} \label{ap:ext_rcs}
 Throughout the paper, I have conditioned on both $S_{i1}$ and $S_{i2}$ and, for the DiD case, used $\ddot Y_{i} = Y_{i2} - Y_{i1}$. When panel data are unavailable, and only repeated cross-sections are observed, these quantities cannot be constructed. This section extends the results to the repeated cross-section setting.

\begin{lemma} \label{prop:cross_qttao}
    Let $Y_{it}(1)$ and $Y_{it}(0)$ be continuous with compact support. Furthermore, let the support of $Y_{it}(0)$ for the treatment group be contained in the support of $Y_{it}(0)$ for the control group. Then, if Assumptions \ref{as:no_anti}, \ref{as:abstate}, \ref{as:rand_samp} and \ref{as:outcome} hold, then
    \begin{align*}
         \qttao(q) = Q_{Y_{2} \mid G=1, V = AO}(q) -  Q_{Y_{2} \mid G=0, V=AO}\left(F_{Y_{1}\mid G=0, V=AO}\left(Q_{Y_{1}\mid G=1,V=AO}(q)\right) \right)
    \end{align*}
    Proof: See equation \eqref{eq:ap_qttao_observed}.
\end{lemma}
\begin{lemma} \label{prop:cross_attao}
     Let $Y_{it}(1)$ and $Y_{it}(0)$ be continuous. If Assumptions \ref{as:no_anti}, \ref{as:abstate}, \ref{as:rand_samp} and \ref{as:ppt} hold, then 
    \begin{align*}
       \attao =& \E[Y_{i2} \mid G_{i} = 1, V_{i} = AO]-\E[Y_{i1} \mid G_{i} = 1, V_{i} = AO]\\
       &+\E[Y_{i2} \mid G_{i} = 0, V_{i} = AO]-\E[Y_{i1} \mid G_{i} = 0, V_{i} = AO]
    \end{align*}
    Proof: See equations \eqref{eq:missing_po_ao} and \eqref{eq:princ_did}.
\end{lemma}
Lemmas \ref{prop:cross_qttao} and \ref{prop:cross_attao} identify the two causal estimands as functions of the expectations or c.d.f.s of observed outcomes, conditional on a latent stratum membership. Proposition \ref{prop:proportions} could be directly applied to the repeated cross-section case to identify the proportion of AO units in each group at $t=2$, that is, $Pr(V_{i} = AO \mid G=g, S_{2} = 1)$. These proportions can be used to partially identify the distribution of $Y_{2} \mid G=g, S_{2} = 1$ and its expectation.

In the panel setting, these proportions can also be used to partially identify the distribution of $Y_{i1} \mid G=g, V_{i} = AO$ and its expectation, exploiting the observed distribution of $Y_{i1} \mid G=g, V_{i} = AO, S_{i2} = 1$. However, in the repeated cross-sections setting, this is not feasible, and identification must instead rely on the distribution of $Y_{i1} \mid G=g,S_{1} = 1$.

\begin{proposition} \label{prop:cross_bounds_qttao}
    Let $Y_{it}(1)$ and $Y_{it}(0)$ be continuous with compact 
    support. Furthermore, let the support of $Y_{it}(0)$ for the treatment group be contained in the support of $Y_{it}(0)$ for the control group. Then, if Assumptions \ref{as:no_anti}, \ref{as:abstate}, \ref{as:rand_samp} and \ref{as:outcome} hold, then $\Lambda^{LB}(q)$ and $\Lambda^{UB}(q)$ are lower and upper bounds for the Quantile Treatment Effect on the Treated Always-Observed units ($\qttao(q)$), where:
    \begin{align*}
        \Lambda^{LB}(q) &= Q_{Y_2 \mid G=1, S_{2} = 1}(q\pi_1) - Q_{Y_{2} \mid G=0, S_2 = 1}\left(F_{Y_{1}\mid G=0, S_1 = 1}\left(Q_{Y_{1}\mid G=1,S_1 = 1}(q\varpi_1 + 1 - \varpi_1)\right)\frac{\pi_{0}}{\varpi_{0}} + 1 - \pi_0 \right), \\
         \Lambda^{UB}(q) &= Q_{Y_2 \mid G=1, S_{2} = 1}(q\pi_1 + 1 - \pi_1) -  Q_{Y_{2} \mid G=0,S_2 = 1}\left(\left(F_{Y_{1}\mid G=0, S_1 = 1}\left(Q_{Y_{1}\mid G=1,S_1=1}(q\varpi_1)\right)-(1-\varpi_0)\right)\frac{\pi_{0}}{\varpi_{0}}\right), \\
        \pi_{1} & := Pr(V_{i} = AO \mid G_{i} = 1, S_{i2} = 1)\\
       \pi_{0} & := Pr(V_{i} = AO \mid G_{i} = 0, S_{i2} = 1) \\
       \varpi_{1} & := Pr(V_{i} = AO \mid G_{i} = 1, S_{i1} = 1)\\
       \varpi_{0} & := Pr(V_{i} = AO \mid G_{i} = 0, S_{i1} = 1) \\
       F_{Y}(y) & := Pr(Y \leq y) \\
       Q_{Y}(q) & := \inf\{y: F_{Y}(y) \geq q\}
    \end{align*}
    provided that 
    \begin{align*}
        F_{Y_{1}\mid G=0, S_2 = 1}\left(Q_{Y_{1}\mid G=1,S_2 = 1}(q\varpi_1 + 1 - \varpi_1)\right) &\leq \varpi_0 \\
        F_{Y_1 \mid G=0, S_2 = 1}\left(Q_{Y_1\mid G=1, S_2 = 0} (q\varpi_1)\right) &\geq 1 - \varpi_0.
    \end{align*}
    If $F_{Y_{1}\mid G=0, S_2 = 1}\left(Q_{Y_{1}\mid G=1,S_2 = 1}(q\varpi_1 + 1 - \varpi_1)\right) > \varpi_0$, then 
    \begin{equation*}
         \Lambda^{LB}(q) = Q_{Y_2 \mid G=1, S_{2} = 1}(q\pi_1) - Q_{Y_{2} \mid G=0, S_2 = 1}\left(1\right).
    \end{equation*}
    If $F_{Y_1 \mid G=0, S_2 = 1}\left(Q_{Y_1\mid G=1, S_2 = 0} (q\varpi_1)\right) < 1 - \varpi_0$, then 
    \begin{equation*}
        \Lambda^{UB}(q) = Q_{Y_2 \mid G=1, S_{2} = 1}(q\pi_1 + 1 - \pi_1) -  Q_{Y_{2} \mid G=0,S_2 = 1}\left(0\right)
    \end{equation*}
\end{proposition}
\begin{proposition} \label{prop:cross_bounds_attao}
    Let $Y_{it}(1)$ and $Y_{it}(0)$ be continuous. If Assumptions \ref{as:no_anti}, \ref{as:abstate}, \ref{as:rand_samp} and \ref{as:ppt} hold, then $\Delta^{LB}$ and $\Delta^{UB}$ are lower and upper bounds for the Average Treatment Effect on the Treated Always-Observed units ($\attao$), where
    \begin{align*}
        \Delta^{LB} &=  \E[Y_{i2} \mid G_{i} = 1, Y_{i2} \leq y^{1}_{2,\pi_{1}}]-\E[Y_{i1} \mid G_{i} = 1, Y_{i1} \geq y^{1}_{1, 1-\varpi_{1}}] \\
        &+\E[Y_{i2} \mid G_{i} = 0, Y_{i2} \leq y^{0}_{2, \pi_{0}}]-\E[Y_{i1} \mid G_{i} = 0, Y_{i1} \geq y^{1}_{1, 1 - \varpi_{0}}]\\
       \Delta^{UB} &=  \E[Y_{i2} \mid G_{i} = 1, Y_{i2} \geq y^{1}_{2,1-\pi_{1}}]-\E[Y_{i1} \mid G_{i} = 1, Y_{i1} \leq y^{1}_{1, \varpi_{1}}]\\
       & +\E[Y_{i2} \mid G_{i} = 0, Y_{i2} \geq y^{0}_{2, 1-\pi_{0}}]-\E[Y_{i1} \mid G_{i} = 0, Y_{i1} \leq y^{1}_{1, \varpi_{0}}]\\
       \ddot y_{t, q}^{g} &:= \inf \{\ddot y: F_{Y_{t} \mid G=g, S_{t} =1}(\ddot y) \geq q\}, \text{ with }F_{Y_{t}\mid G =g, S_{t} = 1} \text{ the c.d.f. of } Y_{t}\text{ conditional on } S_{t} = 1 \text{ and } G = g \\
       \pi_{1} & := Pr(V_{i} = AO \mid G_{i} = 1, S_{i2} = 1)\\
       \pi_{0} & := Pr(V_{i} = AO \mid G_{i} = 0, S_{i2} = 1) \\
       \varpi_{1} & := Pr(V_{i} = AO \mid G_{i} = 1, S_{i1} = 1)\\
       \varpi_{0} & := Pr(V_{i} = AO \mid G_{i} = 0, S_{i1} = 1).
    \end{align*}
\end{proposition}
The proof of Propositions \ref{prop:cross_bounds_qttao} and \ref{prop:cross_bounds_attao} are similar to the proofs of Propositions \ref{prop:bounds_qttao} and \ref{prop:bounds_attao} respectively and are omitted. 

The remaining challenge is the identification of $\varpi_{0}$ and $\varpi_{1}$, used to partially identify the distribution of $Y_{1} \mid G=g, V_{i} =AO$ from the distribution of $Y_{1} \mid G_{i} =g, S_{1} =1$.

\begin{proposition} \label{prop:cross_proportion}
    Under Assumptions \ref{as:no_anti}, \ref{as:abstate}, and \ref{as:rand_samp}, the proportions of AO units in the pre-treatment period are given by
    \begin{align*}
        \varpi_{1} &= \pi_{1}\frac{\E[S_{i2} \mid G_{i} = 1]}{\E[S_{i1} \mid G_{i} = 1]} \\
        \varpi_{0} &= \pi_{0}\frac{\E[S_{i2} \mid G_{i} = 0]}{\E[S_{i1} \mid G_{i} = 0]} 
    \end{align*}
    Proof: See Appendix \ref{proof:cross_proportion}.
\end{proposition}

Finally, notice that without panel data, Assumption \ref{as:abstate} cannot be tested since the pair $(S_{i1},S_{i2})$ is never jointly observed for any unit. Nevertheless, if $\E[S_{i1}\mid G_{i}=g] \geq \E[S_{i2}\mid G_{i} =g]$, this could be interpreted as evidence against this assumption. When Assumption \ref{as:abstate} does not hold, $\pi_{1}$ and $\pi_{0}$ can still be identified as described in Proposition \ref{prop:proportions_general}. On the other hand, $\varpi_{1}$ and $\varpi_{0}$ can only be partially identified as follows:
\begin{proposition} \label{prop:cross_proportion_bounds}
    Under Assumptions \ref{as:no_anti} and \ref{as:rand_samp}, the proportions of AO units in the pre-treatment period are bounded by
    \begin{align*}
        \max\left\{0, \frac{\pi_{1}\E[S_{i2}\mid G_{i}= 1]-\left(1 - \E[S_{i1}\mid G_{i} = 1]\right)}{\E[S_{i1}\mid G_{i} = 1]}\right\}\leq \varpi_{1} \leq \min\left\{1,\pi_{1}\frac{\E[S_{i2}\mid G_{i}=1]}{\E[S_{i1}\mid G_{i}=1]}\right\} \\
        max\left\{0, \frac{\pi_{0}\E[S_{i2}\mid G_{i}= 0]-\left(1 - \E[S_{i1}\mid G_{i} = 0]\right)}{\E[S_{i1}\mid G_{i} = 0]}\right\}\leq \varpi_{0} \leq \min\left\{1,\pi_{0}\frac{\E[S_{i2}\mid G_{i}=0]}{\E[S_{i1}\mid G_{i}=0]}\right\}
    \end{align*}
    Proof: See Appendix \ref{proof:cross_proportion_bounds}.
\end{proposition}
Intuitively, once the proportion of AO units in the observed post-treatment period is identified ($\pi_{g}$), then the total proportion of AO units in the group $g$ is identified as well. Because $V_{i}$ is defined using the post-treatment potential selection outcomes, under Assumption \ref{as:rand_samp}, the proportion of units with $V_{i} =AO$ is the same in the pre and post-treatment periods. Under Assumption \ref{as:abstate}, all the units such that $V_{i}=AO$ would also be observed in the pre-treatment period and hence $Pr(V_{i} = AO \mid G_{i} = g, S_1 = 1)$ is identified. When Assumption \ref{as:abstate} does not hold, it could be that some of the missing units in the pre-treatment period have $V_{i} = AO$. In the worst-case scenario, all the missing units in the pre-treatment period belong to the $AO$ stratum. In the best-case scenario, all the units with $V_{i} =AO$ are observed at $t = 1$. These extreme cases yield bounds on $Pr(V_{i} =AO \mid G_{i} =g, S_{i1}= 1)$ as described in Proposition \ref{prop:cross_proportion_bounds} and Appendix \ref{proof:cross_proportion_bounds}. These bounds can then be used in Propositions \ref{prop:cross_bounds_qttao} and \ref{prop:cross_bounds_attao}.

\begin{example}
    Consider the case where 90\% of the units are observed in the pre-treatment period in group $G=g$, 80\% of the units are observed in the post-treatment period in the same group, out of which 90\% are AO. That is, $\E[S_{i1}\mid G_i=g] = 0.9$, $\E[S_{i2}\mid G_i=g] = 0.8$ and $Pr(V_{i} = AO \mid G_i = g, S_{i2} = 1) = 0.9$.
    
    90\% of the observed units in the post-treatment period are AO, meaning that $0.9\times 0.8 = 0.72$ of the total units in group $g$ belong to this stratum. The remaining $0.1 \times 0.8 = 0.08$ belongs to the OT stratum, and the 0.2 not observed belongs to the OC and NO strata.

    This means that 72\% of the units belong to the AO stratum. If $V_{i} = AO \implies S_{i1} = 1$, it follows that $0.72/0.9 = 0.8$ units observed in the pre-treatment period would belong to the AO stratum in the post-treatment period. On the other hand, if all the missing units in the pre-treatment period (10\%) belong to the AO stratum, then it follows that there are still 62\% of units who have $S_{i1} = 1$ and $V_{i} = AO$, which is $0.62/0.9\approx0.68$ of the observed treated group at $t=1$. 
\end{example}

\subsection{Discrete outcomes} \label{ap:ext_discrete}
So far, I have assumed that the outcome $Y$ is continuously distributed over a compact support. When the outcome is discrete, it is still possible to define the quantile function $Q_{Y}(q) = \{\inf y : F(y) \geq q\}$. However, because discrete data typically exhibit ties, this function is no longer strictly increasing in $q$; instead, it becomes a step function. As a result, the identification strategy developed for continuous outcomes is no longer directly applicable. In the DiD case, the trimming procedure developed in Proposition \ref{prop:bounds_attao} relies on the identification of the trimming values, $\ddot y_{q}^{} $. However, with discrete outcomes (or continuous outcomes with ties), it is possible that $F(\ddot y_{q}) \neq q$. In the CiC case, the strict monotonicity of the function $m(\cdot)$ in Assumption \ref{as:outcome} has too strong consequences. Both points are illustrated in Example \ref{ex:discrete}.

\begin{example} \label{ex:discrete}
    Assume that $ Y_{i} \sim Bernoulli(0.7)$ and that $\pi_{} = 0.9$. That is, I have a binary outcome, and I know that 90\% of the observed units belong to the AO stratum. In this case, $y_{0.9} = Q(0.9) = 1$, and the trimming procedure described in Proposition \ref{prop:bounds_attao} would not trim any unit, with $\E[Y_{i} \mid Y_{i} \leq y_{0.9}] = 0.7$. If, instead, one were to apply a strict inequality, the procedure would trim 30\% of the units, with $\E[Y_{i} \mid y_{0.9}< 1] = 0$. It is clear that neither approach delivers a satisfactory trimming rule in this setting.

    Consider now the CiC outcome model in Assumption \ref{as:outcome}. According to this model, $\mathcal{U}_{i} > \mathcal{U}_{j} \iff Y_{i} > Y_{j}$. With a binary outcome, this implies that the unobservable $\mathcal{U}_{i}$ can take only two values: $\mathcal{U}_{0}$ for all the units with $Y_{i} = 0$ and $\mathcal{U}_{1}$ for all the units with $Y_i = 1$. Furthermore, since this distribution is constant over time, it implies that the distribution of $Y$ is also constant, which is restrictive.
\end{example}

\subsubsection{Identification in the CiC setup with Discrete outcomes}
For the CiC case, I extend the methodoloy developed in this paper to the Discrete Changes-in-Changes (DCIC) design \textcite[Section 4]{athey_identification_2006}.

\begin{assumption} \label{as:outcome_discrete}
    Outcome model \\
    \begin{equation*}
        (Y_{it}(0) \mid V_{i} = AO )\quad =\quad  m(\mathcal{U}_{it}, t),
    \end{equation*}
    where $m(u,t)$ is non-decreasing in $u$ $\forall t$, and $\mathcal{U}_{it}$ is an unobservable scalar for unit $i$ at time $t$ with continuous and constant distribution over time within groups,
    \begin{equation*}
        \mathcal{U}_{i1} \mid G_{i}, V_{i} = AO  \sim \mathcal{U}_{i2} \mid G_{i}, V_{i} = AO.
    \end{equation*}
\end{assumption}
In Assumption \ref{as:outcome}, the function $m(\mathcal{U}_{it},t)$ was assumed to be strictly increasing in $\mathcal{U}_{it}$. In Assumption \ref{as:outcome_discrete}, this function is assumed to be non-decreasing. This assumption is weaker and eliminates the strong implications of the CiC model with a discrete outcome highlighted earlier. The downside is that the distribution of the missing potential outcome for the treated group can only be partially identified \parencite{athey_identification_2006}. Furthermore, the unobservable $\mathcal{U}_{it}$ is assumed to be continuous.

\begin{proposition} \label{prop:bounds_qttao_discrete}
    Let $Y_{it}(1)$ and $Y_{it}(0)$ be discrete with compact support. Furthermore, let the support of $Y_{it}(0)$ for the treatment group be contained in the support of $Y_{it}(0)$ for the control group. Then, if Assumptions \ref{as:no_anti}, \ref{as:abstate}, \ref{as:rand_samp} and \ref{as:outcome_discrete} hold, then $\Lambda^{LB}(q)$ and $\Lambda^{UB}(q)$ are lower and upper bounds for the Quantile Treatment Effect on the Treated Always-Observed units ($\qttao(q)$), where:
    \begin{align*}
        \Lambda^{LB}(q) &= Q_{Y_2 \mid G=1, S_{2} = 1}(q\pi_1) -  Q_{Y_{2} \mid G=0, S_2 = 1}\left(F_{Y_{1}\mid G=0, S_2 = 1}\left( Q_{Y_{1}\mid G=1,S_2 = 1}(q\pi_1 + 1 - \pi_1)\right) + 1 - \pi_0 \right), \\
         \Lambda^{UB}(q) &= Q_{Y_2 \mid G=1, S_{2} = 1}(q\pi_1 + 1 - \pi_1) -  \widetilde Q_{Y_{2} \mid G=0,S_2 = 1}\left(F_{Y_{1}\mid G=0, S_2 = 1}\left(\widetilde Q_{Y_{1}\mid G=1,S=1}(q\pi_1)\right)-(1-\pi_0)\right), \\
        \pi_{1} & = Pr(S_{i2}(0) = 1 \mid G_{i} = 1, S_{i2}(1) = 1),\\
       \pi_{0} & = Pr(S_{i2}(1) = 1 \mid G_{i} = 0, S_{i2}(0) = 1), \\
       F_{Y}(y) & := Pr(Y \leq y), \\
       Q_{Y}(q) & := \inf\{y: F_{Y}(y) \geq q\} , \\
        \widetilde Q_{Y}(q) & := \sup\{y\cup\{-\infty\}: F_{Y}(y) \leq q\} ,
    \end{align*}
    provided that 
    \begin{align*}
        F_{Y_{1}\mid G=0, S_2 = 1}\left(Q_{Y_{1}\mid G=1,S_2 = 1}(q\pi_1 + 1 - \pi_1)\right) &\leq \pi_0 \\
        F_{Y_1 \mid G=0, S_2 = 1}\left(\widetilde Q_{Y_1\mid G=1, S_2 = 0} (q\pi_1)\right) &\geq 1 - \pi_0.
    \end{align*}
    If $F_{Y_{1}\mid G=0, S_2 = 1}\left( Q_{Y_{1}\mid G=1,S_2 = 1}(q\pi_1 + 1 - \pi_1)\right) > \pi_0$, then 
    \begin{equation*}
         \Lambda^{LB}(q) = Q_{Y_2 \mid G=1, S_{2} = 1}(q\pi_1) - Q_{Y_{2} \mid G=0, S_2 = 1}\left(1\right).
    \end{equation*}
    If $F_{Y_1 \mid G=0, S_2 = 1}\left(\widetilde Q_{Y_1\mid G=1, S_2 = 0} (q\pi_1)\right) < 1 - \pi_0$, then 
    \begin{equation*}
        \Lambda^{UB}(q) = Q_{Y_2 \mid G=1, S_{2} = 1}(q\pi_1 + 1 - \pi_1) -  \widetilde Q_{Y_{2} \mid G=0,S_2 = 1}\left(0\right).
    \end{equation*}
    Where I use the convention that $F_{Y}(-\infty) = 0$. \\
    
    Proof: Similar as proof in Appendix \ref{proof:bounds_qttao}, replacing Lemma \ref{prop:AIprop3} by Theorem 4.1 in \textcite{athey_identification_2006}.
\end{proposition}
Proposition \ref{prop:bounds_qttao_discrete} partially identifies the $\qttao(q)$ under Assumption \ref{as:outcome_discrete}. The key difference relative to Proposition \ref{prop:bounds_qttao} concerns the lower bound, $\Delta^{LB}$, where the left-continuous quantile function, $Q_{Y}(q)$, is replaced by the right-continuous quantile function, $\widetilde{Q}_{Y}(q)$. This modification follows from the results in Section 4 in \textcite{athey_identification_2006}. When the outcome is discrete, or more generally, there are mass points that generate ties in the data, the distribution $F_{Y_{2}(0)\mid G=1}(y)$ can only be partially identified. The upper bound is given using the left-continuous quantile function, $Q$, as it `jumps' at the first $y$ where $F_{Y}(y)$ reaches or exceeds $q$. On the other hand, the right-continuous quantile function, $\widetilde Q$, `jumps' at the last value of $y$ such that $F_{Y}(y)$ is not larger than $q$. When $F_{Y}(y)$ is strictly increasing, then $Q_{Y}(q) = \widetilde Q_{Y}(q)$. When $F_{Y}(y)$ is flat, $Q_{Y}(q) > \widetilde Q_{Y}(q)$. Consequently, $Q$ is used to compute the upper bound for AO units, $\Delta^{LB}$, while $\widetilde Q$ is used for the lower bound, $\Delta^{LB}$. 

\subsubsection{Identification in the DiD setup with Discrete outcomes}
The solution to the problem illustrated in Example \ref{ex:discrete} is to reweight units so that the weighted average of the outcome successfully captures the upper or lower bound of the latent distribution for the AO units. That is, define a new p.m.f. that is equivalent to the original p.m.f. after trimming $\pi\times 100$ percent of the units. To that end, define:

\begin{equation} \label{eq:discrete_lb}
     \mathcal{\underline P}_{Y}^{\pi}( y) = \begin{cases}
         \frac{Pr(Y = y)}{\pi} \quad \quad \text{if} \quad y < Q_{Y}(\pi) \\
         \frac{\pi -Pr(y < Q_{Y}(\pi))}{\pi} \quad \quad \text{if} \quad y = Q_{Y}(\pi) \\
         0 \quad \quad \text{if} \quad y > Q_{Y}(\pi)
     \end{cases}  
\end{equation}
\begin{equation} \label{eq:discrete_ub}
    \mathcal{\overline P }_{Y}^{\pi} (y) = \begin{cases}
        0 \quad \quad \text{if} \quad y < Q_{Y}(1-\pi) \\
        \frac{\pi - Pr(Y > Q_{Y}(1-\pi))}{\pi} \quad \quad \text{if} \quad y =  Q_{Y}(1-\pi)\\
        \frac{Pr(Y = y)}{\pi} \quad \quad \text{if} \quad y > Q_{Y}(1-\pi)
    \end{cases}
\end{equation}
Equations \eqref{eq:discrete_lb} and \eqref{eq:discrete_ub} describe a trimming function for discrete outcomes given a trimming proportion $\pi$. In equation \eqref{eq:discrete_lb}, the trimming threshold is given by $Q_{Y}(\pi)$. All the outcomes strictly above this threshold receive zero weight, and all the units strictly below receive a weight equal to their probability. Because $Y$ is discrete, there could be a mass at $Y = Q_{y}(\pi)$, implying that $Pr(Y < Q_{Y}(\pi)) < \pi < Pr(Y \leq Q_{Y}(\pi))$. If that is the case, the units with $Y = Q_{Y}(\pi)$ are reweighted, using $\pi - Pr(Y<Q_{Y}(\pi))$ as weight instead of $Pr(Y = Q(\pi)) = Pr(Y \leq Q_{Y}(\pi)) - Pr(Y>Q_{Y}(\pi))$. Finally, the weights are normalized by $\pi$, ensuring that $\mathcal{\underline P}_{Y}^{\pi}(y)$ is a proper probability mass function. Exactly the same intuition is behind equation \eqref{eq:discrete_ub}.

\begin{proposition} \label{prop:bounds_attao_discrete}
    Let $Y_{it}(1)$ and $Y_{it}(0)$ be discrete. If Assumptions \ref{as:no_anti}, \ref{as:abstate}, \ref{as:rand_samp} and \ref{as:ppt} hold, then $\Delta^{LB}$ and $\Delta^{UB}$ are lower and upper bounds for the Average Treatment Effect on the Treated Always-Observed units ($\attao$), where
    \begin{align*}
        \Delta^{LB} & = \sum_{i:G_{i} = 1, S_{i2}= 1}\ddot Y_{i}\mathcal{\underline P}_{\ddot Y_{i} \mid G_{i} = 1, S_{i2} = 1}^{\pi_{1}}(y)-\sum_{i:G_{i} = 0, S_{i2}= 1}\ddot Y_{i}\mathcal{\overline P}_{\ddot Y_{i} \mid G_{i} = 0, S_{i2} = 1}^{\pi_{0}}(y) \\
         \Delta^{UB} & = \sum_{i:G_{i} = 1, S_{i2}= 1}\ddot Y_{i}\mathcal{\overline P}_{\ddot Y_{i} \mid G_{i} = 1, S_{i2} = 1}^{\pi_{1}}(y)-\sum_{i:G_{i} = 0, S_{i2}= 1}\ddot Y_{i}\mathcal{\underline P}_{\ddot Y_{i} \mid G_{i} = 0, S_{i2} = 1}^{\pi_{0}}(y) \\
       \pi_{1} & = Pr(S_{i2}(0) = 1 \mid G_{i} = 1, S_{i2}(1) = 1)\\
       \pi_{0} & = Pr(S_{i2}(1) = 1 \mid G_{i} = 0, S_{i2}(0) = 1).
    \end{align*}
    
\end{proposition}

\renewcommand{\theequation}{B\arabic{equation}}
\renewcommand{\thetable}{B\arabic{table}}
\renewcommand{\thefigure}{B\arabic{figure}}
\renewcommand{\theexample}{B\arabic{example}}
\renewcommand{\thelemma}{B\arabic{lemma}}
\setcounter{equation}{0}
\setcounter{table}{0}
\setcounter{figure}{0}
\setcounter{example}{0}
\setcounter{lemma}{0}

\section{Proofs} \label{ap:proofs}
\subsection{Proof of Proposition \ref{prop:bounds_qttao}} \label{proof:bounds_qttao}
The structure of the proof is as follows. First, I draw on results from \textcite{horowitz_identification_1995} to partially identify the outcome distribution for the Always-Observed units in all the periods and groups. Next, I use results from \textcite{athey_identification_2006} to express the Causal Estimand \ref{es:qttao}, $\qttao$, as a function of distributions of observed outcomes. Finally, I combine these results to partially identify the $\qttao(q)$.

\begin{proof}
    \begin{lemma} \label{prop:HMprop2}
    Let $Y$ be a continuous random variable and a mixture of two random variables, $Y_{0}$ and $Y_{1}$, with known mixing probability, $p$. The c.d.f. of the observed $Y$ is $F(y) = (1-p) F_{0}(y) + pF_{1}(y)$, where $F_{0}$ and $F_{1}$ are the c.d.f.s of $Y_{0}$ and $Y_{1}$ respectively. Then, the c.d.f. of $Y_{1}$ is bounded as follows:
    \begin{equation*}
        F_{1}(y) \in \left[\max\left\{0, \frac{F(y) - (1-p)}{p}\right\}, \min\left\{\frac{F(y)}{p},1\right\}\right]
    \end{equation*}
    These bounds are sharp.\\
    Proof: See \textcite[Proposition 2]{horowitz_identification_1995}.
\end{lemma}
Consider the observed outcome for the treatment group in the post-treatment period. Denote its c.d.f. by $F_{Y_{2} \mid G = 1, S_2 = 1}(y)$. This distribution is a mixture of the outcome distribution for AO and OT units with mixing probability $\pi_{1}$. Accordingly, I can write
\begin{equation} \label{eq:proof_mixture}
    F_{Y_{2} \mid G = 1,S_2 = 1}(y) = \pi_{1}F_{Y_2 \mid G=1, V = AO}(y) + (1-\pi_{1})F_{Y_{2}\mid G = 1, V = OT}(y).
\end{equation}
From Lemma \ref{prop:HMprop2}, it follows that 
\begin{equation}
    F_{Y_2 \mid G=1, V = AO}(y) \in \left[\max\left\{0,\frac{F_{Y_{2} \mid G = 1, S_2 = 1}(y) - (1- \pi_{1})}{\pi_{1}}\right\},\min\left\{\frac{F_{Y_{2} \mid G = 1, S_2 = 1}(y)}{\pi_{1}},1\right\}\right].
\end{equation}

Consider now the observed outcome for the treatment group in the pre-treatment period, restricted to units also observed in the post-treatment period. Denote its c.d.f. by $F_{Y_{1}\mid G=1, S_{2} = 1}(y)$. Because I am conditioning on being observed in the post-treatment period, this distribution is also a mixture of the outcome distribution for AO and OT units with mixing probability $\pi_{1}$. Therefore, I can write
\begin{equation} \label{eq:proof_mixture_bi}
    F_{Y_{1} \mid G = 1, S_{2}= 1}(y) = \pi_{1}F_{Y_1 \mid G=1, V = AO}(y) + (1-\pi_{1})F_{Y_{1}\mid G = 1, V = OT}(y).
\end{equation}
From Lemma \ref{prop:HMprop2}, it follows that 
\begin{equation}
    F_{Y_1 \mid G=1, V = AO}(y) \in \left[\max\left\{0,\frac{F_{Y_{1} \mid G = 1,S_{2} = 1}(y) - (1- \pi_{1})}{\pi_{1}}\right\},\min\left\{\frac{F_{Y_{1} \mid G = 1,S_{2}= }(y)}{\pi_{1}},1\right\}\right].
\end{equation}

Symmetric arguments can be used to partially identify the outcome distribution for AO units in the control group. Table \ref{tab:distribution_AO} summarizes the resulting bounds for the outcome distribution of AO units across all periods and groups.
\begin{table}[H] \centering
\caption{Bounds for the outcome distribution for Always-Observed units.}
\begin{tabular}{|c|c|c|}
\hline
\textbf{Distribution}  & \textbf{Lower Bound} & \textbf{Upper Bound} \\ \hline 
   $F_{Y_1 \mid G=1, V = AO}(y)$ & $\max\left\{0,\frac{F_{Y_{1} \mid G = 1,S_{2} = 1}(y) - (1- \pi_{1})}{\pi_{1}}\right\}$  & $\min\left\{\frac{F_{Y_{1} \mid G = 1,S_{2}= }(y)}{\pi_{1}},1\right\}$ \\
    $F_{Y_2 \mid G=1, V = AO}(y) $ & $\max\left\{0,\frac{F_{Y_{2} \mid G = 1, S_2 = 1}(y) - (1- \pi_{1})}{\pi_{1}}\right\}$  & $\min\left\{\frac{F_{Y_{2} \mid G = 1, S_2 = 1}(y)}{\pi_{1}},1\right\}$ \\
    
    $F_{Y_1 \mid G=0, V = AO}(y)$ & $\max\left\{0,\frac{F_{Y_{1} \mid G = 0,S_{2} = 1}(y) - (1- \pi_{0})}{\pi_{0}}\right\}$  & $\min\left\{\frac{F_{Y_{1} \mid G = 0,S_{2}= }(y)}{\pi_{0}},1\right\}$ \\
    $F_{Y_2 \mid G=0, V = AO}(y) $ & $\max\left\{0,\frac{F_{Y_{2} \mid G = 0, S_2 = 1}(y) - (1- \pi_{0})}{\pi_{0}}\right\}$  & $\min\left\{\frac{F_{Y_{2} \mid G = 0, S_2 = 1}(y)}{\pi_{0}},1\right\}$ \\
    \hline
\end{tabular}
\label{tab:distribution_AO}
\begin{minipage}{0.81\linewidth  \setstretch{0.75}}
{\scriptsize  Notes: Sharp bounds on the outcome distribution for the Always-Observed units in all periods ans groups. $F_{Y_{t} \mid G=g, V= AO}(y)$ denotes the c.d.f. of the outcome at period $t$ in group $g$ for the AO principal stratum. $F_{Y_{t}\mid G= g, S_2 = 1}(y)$ denotes the c.d.f. of the outcome at period $t$ in group $g$ for units observed at $t = 2$. Results follow from Lemma \ref{prop:HMprop2}.}
 \end{minipage}
\end{table}

\begin{lemma} \label{prop:HMprop3}
    Let $Y$ be a continuous random variable and a mixture of two random variables, $Y_{0}$ and $Y_{1}$, with known mixing probability, $p$. The c.d.f. of the observed $Y$ is $F(y) = (1-p)F_{0}(y) + p F_{1}(y)$, where $F_{0}$ and $F_{1}$ are the c.d.f.s of $Y_{0}$ and $Y_{1}$ respectively. The $q$ quantile of the distribution $F_{1}$ is bounded as follows:
    \begin{equation*}
        Q_{1}^{}(q) \in \left[Q^{}(qp), Q^{}(qp+1-p)\right],
    \end{equation*}
    where $Q(\cdot)$ is the quantile function, $$
    Q(q) = \inf\{y: F(y) \geq q\}.$$
    These bounds are sharp.\\
    Proof: See \textcite[Proposition 3]{horowitz_identification_1995}.
\end{lemma}

As shown in equations \eqref{eq:proof_mixture} and \eqref{eq:proof_mixture_bi}, I can write the observed outcome in the treatment group as a mixture of the distributions for AO and OT units. Therefore, by Lemma \ref{prop:HMprop3}, the quantile functions for the AO units are bounded as follows:
\begin{equation}
    Q_{Y_{2} \mid G=1, V = AO}(q) \in \left[Q_{Y_{2}\mid G=1, S_2 = 1}(q \pi_{1})\quad ,\quad Q_{Y_{2}\mid G=1, S_2 = 1}(q\pi_{1} + 1-\pi_{1}) \right],
\end{equation}
\begin{equation}
     Q_{Y_{1} \mid G=1, V = AO}(q) \in \left[Q_{Y_{1}\mid G=1, S_2 = 1}(q \pi_{1})\quad ,\quad Q_{Y_{1}\mid G=1, S_2 = 1}(q\pi_{1} + 1-\pi_{1}) \right].
\end{equation}
Note that, because both $q$ and $\pi_1$ lie in the unit interval, it follows that $0 \leq q\pi_1 \leq 1$ and $0 \leq q\pi_1 + 1- \pi_1 \leq 1$.

Symmetric arguments can be used to partially identify the quantiles of the distribution for AO units in the control group. Table \ref{tab:distribution_quant_AO} summarizes the resulting bounds for the $q$ quantile of the outcome distribution for AO units across all periods and groups.
\begin{table}[H] \centering
\caption{Bounds for the quantile function for Always-Observed units.}
\begin{tabular}{|c|c|c|}
\hline
\textbf{Quantile function}  & \textbf{Lower Bound} & \textbf{Upper Bound} \\ \hline 
   $ Q_{Y_{1} \mid G=1, V = AO}(q)$ & $Q_{Y_{1}\mid G=1, S_2 = 1}(q \pi_{1})$  & $Q_{Y_{1}\mid G=1, S_2 = 1}(q\pi_{1} + 1-\pi_{1})$ \\
    $Q_{Y_{2} \mid G=1, V = AO}(q) $ & $Q_{Y_{2}\mid G=1,S_2 = 1}(q \pi_{1})$  & $Q_{Y_{2}\mid G=1,S_2=1}(q\pi_{1} + 1-\pi_{1}) $ \\
    
    $ Q_{Y_{1} \mid G=0, V = AO}(q)$ & $Q_{Y_{1}\mid G=0, S_2 = 1}(q \pi_{0})$  & $Q_{Y_{1}\mid G=0, S_2 = 1}(q\pi_{0} + 1-\pi_{0})$ \\
    $Q_{Y_{2} \mid G=0, V = AO}(q) $ & $Q_{Y_{2}\mid G=0, S_2 = 1}(q \pi_{0})$  & $Q_{Y_{2}\mid G=0, S_2 = 1}(q\pi_{0} + 1-\pi_{0}) $ \\ 
    \hline
\end{tabular}
\label{tab:distribution_quant_AO}
\begin{minipage}{0.81\linewidth  \setstretch{0.75}}
{\scriptsize  Notes: Sharp bounds on the quantiles of the outcome distribution for the Always-Observed units in all periods ans groups. $Q_{Y_{t}\mid G=g,S_2 = 1}(p) := \inf\{y : F_{Y_{t} \mid G=g, S_{2} = 1} \geq p\}$, and $F_{Y_{t} \mid G=g, S_{2} = 1}$ denotes the c.d.f. of the outcome at period $t$ in group $g$ for units observed in the post-treatment outcome. Results follow from Lemma \ref{prop:HMprop3}.}
 \end{minipage}
\end{table}

Now, recall the expression for the $\qttao(q)$, $$\qttao(q) = Q_{Y_{2}(1) \mid G = 1, V_{} = AO}^{}(q)-Q_{Y_{2}(0) \mid G = 1, V_{} = AO}^{}(q).$$ I can substitute the first term in the estimand by the realized quantile for AO units, $Q_{Y_{2} \mid G=1, V = AO}(q)$, which is partially identified as expressed in the second row of Table \ref{tab:distribution_quant_AO}. The second term, $Q_{Y_{2}(0) \mid G = 1, V_{} = AO}^{}(q)$, involves the quantile of the distribution of the missing potential outcome, and needs to be imputed.

\begin{lemma} \label{prop:AIprop3}
    Under Assumption \ref{as:outcome},
    \begin{equation*}
        F_{Y_{2}(0) \mid G = 1, V=AO}(y) = F_{Y_{1}\mid G=1,V=AO }\left(Q_{Y_{1}\mid G = 0,V=AO}^{}\left(F_{Y_{2}\mid G=0, V=AO}(y)\right)\right).
    \end{equation*}
    Proof: See \textcite[Theorem 3.1]{athey_identification_2006}.
\end{lemma}
Note that \textcite{athey_identification_2006} use different notation: the distribution $F_{Y_t \mid G = g}$ corresponds to $F_{Y, gt}$ in their notation, and they index time as $t \in \{0,1\}$ rather than $t \in \{1,2\}$. Furthermore, they write $Y(0)$ and $Y(1)$ as $Y^{N}$ and $Y^{I}$ respectively. For example, $F_{Y_1 \mid G = 1}$ is written as $F_{Y,10}$ in their notation. Finally, they write $Q(\cdot)$ as $F^{-1}(\cdot)$.

Notice that Assumption \ref{as:outcome} is the principal counterpart to assumptions 3.1-3.4 in \textcite{athey_identification_2006}, under which Theorem 3.1 is derived. However, as argued in \textcite[443]{athey_identification_2006}, by conditioning all the assumptions on a covariate $X$, Theorem 3.1 can be used to identify the distribution of $Y_{2}(0) \mid G=1, X$, where $X$ in this case is the latent principal stratum, $V$.

 Because $Y$ is continuous and has compact support, it follows that $F(Q(q)) = q$ and $Q(F(y)) = y$ for any $y$ in the support of $Y$. Then, from Lemma \ref{prop:AIprop3},
 \begin{equation}
     F_{Y_{2}(0) \mid G = 1, V=AO}(y) = q = F_{Y_{1}\mid G=1,V=AO }\left(Q_{Y_{1}\mid G = 0,V=AO}^{}\left(F_{Y_{2}\mid G=0, V=AO}(y)\right)\right) \end{equation}\begin{equation}
     Q_{Y_{1}\mid G=1,V=AO}(q) = Q_{Y_{1}\mid G=1,V=AO}\left(F_{Y_{1}\mid G=1,V=AO }\left(Q_{Y_{1}\mid G = 0,V=AO}^{}\left(F_{Y_{2}\mid G=0, V=AO}(y)\right)\right)\right) \end{equation}\begin{equation}
     Q_{Y_{1}\mid G=1,V=AO}(q) = Q_{Y_{1}\mid G = 0,V=AO}^{}\left(F_{Y_{2}\mid G=0, V=AO}(y)\right) \end{equation}\begin{equation}
     F_{Y_{1}\mid G=0, V=AO}\left(Q_{Y_{1}\mid G=1,V=AO}(q)\right) = F_{Y_{1}\mid G=0, V=AO}\left(Q_{Y_{1}\mid G = 0,V=AO}^{}\left(F_{Y_{2}\mid G=0, V=AO}(y)\right)\right)
 \end{equation}\begin{equation}
     F_{Y_{1}\mid G=0, V=AO}\left(Q_{Y_{1}\mid G=1,V=AO}(q)\right) =F_{Y_{2}\mid G=0, V=AO}(y)
 \end{equation}\begin{equation}
     Q_{Y_{2} \mid G=0, V=AO}\left(F_{Y_{1}\mid G=0, V=AO}\left(Q_{Y_{1}\mid G=1,V=AO}(q)\right) \right)=Q_{Y_2\mid G=0, V=AO}\left(F_{Y_{2}\mid G=0, V=AO}(y)\right)
 \end{equation}\begin{equation}
     Q_{Y_{2} \mid G=0, V=AO}\left(F_{Y_{1}\mid G=0, V=AO}\left(Q_{Y_{1}\mid G=1,V=AO}(q)\right) \right)=y = Q_{Y_{2}(0) \mid G = 1, V=AO}(q),
 \end{equation}
 where the last equality comes from
 \begin{equation}
     F_{Y_{2}(0) \mid G = 1, V=AO}(y) = q 
 \end{equation}
 \begin{equation}
     Q_{Y_{2}(0) \mid G = 1, V=AO}\left(F_{Y_{2}(0) \mid G = 1, V=AO}(y)\right) = Q_{Y_{2}(0) \mid G = 1, V=AO}(q)
 \end{equation}
 \begin{equation}
     y = Q_{Y_{2}(0) \mid G = 1, V=AO}(q).
 \end{equation}
 
Hence, I can replace the potential outcomes in the Causal Estimand \ref{es:qttao} by observed outcomes as follows:
\begin{equation}
    \qttao(q) = Q_{Y_{2} \mid G=1, V = AO}(q) -  Q_{Y_{2} \mid G=0, V=AO}\left(F_{Y_{1}\mid G=0, V=AO}\left(Q_{Y_{1}\mid G=1,V=AO}(q)\right) \right), \label{eq:ap_qttao_observed}
\end{equation}
which is equivalent to equation (18) in \textcite{athey_identification_2006} conditioning on $V = AO$. All these quantities can be partially identified as noted above in Tables \ref{tab:distribution_AO} and \ref{tab:distribution_quant_AO}. Therefore, the partial identification of $\qttao(q)$ follows. Specifically, the upper (lower) bound of the $\qttao(q)$ is given by the upper (lower) bound of $ Q_{Y_{2} \mid G=1, V = AO}(q)$ and the lower (upper) bound of $Q_{Y_{2} \mid G=0, V=AO}\left(F_{Y_{1}\mid G=0, V=AO}\left(Q_{Y_{1}\mid G=1,V=AO}(q)\right) \right)$. The bounds of the former can be found in Table \ref{tab:distribution_quant_AO}. I derive the bounds of the latter next.

First, I derive the lower bound of $Q_{Y_{2} \mid G=0, V=AO}\left(F_{Y_{1}\mid G=0, V=AO}\left(Q_{Y_{1}\mid G=1,V=AO}(q)\right) \right)$. From Lemma \ref{prop:HMprop3}, as shown in Table \ref{tab:distribution_quant_AO}, the lower bound of $Q_{Y_2\mid G=0, V=AO}(q)$ is given by $Q_{Y_2 \mid  G=0, S_2 =1 }(q \pi_0)$. So it follows that
\begin{equation}
\begin{aligned}
    Q_{Y_{2} \mid G=0, V=AO}\left(F_{Y_{1}\mid G=0, V=AO}\left(Q_{Y_{1}\mid G=1,V=AO}(q)\right) \right) \geq \\ Q_{Y_{2} \mid G=0,S_2 = 1}\left(F_{Y_{1}\mid G=0, V=AO}\left(Q_{Y_{1}\mid G=1,V=AO}(q)\right) \pi_0\right).
    \end{aligned}
\end{equation}

From Lemma \ref{prop:HMprop2}, as shown in Table \ref{tab:distribution_AO}, the lower bound of $F_{Y_{1}\mid G=0, V=AO}(y)$ is given by $\max\left\{0,\frac{F_{Y_{1}\mid G=0, S_2 = 1}(y) - (1- \pi_0)}{\pi_0}\right\}$. Then, it follows that:
\begin{equation}
\begin{aligned}
    Q_{Y_{2} \mid G=0, V=AO}\left(F_{Y_{1}\mid G=0, V=AO}\left(Q_{Y_{1}\mid G=1,V=AO}(q)\right) \right) \geq \\
    Q_{Y_{2} \mid G=0,S_2 = 1}\left(\max \left\{0, \frac{F_{Y_{1}\mid G=0, S_2 = 1}\left(Q_{Y_{1}\mid G=1,V=AO}(q)\right)-(1-\pi_0)}{\pi_{0}} \right\} \pi_0\right).
\end{aligned}
\end{equation}

From Lemma \ref{prop:HMprop3}, as shown in Table \ref{tab:distribution_quant_AO}, the lower bound of $Q_{Y_{1}\mid G=1,V=AO}(q)$ is given by $Q_{Y_{1}\mid G=1,S_2 = 1}(q\pi_1)$. So it follows that 
\begin{equation}
\begin{aligned}
    Q_{Y_{2} \mid G=0, V=AO}\left(F_{Y_{1}\mid G=0, V=AO}\left(Q_{Y_{1}\mid G=1,V=AO}(q)\right) \right) \geq \\
    Q_{Y_{2} \mid G=0,S_2 = 1}\left(\max \left\{0, \frac{F_{Y_{1}\mid G=0, S_2 = 1}\left(Q_{Y_{1}\mid G=1,S_{2}=1}(q\pi_1)\right)-(1-\pi_0)}{\pi_{0}} \right\} \pi_0\right).
\end{aligned}
\end{equation}

Because I am assuming that $Y$ has compact support and the support of $Y_{t} \mid G=1 $ is contained in the support of $Y_{t} \mid G=0$, then $F_{Y_1 \mid G = 0, S_2 = 1}\left(Q_{Y_1 \mid G=1, S_2 = 1}(q)\right)$ is identified for any $q$. Because both $q$ and $\pi_g$ lie in the unit interval, their product also lies in the unit interval. All in all, the lower bound for $Q_{Y_{2} \mid G=0, V=AO}\left(F_{Y_{1}\mid G=0, V=AO}\left(Q_{Y_{1}\mid G=1,V=AO}(q)\right) \right)$ is given by:
\begin{equation} \label{eq:ap_lb_rhs}
    \begin{cases}
        Q_{Y_{2} \mid G=0,S_2 = 1}\left(F_{Y_{1}\mid G=0, S_2 = 1}\left(Q_{Y_{1}\mid G=1,S_2=1}(q\pi_1)\right)-(1-\pi_0)\right) \quad \quad \text{if} \quad \quad F_{Y_1 \mid G=0, S_2 = 1}\left(Q_{Y_1\mid G=1, S_2 = 0} (q\pi_1)\right) \geq 1 - \pi_0  \\
        Q_{Y_{2} \mid G=0,S_2 = 1}\left(0\right) \hspace{7.4cm}\text{otherwise} 
    \end{cases}
\end{equation}

I can follow the same steps to derive the upper bound of $Q_{Y_{2} \mid G=0, V=AO}\left(F_{Y_{1}\mid G=0, V=AO}\left(Q_{Y_{1}\mid G=1,V=AO}(q)\right) \right)$:
\begin{equation}
    \begin{aligned}
        Q_{Y_{2} \mid G=0, V=AO}\left(F_{Y_{1}\mid G=0, V=AO}\left(Q_{Y_{1}\mid G=1,V=AO}(q)\right) \right) \leq \\
        Q_{Y_{2} \mid G=0, S_2 = 1}\left(F_{Y_{1}\mid G=0, V=AO}\left(Q_{Y_{1}\mid G=1,V=AO}(q)\right)\pi_0 + 1 - \pi_0 \right)
    \end{aligned}
\end{equation}
\begin{equation}
    \begin{aligned}
        Q_{Y_{2} \mid G=0, V=AO}\left(F_{Y_{1}\mid G=0, V=AO}\left(Q_{Y_{1}\mid G=1,V=AO}(q)\right) \right) \leq \\
        Q_{Y_{2} \mid G=0, S_2 = 1}\left(\min\left\{\frac{F_{Y_{1}\mid G=0, S_2 = 1}\left(Q_{Y_{1}\mid G=1,V=AO}(q)\right)}{\pi_0},1\right\}\pi_0 + 1 - \pi_0 \right)
    \end{aligned}
\end{equation}

\begin{equation}
    \begin{aligned}
        Q_{Y_{2} \mid G=0, V=AO}\left(F_{Y_{1}\mid G=0, V=AO}\left(Q_{Y_{1}\mid G=1,V=AO}(q)\right) \right) \leq \\
        Q_{Y_{2} \mid G=0, S_2 = 1}\left(\min\left\{\frac{F_{Y_{1}\mid G=0, S_2 = 1}\left(Q_{Y_{1}\mid G=1,S_2 = 1}(q\pi_1 + 1 - \pi_1)\right)}{\pi_0},1\right\}\pi_0 + 1 - \pi_0 \right).
    \end{aligned}
\end{equation}
So the upper bound of $Q_{Y_{2} \mid G=0, V=AO}\left(F_{Y_{1}\mid G=0, V=AO}\left(Q_{Y_{1}\mid G=1,V=AO}(q)\right) \right)$ is given by:
\begin{equation} \label{eq:ap_ub_rhs}
    \begin{cases}
        Q_{Y_{2} \mid G=0, S_2 = 1}\left(F_{Y_{1}\mid G=0, S_2 = 1}\left(Q_{Y_{1}\mid G=1,S_2 = 1}(q\pi_1 + 1 - \pi_1)\right) + 1 - \pi_0 \right)   \text{ if }   F_{Y_{1}\mid G=0, S_2 = 1}\left(Q_{Y_{1}\mid G=1,S_2 = 1}(q\pi_1 + 1 - \pi_1)\right) \leq \pi_0 \\
        Q_{Y_{2} \mid G=0, S_2 = 1}(1) \hspace{8cm}  \text{otherwise}
    \end{cases}
\end{equation}
To conclude, equation \eqref{eq:ap_qttao_observed} presents the $\qttao(q)$ as a function of the distributions of observed outcomes for the AO units. These distributions are partially identified. The upper bound of the $\qttao(q)$ is given by the upper bound of the first term (second row in Table \ref{tab:distribution_quant_AO}) and the lower bound of the second term (equation \eqref{eq:ap_lb_rhs}). The lower bound of the $\qttao(q)$ is given by the lower bound of the first term (second row in Table \ref{tab:distribution_quant_AO}) and the upper bound of the second term (equation \eqref{eq:ap_ub_rhs}). That is,
\begin{equation}
    \Lambda^{LB}(q) = Q_{Y_2 \mid G=1, S_{2} = 1}(q\pi_1) - Q_{Y_{2} \mid G=0, S_2 = 1}\left(F_{Y_{1}\mid G=0, S_2 = 1}\left(Q_{Y_{1}\mid G=1,S_2 = 1}(q\pi_1 + 1 - \pi_1)\right) + 1 - \pi_0 \right)
\end{equation}
if $$F_{Y_{1}\mid G=0, S_2 = 1}\left(Q_{Y_{1}\mid G=1,S_2 = 1}(q\pi_1 + 1 - \pi_1)\right) \leq \pi_0$$ and 
\begin{equation}
    \Lambda^{LB}(q) = Q_{Y_2 \mid G=1, S_{2} = 1}(q\pi_1) - Q_{Y_{2} \mid G=0, S_2 = 1}\left(1\right)
\end{equation}
otherwise, and 
\begin{equation}
    \Lambda^{UB}(q)= Q_{Y_2 \mid G=1, S_{2} = 1}(q\pi_1 + 1 - \pi_1) -  Q_{Y_{2} \mid G=0,S_2 = 1}\left(F_{Y_{1}\mid G=0, S_2 = 1}\left(Q_{Y_{1}\mid G=1,S_2=1}(q\pi_1)\right)-(1-\pi_0)\right)
\end{equation}
if $$F_{Y_1 \mid G=0, S_2 = 1}\left(Q_{Y_1\mid G=1, S_2 = 0} (q\pi_1)\right) \geq 1 - \pi_0 $$ and
\begin{equation}
     \Lambda^{UB}(q)= Q_{Y_2 \mid G=1, S_{2} = 1}(q\pi_1 + 1 - \pi_1) -  Q_{Y_{2} \mid G=0,S_2 = 1}\left(0\right)
\end{equation}
otherwise.
\end{proof}

\subsection{Proof of Remark \ref{prop:remark5}} \label{proof:remark5}
The treatment effect on selection, $\E[S_{i2}(1) - S_{i2}(0) \mid G_{i} = g]$ will have the same sign for both groups $g \in \{0,1\}$.

\begin{proof}
    Assume without loss of generality that the treatment effect on selection for the treated units is positive:
    \begin{align}
        \E[S_{i2}(1) - S_{i2}(0) \mid G_{i} = 1] & > 0 \\
        \E[S_{i2} \mid G_{i} = 1] -\E[S_{i1} \mid G_{i} = 1]\frac{\E[S_{i2} \mid G_{i} =0]}{\E[S_{i1} \mid G_{i} = 0]} & > 0 \\
         \E[S_{i2} \mid G_{i} = 1] & > \E[S_{i1} \mid G_{i} = 1]\frac{\E[S_{i2} \mid G_{i} =0]}{\E[S_{i1} \mid G_{i} = 0]} \\
        \E[S_{i1} \mid G_{i} = 0] \frac{ \E[S_{i2} \mid G_{i} = 1]}{ \E[S_{i1} \mid G_{i} = 1]} & > \E[S_{i2} \mid G_{i} =0] \\ 
         \E[S_{i1} \mid G_{i} = 0] \frac{ \E[S_{i2} \mid G_{i} = 1]}{ \E[S_{i1} \mid G_{i} = 1]} - \E[S_{i2} \mid G_{i} =0]  & > 0 \\
         \E[S_{i2}(1) - S_{i2}(0) \mid G_{i} = 0] & >0 
    \end{align}
    Then, the treatment effect on selection for the control units is also positive.
\end{proof}
\\\\
$\E[S_{i1} \mid G_{i}= 1] = \E[S_{i1} \mid G_{i} = 0] \implies \E[S_{i2}(1) - S_{i2}(0) \mid G_{i} = 0] = \E[S_{i2}(1) - S_{i2}(0) \mid G_{i} = 1]$.

\begin{proof}
        \begin{align}
            \E[S_{i2}(1) - S_{i2}(0) \mid G_{i} = 1] & =  \E[S_{i2} \mid G_{i} = 1] -\E[S_{i1} \mid G_{i} = 1]\frac{\E[S_{i2} \mid G_{i} =0]}{\E[S_{i1} \mid G_{i} = 0]} \\
            & = \E[S_{i2} \mid G_{i} = 1]\frac{\E[S_{i1} \mid G_{i} = 1]}{\E[S_{i1} \mid G_{i} = 1]} -\E[S_{i1} \mid G_{i} = 1]\frac{\E[S_{i2} \mid G_{i} =0]}{\E[S_{i1} \mid G_{i} = 0]} \\
            & = \E[S_{i2} \mid G_{i} = 1]\frac{\E[S_{i1} \mid G_{i} = 0]}{\E[S_{i1} \mid G_{i} = 1]} -\E[S_{i1} \mid G_{i} = 0]\frac{\E[S_{i2} \mid G_{i} =0]}{\E[S_{i1} \mid G_{i} = 0]} \\
            & =  \E[S_{i1} \mid G_{i} = 0] \frac{ \E[S_{i2} \mid G_{i} = 1]}{ \E[S_{i1} \mid G_{i} = 1]} - \E[S_{i2} \mid G_{i} =0]  \\
            & =  \E[S_{i2}(1) - S_{i2}(0) \mid G_{i} = 0]
        \end{align}
\end{proof}

It can also be shown that there are only two cases when the treatment effect on selection is the same in both groups: when the baseline selection is the same across groups, as I just showed, and when the treatment effect is zero. 

\subsection{Proof of Proposition \ref{prop:proportions}} \label{proof:proportions}
Below, I prove Proposition \ref{prop:proportions} under positive monotonicity. The prove under negative monotonicity follows by symmetry.

\begin{proof}
    \begin{equation}
        \pi_{0} = Pr(S_{i2}(1) = 1 \mid G_{i} = 0, S_{i2}(0) = 1) = 1 \label{eq:proof_pi0}
    \end{equation}
    Equation (\ref{eq:proof_pi0}) follows immediately from the positive monotonicity assumption.
    \begin{align}
        \pi_{1} & = Pr(S_{i2}(0) = 1 \mid G_{i} =1, S_{i2}(1) = 1) \\
        &= \frac{Pr(S_{i2}(1) = 1 \mid G_{i} = 1, S_{i2}(0) = 1) Pr(S_{i2}(0) = 1 \mid G_{i} = 1)}{Pr(S_{i2}(1) = 1 \mid G_{i} = 1)} \label{eq:proof_pi1bat} \\
        &= \frac{Pr(S_{i2}(0) = 1 \mid G_{i} =1)}{Pr(S_{i2}(1) = 1 \mid G_{i} = 1)} \label{eq:proof_pi1bi}\\
        &= \frac{\E[S_{i2}(0) \mid G_{i} = 1]}{\E[S_{i2}\mid G_{i} = 1]} \\
        &=\frac{\E[S_{i1} \mid G_{i} = 1]}{\E[S_{i2} \mid G_{i} = 1]}\frac{\E[S_{i2} \mid G_{i} = 0]}{\E[S_{i1} \mid G_{i} = 0]} \label{eq:proof_pi1hiru}
    \end{align}
    Where (\ref{eq:proof_pi1bat}) follows from Bayes' theorem, (\ref{eq:proof_pi1bi}) follows from the positive monotonicity assumption and (\ref{eq:proof_pi1hiru}) follows from Lemma \ref{prop:selection}.
\end{proof}

\subsection{Proof of Lemma \ref{prop:strata_multiple_sources}} \label{proof:strata_multiple_sources}
\begin{proof}
    \begin{align}
            \mathcal{AO} & = \{i: S_{i2}(0) = 1, S_{i2}(1) = 1\} = \{i: \prod_{j}s_{i2}^{j}(0) = 1, \prod_{j}s_{i2}^{j}(1) = 1 \}           = \{i: s_{i2}^{j}(0) = s_{i2}^{j}(1) = 1 \quad \forall j\} \notag \\
            & = \{i: V_{i}^{j} = AO \quad \forall j\} = \{i: i\in \mathcal{AO}_{j} \quad \forall j\} = \bigcap\limits_{j} \mathcal{AO}_{j} \\
             \mathcal{NO} &= \{i: S_{i2}(0) = 0, S_{i2}(1) = 0\} = \{i:\prod_{j}s_{i2}^{j}(0)=0,\prod_{j}s_{i2}^{j}(1)=0\} \notag \\ 
        &= \{i: s_{i2}^{k}(0) = 0 \text{ for one }k\in J,s_{i2}^{\kappa}(1) \text{ for one }\kappa\in J\} \label{eq:proof_nobat} \\
        &= \{i:s_{i2}^{k}(0) = 0 ,s_{i2}^{k}(1) = 0 \text{ for one }k\in J\} = \{i: i \in NO_{k}\text{ for one }k\in J\} = \bigcup\limits_{j} \mathcal{NO}_{j} \label{eq:proof_nobi} \\
        \mathcal{OT} & = \{i:S_{i2}(0) = 0, S_{i2}(1) = 1\} = \{i:\prod_{j}s_{i2}^{j}(0)=0,\prod_{j}s_{i2}^{j}(1)=1\} \notag \\
        & \{i: s_{i2}^{k}(0) = 0 \text{ for one }k \in J, s_{i2}^{j}(0) = 1 \quad \forall j\neq k, \quad s_{i2}^{j}(1) = 1 \quad \forall j\} \label{eq:proof_otbat}\\
        &=\{i: i\in \mathcal{OT}_{k}\text{ for one }k\in J, i \in \mathcal{AO}_{j} \quad \forall j \neq k\} = \bigcup\limits_{j} \mathcal{OT}_{j} \label{eq:proof_otbi} \\
    \mathcal{OC} & = \{i:S_{i2}(0) = 1, S_{i2}(1) = 0\} = \{i:\prod_{j}s_{i2}^{j}(0)=1,\prod_{j}s_{i2}^{j}(1)=0\} \notag  \\ 
        & \{i: s_{i2}^{j}(0) = 1 \quad \forall j, \quad s_{i2}^{k}(1) = 0\text{ for one }k \in J, s_{i2}^{j}(1) = 1 \quad \forall j\neq k \quad \forall j\} \label{eq:proof_ocbat}\\
        &=\{i: i\in \mathcal{OC}_{k}\text{ for one }k\in J, i \in \mathcal{AO}_{j} \quad \forall j \neq k\} = \bigcup\limits_{j} \mathcal{OC}_{j} \label{eq:proof_ocbi}
        \end{align}
    where (\ref{eq:proof_nobat}), (\ref{eq:proof_otbat}) and (\ref{eq:proof_ocbat}) follow from the sources being mutually exclusive and (\ref{eq:proof_nobi}) from the no intersection of OC and OT strata. As shown in the text, when a unit belongs to the NO stratum according to one source, then it belongs to the AO stratum according to all the remaining sources. It suffices to belong to the NO stratum according to one source to belong to the NO stratum overall. Therefore, the NO stratum is given by the union of NO strata according to all the sources. The same reasoning can be used for OC and OT strata.

    These sets constitute a partition of the sample. Consider any source $j$. All the units must belong to one of the four strata defined by the source $j$, i.e., $i \in \{\mathcal{AO}_{j},\mathcal{NO}_{j},\mathcal{OT}_{j},\mathcal{OC}_{j}\}$. If a unit belongs to the NO, OT or OC stratum according to this source, then mutual exclusivity and no intersection of OC and OT strata imply that this unit belongs to the AO stratum for all the other sources. This implies that no unit belongs to more than one set in $\{\mathcal{AO}_{},\mathcal{NO}_{},\mathcal{OT}_{},\mathcal{OC}_{}\}$ simultaneously. Furthermore, the union of these sets equal to the whole sample by construction, so they constitute a partition of the sample.
\end{proof}

\subsection{Proof of Proposition \ref{prop:proportions_multiple_sources}} \label{proof:proportions_multiple_sources}
For illustration purposes, I derive first the expression for $\pi_{1}$ with $J = 2$. For $j = 1$ positive monotonicity holds while for $j = 2 $ negative monotonicity holds. Following the same logic, I prove then the general case for both $\pi_{1}$ and $\pi_{0}$.

\begin{proof}
    Consider the case with $J=2$, where positive monotonicity holds for $j=1$ ($s_{i2}^{1}(1) \geq s_{i2}^{1}(0)$) and negative monotonicity holds for $j=2$ ($s_{i2}^{2}(1) \leq s_{i2}^{2}(0)$).
    \begin{align}
        \pi_{1} &= Pr(S_{i2}(0) = 1 \mid G_{i} = 1, S_{i2}(1) = 1) \\
        & = \frac{Pr(S_{i2}(0) = 1, S_{i2}(1) = 1\mid G_{i} = 1)}{Pr(S_{i2}(1) = 1 \mid G_{i} = 1)} \\
        &= \frac{1
        -Pr(S_{i2}(0) = 0, S_{i2}(1) = 1\mid G_{i} = 1)
        - Pr(S_{i2}(0) = 1 , S_{i2}(1) = 0 \mid G_{i} = 1)
        -Pr(S_{i2}(0) =0 , S_{i2}(1) =0 \mid G_{i} = 1)
        }{\E[S_{i2} \mid G_{i} = 1]} \\
        &= \frac{1
         -Pr(s_{i2}^{1}(0) = 0, s_{i2}^{1}(1) = 1\mid G_{i} = 1)
        - Pr(s_{i2}^{2}(0) = 1 , s_{i2}^{2}(1) = 0 \mid G_{i} = 1)
        -Pr(S_{i2}(0) =0 , S_{i2}(1) =0 \mid G_{i} = 1)
        }{\E[S_{i2} \mid G_{i} = 1]} \label{eq:proof_mpbat} \\
        &= \frac{1
         -Pr(s_{i2}^{1}(0) = 0, s_{i2}^{1}(1) = 1\mid G_{i} = 1)
        - Pr(s_{i2}^{2}(0) = 1 , s_{i2}^{2}(1) = 0 \mid G_{i} = 1)
        }{\E[S_{i2} \mid G_{i} = 1]} \notag  \\
        &- \frac{Pr(s_{i2}^{1}(0) = 0, s_{i2}^{1}(1) = 0 \mid G_{i} = 1)+Pr(s_{i2}^{2}(0) = 0, s_{i2}^{2}(1) = 0 \mid G_{i} = 1)}{\E[S_{i2} \mid G_{i} = 1]} \\
        & + \underbrace{\frac{ Pr(s_{i2}^{1}(0) = 0, s_{i2}^{1}(1) = 0,s_{i2}^{2}(0) = 0, s_{i2}^{2}(1) = 0 \mid G_{i} = 1)}{\E[S_{i2} \mid G_{i} = 1]}}_{ 0 \text{ (by mutual exclusivity)}} \notag \\
        & = \frac{1 - Pr(s_{i2}^{1}(0) = 0 \mid G_{i}=1) - Pr(s_{i2}^{2}(1) = 0 \mid G_{i}=1) }{\E[S_{i2} \mid G_{i} = 1]} \\
        &= \frac{1  - \left(1 - \E[s_{i2}^{1}(0) \mid G_{i} = 1]\right) - \left(1 - \E[s_{i2}^{2}\mid G_{i} = 1]\right)}{\E[S_{i2} \mid G_{i} = 1]} \\ 
        &= \frac{1}{\E[S_{i2} \mid G_{i} = 1]}\left(1 - \left(1 - \E[s_{i1}^{1} \mid G_{i} = 1]\frac{\E[s_{i2}^{1} \mid G_{i} = 0]}{\E[s_{i1}^{1} \mid G_{i} = 0]}\right)- \bigg(1 - \E[s_{i2}^{2}\mid G_{i} = 1]\bigg) \right)       \label{eq:proof_mpbi}
    \end{align}
    where (\ref{eq:proof_mpbat}) follows from the monotonicity assumptions and (\ref{eq:proof_mpbi}) from the selection model in Assumption \ref{as:sele} applied to each source of attrition.
\end{proof}

\begin{proof}
    Consider the generic case with $J$ different sources. Let $J^{+}$ denote the set of sources for which positive monotonicity holds and $J^{-}$ the set of sources for which negative monotonicity holds.

    Mutual exclusivity of the attrition sources implies:
    \begin{equation}
        Pr(V_{i} = NO) = \Pr\left(i \in \bigcup\limits_{j}\mathcal{NO}_{j}\right) = \sum_{j}Pr\left(s_{i2}^{j}(0) = 0, s_{i2}^{j}(1) = 0\right) \label{eq:proof_mebat}
    \end{equation}
    as
    \begin{equation*}
        \mathcal{NO}_{j} \bigcap \mathcal{NO}_{k} = \varnothing \quad \forall j,k \in J.
    \end{equation*}
    Similarly,
    \begin{align}
         Pr(V_{i} = OT) &= \Pr\left(i \in \bigcup\limits_{j}\mathcal{OT}_{j}\right) = \sum_{j}Pr\left(s_{i2}^{j}(0) = 0, s_{i2}^{j}(1) = 1\right) \label{eq:proof_mebi} \\
          Pr(V_{i} = OC) &= \Pr\left(i \in \bigcup\limits_{j}\mathcal{OC}_{j}\right) = \sum_{j}Pr\left(s_{i2}^{j}(0) = 1, s_{i2}^{j}(1) = 0\right) \label{eq:proof_mehiru}
    \end{align}

    Then,
    \begin{align}
        \pi_{1} & = Pr(S_{i2}(0) = 1 \mid G_{i} = 1, S_{i2}(1) = 1) \\
        & = \frac{Pr(S_{i2}(0) = 1, S_{i2}(1) = 1\mid G_{i} = 1)}{Pr(S_{i2}(1) = 1 \mid G_{i} = 1)} \\
        &= \frac{1
        -Pr(V_{i} = OT\mid G_{i} = 1)
        - Pr(V_{i} = OC \mid G_{i} = 1)
        -Pr(V_{i} = NO \mid G_{i} = 1)
        }{\E[S_{i2} \mid G_{i} = 1]} \\
        &= \frac{1
        -\sum_{j}Pr\left(s_{i2}^{j}(0) = 0, s_{i2}^{j}(1) = 1\mid G_{i} = 1\right)
        -\sum_{j}Pr\left(s_{i2}^{j}(0) = 1, s_{i2}^{j}(1) = 0\mid G_{i} = 1\right)
        }{\E[S_{i2} \mid G_{i} = 1]} \notag \\
        & - \frac{\sum_{j}Pr\left(s_{i2}^{j}(0) = 0, s_{i2}^{j}(1) = 0 \mid G_{i}=1\right)
        }{\E[S_{i2} \mid G_{i} = 1]} \label{eq:proof_msbat} \\
        &= \frac{1
        -\sum_{j \in J^{+}}Pr\left(s_{i2}^{j}(0) = 0, s_{i2}^{j}(1) = 1\mid G_{i} = 1\right)
        -\sum_{j \in J^{-}}Pr\left(s_{i2}^{j}(0) = 1, s_{i2}^{j}(1) = 0\mid G_{i} = 1\right)
        }{\E[S_{i2} \mid G_{i} = 1]} \notag \\
        & - \frac{\sum_{j}Pr\left(s_{i2}^{j}(0) = 0, s_{i2}^{j}(1) = 0 \mid G_{i}=1\right)
        }{\E[S_{i2} \mid G_{i} = 1]} \label{eq:proof_msbi} \\
         &= \frac{1
        -\sum_{j \in J^{+}}Pr\left(s_{i2}^{j}(0) = 0\mid G_{i} = 1\right)
        -\sum_{j \in J^{-}}Pr\left(s_{i2}^{j}(1) = 0\mid G_{i} = 1\right)
        }{\E[S_{i2} \mid G_{i} = 1]}  \\
        &= \frac{1
        -\sum_{j \in J^{+}}\left(1 - \E[s_{i2}^{j}(0)\mid G_{i} = 1]\right)
        -\sum_{j \in J^{-}}\left(1 - \E[s_{i2}^{j}\mid G_{i} = 1]\right)
        }{\E[S_{i2} \mid G_{i} = 1]}  \\
        &= \frac{1}{\E[S_{i2} \mid G_{i} = 1]}\left(1 - \sum_{j \in J^{+}}\left(1 - \E[s_{i1}^{j}\mid G_{i} = 1]\frac{\E[s_{i2}^{j} \mid G_{i} = 0]}{\E[s_{i1}^{j}\mid G_{i} = 0]}\right)
        -\sum_{j \in J^{-}}\left(1 - \E[s_{i2}^{j}\mid G_{i} = 1]\right)\right) ,\label{eq:proof_mshiru}
    \end{align}
    where equation (\ref{eq:proof_msbat}) follows from equations (\ref{eq:proof_mebat}-\ref{eq:proof_mehiru}), equation (\ref{eq:proof_msbi}) follows from the monotonicity assumptions and equation (\ref{eq:proof_mshiru}) follows from the selection model in Assumption \ref{as:sele} applied to each source of attrition. Similarly, 
    \begin{align}
         \pi_{0} & = Pr(S_{i2}(1) = 1 \mid G_{i} = 0, S_{i2}(0) = 1) \\
        & = \frac{Pr(S_{i2}(0) = 1, S_{i2}(1) = 1\mid G_{i} = 0)}{Pr(S_{i2}(0) = 1 \mid G_{i} = 0)} \\
        &= \frac{1
        -Pr(V_{i} = OT\mid G_{i} = 0)
        - Pr(V_{i} = OC \mid G_{i} = 0)
        -Pr(V_{i} = NO \mid G_{i} = 0)
        }{\E[S_{i2} \mid G_{i} = 0]} \\
        &= \frac{1
        -\sum_{j}Pr\left(s_{i2}^{j}(0) = 0, s_{i2}^{j}(1) = 1\mid G_{i} = 0\right)
        -\sum_{j}Pr\left(s_{i2}^{j}(0) = 1, s_{i2}^{j}(1) = 0\mid G_{i} = 0\right)
        }{\E[S_{i2} \mid G_{i} = 0]} \notag \\
        & - \frac{\sum_{j}Pr\left(s_{i2}^{j}(0) = 0, s_{i2}^{j}(1) = 0 \mid G_{i}=0\right)
        }{\E[S_{i2} \mid G_{i} = 0]}  \\
        &= \frac{1
        -\sum_{j \in J^{+}}Pr\left(s_{i2}^{j}(0) = 0, s_{i2}^{j}(1) = 1\mid G_{i} = 0\right)
        -\sum_{j \in J^{-}}Pr\left(s_{i2}^{j}(0) = 1, s_{i2}^{j}(1) = 0\mid G_{i} = 0\right)
        }{\E[S_{i2} \mid G_{i} = 0]} \notag \\
        & - \frac{\sum_{j}Pr\left(s_{i2}^{j}(0) = 0, s_{i2}^{j}(1) = 0 \mid G_{i}=0\right)
        }{\E[S_{i2} \mid G_{i} = 0]}  \\
         &= \frac{1
        -\sum_{j \in J^{+}}Pr\left(s_{i2}^{j}(0) = 0\mid G_{i} = 0\right)
        -\sum_{j \in J^{-}}Pr\left(s_{i2}^{j}(1) = 0\mid G_{i} = 0\right)
        }{\E[S_{i2} \mid G_{i} = 0]}  \\
        &= \frac{1
        -\sum_{j \in J^{+}}\left(1 - \E[s_{i2}^{j}\mid G_{i} = 0]\right)
        -\sum_{j \in J^{-}}\left(1 - \E[s_{i2}^{j}(1)\mid G_{i} = 0]\right)
        }{\E[S_{i2} \mid G_{i} = 0]}  \\
        &= \frac{1}{\E[S_{i2} \mid G_{i} = 0]}\left(1-
        \sum_{j \in J^{+}}\left(1 - \E[s_{i2}^{j}\mid G_{i} = 0]\right)
        -\sum_{j \in J^{-}}\left(1 - \E[s_{i2}^{j}\mid G_{i} = 0]\frac{\E[s_{i2}^{j}\mid G_{i} = 1]}{\E[s_{i1}^{j}\mid G_{i} = 1]}\right)\right) .
    \end{align}
\end{proof}

\subsection{Proof of Proposition \ref{prop:proportions} as a specific case of Proposition \ref{prop:proportions_multiple_sources}} \label{proof:equivalence}
\begin{proof}
Consider the case where $J = 1$. Without loss of generality, assume the unique source of attrition observable in the data, $S$, is such that $S_{i2}(1) \geq S_{i2}(0) \quad \forall i$, i.e., $J^{-} = \varnothing$. From Proposition \ref{prop:proportions_multiple_sources}:
\begin{align}
     \pi_{0} & = \frac{1}{\E[S_{i2} \mid G_{i} = 0]}\left(1-
        \left(1 - \E[S_{i2}^{}\mid G_{i} = 0]\right)\right) = \frac{\E[S_{i2} \mid G_{i} = 0]}{\E[S_{i2} \mid G_{i} = 0]} = 1 \\
        \pi_{1} & = \frac{1}{\E[S_{i2} \mid G_{i} = 1]}\left(1 - \left(1 - \E[S_{i1}^{}\mid G_{i} = 1]\frac{\E[S_{i2}^{} \mid G_{i} = 0]}{\E[S_{i1}^{}\mid G_{i} = 0]}\right)\right) = \frac{\E[S_{i1}^{}\mid G_{i} = 1]}{\E[S_{i2} \mid G_{i} = 0]}\frac{\E[S_{i2}^{} \mid G_{i} = 0]}{\E[S_{i1}^{}\mid G_{i} = 0]}
\end{align}
\end{proof}
\subsection{Proof of Proposition \ref{prop:normal_qttao}} \label{proof:normal_qttao}
This proof uses the following results from \textcite{athey_identification_2006}:
\begin{lemma} The CiC quantile estimator has an asymptotically linear representation. \\\label{prop:AItheo53}
    Let $\tau(q) = F_{Y_{2}(1) \mid G = 1}^{-1}(q) - F_{Y_{2}(0) \mid G=1}^{-1}(q)$. Let $\hat \tau(q)$ be CiC quantile estimator, $\hat \tau(q) = \widehat F_{Y_{2}(1) \mid G = 1}^{-1}(q) - \widehat{F}_{Y_{2} \mid G=0}^{-1}\left(\widehat F_{Y_{1} \mid G  = 0}\left(\widehat F_{Y_{1} \mid G=1}^{-1}(q)\right)\right)$. Then,
    \begin{equation*}
        \hat \tau(q) = \sum_{g,t}\hat \tau_{g,t}(q) + o_{p}(n^{-1/2})
    \end{equation*}
    Proof: See \textcite[Theorem 5.3]{athey_identification_2006}.
\end{lemma}
\begin{corollary}The CiC quantile estimator is asymptotically normal.  \label{prop:cicasdist}
        \begin{equation*}
        \sqrt{n}(\hat \tau(q) - \tau(q)) \xrightarrow{d} \mathcal{N}(0, V_{q}^{12}+ V_{q}^{02}+ V_{q}^{01}+ V_{q}^{11})
    \end{equation*}
    Proof: See \textcite[Theorem 5.3]{athey_identification_2006}.
\end{corollary}
In general, the asymptotic variance of the CiC quantile estimator is difficult to interpret. However, it is important to note that it is a function of the quantile $q$, and that the variance is also linear in $g$ and $t$. 
\begin{lemma} The estimated proportions are asymptotically normal. \label{prop:AItheo54}
    \begin{align*}
        \sqrt{n}(\hat \pi_{1} - \pi_{1}) \xrightarrow{d} \mathcal{N}(0, V_{\pi_{1}}) \\
        \sqrt{n}(\hat \pi_{0} - \pi_{0}) \xrightarrow{d} \mathcal{N}(0, V_{\pi_{0}}) 
    \end{align*}
    Proof: See \textcite[Theorem 5.4]{athey_identification_2006}.
\end{lemma}

\begin{proof}
    I will derive the asymptotic distribution for $\widehat{ \Lambda^{LB}}(q)$. The same steps could be followed to derive the asymptotic distribution of $\widehat{ \Lambda^{UB}}(q)$.

    First, denote with $$\tilde q  := F_{Y_{1}\mid G=0, S_2 = 1}\left(Q_{Y_{1}\mid G=1,S_2 = 1}(q\pi_1 + 1 - \pi_1)\right) +1 - \pi_{0}$$
    and $$q^{*}_{LB} := \min\left\{\tilde q, 1\right\}$$

    The structure of the proof is as follows. First, I show that $\hat q_{LB}^{*}$ converges in probability to $q_{LB}^{*}$. Then, I use the Delta method to prove asymptotic normality for the interior case when $\tilde q < 1$, which is guaranteed when $q \in (0,1)$  and $F_{Y_{1}\mid G=0, S_2 = 1}\left(Q_{Y_{1}\mid G=1,S_2 = 1}(q\pi_1 + 1 - \pi_1)\right) +1 - \pi_{0} < 1.$
    
    \begin{itemize}
        \item $\hat q_{LB}^{*} \xrightarrow{p} q_{LB}^{*}$
    \end{itemize}

    From Lemma \ref{prop:AItheo54}, it follows that $\hat \pi_{0} \xrightarrow{p} \pi_{0}$ and $\hat \pi_{1} \xrightarrow{p} \pi_{1}$. From the Glivenko–Cantelli theorem, I know that the empirical distribution function converges almost surely to the true distribution function, $\widehat F_{Y}(y) \xrightarrow{a.s.} F_{Y}(y)$. Because $Y$ is continuous with compact support, it follows that $\widehat Q_{Y}(q) \xrightarrow{p} Q_{Y}(q)$. Then, by the continuous mapping theorem, 
    \begin{align}
        \widehat Q_{Y_{1} \mid G=1, S_{2} = 1}^{}(q\hat\pi_{1}+1 - \hat\pi_{1}) & \xrightarrow{p} Q_{Y_{1}\mid G=1,S_2 = 1}(q\pi_1 + 1 - \pi_1)\\
        \widehat F_{Y_{1} \mid G = 0, S_{2} = 1}^{}\left(\widehat Q_{Y_{1} \mid G=1, S_{2} = 1}^{}(q\hat\pi_{1}+1 - \hat\pi_{1}) \right) + 1 - \hat\pi_{0} &\xrightarrow{p} \notag \\ F_{Y_{1}\mid G=0, S_2 = 1}\left(Q_{Y_{1}\mid G=1,S_2 = 1}(q\pi_1 + 1 - \pi_1)\right) +1 - \pi_{0}  = \tilde q
    \end{align}
    \begin{align}
        \hat{ q}_{LB}^{*} = \min\left\{\widehat F_{Y_{1} \mid G = 0, S_{2} = 1}^{}\left(\widehat Q_{Y_{1} \mid G=1, S_{2} = 1}^{}(q\hat\pi_{1}+1 - \hat\pi_{1}) \right) + 1 - \hat\pi_{0} , 1\right\} \xrightarrow{p} \min\{\tilde q,1\} = q_{LB}^{*}
    \end{align}

    \begin{itemize}
        \item Interior case, $\tilde q < 1$
    \end{itemize}
    
    In an abuse of notation, denote:
    {\small
    \begin{align*}
        \widehat{ \Lambda^{LB}}(q) = \widehat Q_{Y_{2} \mid G = 1, S_{2} = 1}^{}(q\hat\pi_{1}) -  \widehat  Q_{Y_{2} \mid G = 0, S_{2} = 1}^{}\left(\widehat F_{Y_{1} \mid G = 0, S_{2} = 1}^{}\left(\widehat Q_{Y_{1} \mid G=1, S_{2} = 1}^{}(q\hat\pi_{1}+1 - \hat\pi_{1}) \right) + 1 - \hat\pi_{0}\right) = \widehat \Lambda_{q}\left(\hat \pi_{1},\hat \pi_{0}\right) \\
        \Lambda^{LB}(q) = Q_{Y_{2} \mid G = 1, S_{2} = 1}^{}(q\pi_{1}) -  Q_{Y_{2} \mid G = 0, S_{2} = 1}^{}\left(F_{Y_{1} \mid G = 0, S_{2} = 1}^{}\left(Q_{Y_{1} \mid G=1, S_{2} = 1}^{}(q\pi_{1}+1 - \pi_{1}) \right) + 1 - \pi_{0}\right) = \Lambda_{q}\left(\pi_{1},\pi_{0}\right)
    \end{align*}
    }%
    Using a first-order Taylor approximation, I can write:
    \begin{align} 
        \widehat \Lambda_{q}\left(\hat \pi_{1},\hat \pi_{0}\right) & \approx \widehat \Lambda_{q}\left(\pi_{1},\pi_{0}\right) + \frac{\partial \widehat \Lambda_{q}\left(\pi_{1},\pi_{0}\right)}{\partial \pi_{1}}\left(\hat\pi_{1} - \pi_{1}\right) + \frac{\partial \widehat \Lambda_{q}\left(\pi_{1},\pi_{0}\right)}{\partial \pi_{0}}\left(\hat\pi_{0} - \pi_{0}\right) \\
        \widehat \Lambda_{q}\left(\hat \pi_{1},\hat \pi_{0}\right) - \Lambda_{q}\left(\pi_{1},\pi_{0}\right) & \approx \widehat \Lambda_{q}\left(\pi_{1},\pi_{0}\right) - \Lambda_{q}\left(\pi_{1},\pi_{0}\right) + \frac{\partial\widehat \Lambda_{q}\left(\pi_{1},\pi_{0}\right)}{\partial \pi_{1}}\left(\hat\pi_{1} - \pi_{1}\right) + \frac{\partial \widehat \Lambda_{q}\left(\pi_{1},\pi_{0}\right)}{\partial \pi_{0}}\left(\hat\pi_{0} - \pi_{0}\right) 
    \end{align}
    From Corollary \ref{prop:cicasdist}, it follows that 
    \begin{equation}
      \sqrt{n} \left( \widehat \Lambda_{q}(\pi_{1},\pi_{0})-\Lambda_{q}(\pi_{1},\pi_{0})\right) \xrightarrow{d} \mathcal{N}\left(0,V_{q\pi_{1}}^{12}+ V_{q + 1 - \pi_{0}}^{02}+ V_{q}^{01}+ V_{q\pi_{1} + 1 - \pi_{1}}^{11}\right)  .\label{eq:ap6bat}
    \end{equation}
    From Corollary \ref{prop:cicasdist}, and provided that $\widehat \Lambda_{q}(\pi_{1},\pi_{0})$ uniformly converges to $\Lambda_{q}(\pi_{1},\pi_{0})$, I can use the mean value theorem to write:
    \begin{equation}\label{eq:ap6bi}
        \frac{\partial\widehat \Lambda_{q}\left(\pi_{1},\pi_{0}\right)}{\partial \pi_{1}} = \lim_{h \to 0} \frac{\widehat \Lambda_{q}\left(\pi_{1} + h,\pi_{0}\right) - \widehat \Lambda_{q}\left(\pi_{1},\pi_{0}\right)}{h} \xrightarrow{p} \lim_{h \to 0} \frac{ \Lambda_{q}\left(\pi_{1} + h,\pi_{0}\right) -  \Lambda_{q}\left(\pi_{1},\pi_{0}\right)}{h} = \frac{\partial \Lambda_{q}\left(\pi_{1},\pi_{0}\right)}{\partial \pi_{1}}
    \end{equation}
    From equation (\ref{eq:ap6bi}), Lemma \ref{prop:AItheo54} and Slutsky theorem, it follows that:
    \begin{equation}
        \frac{\partial\widehat \Lambda_{q}\left(\pi_{1},\pi_{0}\right)}{\partial \pi_{1}}\left(\hat\pi_{1} - \pi_{1}\right) \xrightarrow{d} \mathcal{N}\left(0, \left(\frac{\partial \Lambda_{q}\left(\pi_{1},\pi_{0}\right)}{\partial \pi_{1}}\right)^2V_{\pi_{1}} \right) \label{eq:ap6hiru}
    \end{equation}
    Similarly,
    \begin{equation}
        \frac{\partial\widehat \Lambda_{q}\left(\pi_{1},\pi_{0}\right)}{\partial \pi_{0}}\left(\hat\pi_{0} - \pi_{0}\right) \xrightarrow{d} \mathcal{N}\left(0, \left(\frac{\partial \Lambda_{q}\left(\pi_{1},\pi_{0}\right)}{\partial \pi_{0}}\right)^2V_{\pi_{0}} \right)\label{eq:ap6lau}
    \end{equation}
    Combining equations (\ref{eq:ap6bat}) and (\ref{eq:ap6hiru}) and (\ref{eq:ap6lau}), it follows that
    \begin{equation}
         \sqrt{n}\left(\widehat \Lambda_{q}\left(\hat \pi_{1},\hat \pi_{0}\right) - \Lambda_{q}\left(\pi_{1},\pi_{0}\right)\right) \xrightarrow{d} \mathcal{N}\left(0,\varsigma_{LB}^{2}\right),
    \end{equation}
    where
    \begin{equation}
        \varsigma_{LB}^{2} = \left(V_{q\pi_{1}}^{12}+ V_{q + 1 - \pi_{0}}^{02}+ V_{q}^{01}+ V_{q\pi_{1} + 1 - \pi_{1}}^{11}\right) +
        \left(\frac{\partial \Lambda_{q}\left(\pi_{1},\pi_{0}\right)}{\partial \pi_{1}}\right)^2V_{\pi_{1}} 
        + \left(\frac{\partial \Lambda_{q}\left(\pi_{1},\pi_{0}\right)}{\partial \pi_{0}}\right)^2V_{\pi_{0}}.
    \end{equation}
    The second and third terms in $\varsigma_{LB}^{2}$ come from estimating the proportion of AO units in each of the groups. If these proportions were known, these two terms would disappear. Similarly, the less precisely these proportions are estimated, i.e., the larger $V_{\pi_1}$ and $V_{\pi_0}$ are, the more they contribute to the variance of the bounds.

\end{proof}
\subsection{Proof of Proposition \ref{prop:bounds_attao}} \label{proof:bounds_attao}
This proof uses the following result from \textcite{horowitz_identification_1995}:
    \begin{lemma} \label{prop:HMprop4}
        Let $Y$ be a continuous random variable and a mixture of two random variables, $Y_{0}$ and $Y_{1}$, with known mixing probability, $p$. The c.d.f. of the observed $Y$ is $F(y) = (1-p) F_{0}(y) + p F_{1}(y)$, where $F_{0}$ and $F_{1}$ are the c.d.f.s of $Y_{0}$ and $Y_{1}$ respectively. The c.d.f. of $Y$ after trimming the $1-p$ upper tail of $Y$ is given by $G(y) = \min\left\{\frac{F(y)}{p}, 1\right\}$. Then, $\int_{-\infty}^{\infty}ydG(y) \leq \int_{-\infty}^{\infty}y dF_{1}(y)$ and $\int_{-\infty}^{\infty}ydG(y)$ is a sharp lower bound for $\int_{-\infty}^{\infty}y dF_{1}(y)$. The c.d.f. of $Y$ after trimming the $1-p$ bottom tail of $Y$ is given by $G^{'}(y) = \max\left\{0,\frac{F(y)-(1-p)}{p}\right\}$. Then, $\int_{-\infty}^{\infty}ydG^{'}(y) \geq \int_{-\infty}^{\infty}y dF_{1}(y)$ and $\int_{-\infty}^{\infty}ydG^{\prime}(y)$ is a sharp upper bound for $\int_{-\infty}^{\infty}y dF_{1}(y)$. \\
        Proof: See \textcite[Proposition 4]{horowitz_identification_1995}.
    \end{lemma}
\begin{proof}
    From equation (\ref{eq:princ_did}), it follows that the lower bound for the $\attao$ is given by the lower bound of $\E[\ddot Y_{i} \mid G_{i} = 1, V_{i} = AO]$ and the upper bound of $\E[\ddot Y_{i} \mid G_{i} = 0, V_{i} = AO]$. To prove Proposition \ref{prop:bounds_attao}, I prove next that $\E[\ddot Y_{i} \mid G_{i} = 1, S_{i2} = 1, \ddot Y_{i} \leq \ddot y_{\pi_{1}}^{1} ]$ is a lower bound for $\E[\ddot Y_{i} \mid G_{i} = 1, V_{i} = AO]$. Similar arguments can be employed for the upper bound of $\E[\ddot Y_{i} \mid G_{i} = 0, V_{i} = AO]$ and the upper bound of the $\attao$. 

    Let $F(\ddot y)$ be the c.d.f. of $\ddot Y$ given $S_{2} = 1$, $G = 1$. It can be expressed as the convex combination of the c.d.f. of $\ddot Y$ for the AO and OT strata:
    \begin{equation}
        F(\ddot y) = \pi_{1} F^{AO}(\ddot y) + (1-\pi_{1}) F^{OT}(\ddot y),
    \end{equation}
    where $F^{v}(\ddot y)$ is the c.d.f. of $\ddot Y$ given $G =1$ and $V = v$. $\pi_{1} =  Pr(S_{i2}(0) = 1 \mid G_{i} = 1, S_{i2}(1) = 1)$, as defined in Proposition \ref{prop:bounds_attao}, is taken as given. Then, it follows that

    \begin{equation}
        \E[\ddot Y_{i} \mid G_{i} = 1, S_{i2} = 1, \ddot Y_{i} \leq \ddot y_{\pi_{1}}^{1} ] = \frac{1}{\pi_{1}} \int_{-\infty}^{\ddot y_{\pi_{1}}^{1}} y dF(\ddot y) \leq \int_{-\infty}^{\infty}ydF^{AO}(\ddot y) = \E[\ddot Y \mid G = 1, V_{} = AO],
    \end{equation}
    where $\ddot y_{\pi_{1}}^{1}$ is the $\pi_{1}$th quantile of the distribution of $\ddot Y$ given $S_{2} = 1$ and $G= 1$. The inequality is given by Lemma \ref{prop:HMprop4}. From Lemma \ref{prop:HMprop4} it also follows that this bound is sharp, provided that $\pi_{1}$ is uniquely identified in the data.
\end{proof}
\subsection{Proof of Proposition \ref{prop:normal_attao}} \label{proof:normal_attao}
This proof is similar to the proofs of Propositions 2 and 3 in \textcite{lee_training_2009} and uses the following results from \textcite{newey_chapter_1994}:
\begin{lemma} \label{prop:NFtheo26}
    Suppose that $z_{i}, (i=1,2...),$ are i.i.d., $\mathbb{E}[g(z,\theta)] = 0$ if and only if $\theta = \theta^{*}$; $\theta^{*}\in \Theta$, which is compact; $g(z,\theta)$ is continuous at each $\theta \in \Theta$ with probability one; $\E[\sup_{\theta \in \Theta} \rVert g(z,\theta) \rVert] < \infty$. Then, $\hat \theta  \xrightarrow{p}\theta^{*}$. \\
    Proof: See \textcite[Theorem 2.6]{newey_chapter_1994}.
\end{lemma}
\begin{lemma} \label{prop:NFtheo72}
    Suppose that $\hat g_{n}(\hat \theta)^{\prime}\hat g_{n}(\hat \theta) \leq \inf_{\theta \in \Theta}\hat g_{n}( \theta)^{\prime}\hat g_{n}( \theta) + o_{p}(n^{-1}) $, $\hat \theta \xrightarrow{p}\theta^{*}$, where there is $g^{*}(\theta^{*})$ is such that (i) $g^{*}(\theta^{*}) = 0$, (ii) $g^{*}(\theta^{*})$ is differentiable at $\theta^{*}$ with derivative $G$ such that $G^{\prime}G$ is nonsingular, (iii) $\theta^{*}$ is an interior point of $\Theta$, (iv) $\sqrt{n} \hat{g}_n\left(\theta^{*}\right) \xrightarrow{\mathrm{d}} N(0, \Sigma) ;(\mathrm{v})$ for any $\delta_n \longrightarrow 0, \sup _{\left\|\theta-\theta^{*}\right\| \leqslant \delta_n} \sqrt{n} \| \hat{g}_n(\theta)-$ $\hat{g}_n\left(\theta^{*}\right)-g^{*}(\theta) \| /\left[1+\sqrt{n}\left\|\theta-\theta^{*}\right\|\right] \xrightarrow{\mathrm{p}} 0$. Then $\sqrt{n}\left(\hat{\theta}-\theta^{*}\right) \xrightarrow{\mathrm{d}} \mathcal{N}\left(0,\left(G^{\prime} G\right)^{-1} G^{\prime}  \Sigma G\left(G^{\prime}  G\right)^{-1}\right)$. \\
    Proof: See \textcite[Theorem 7.2]{newey_chapter_1994}.
\end{lemma}

\begin{proof}
    To prove consistency, it suffices to prove that $\frac{\sum_{i} S_{i2}G_{i} \indicator(\ddot Y_{i} \leq \hat{ \ddot y}_{\hat \pi_{1}}^{1})\ddot Y_{i} }{\sum_{i}S_{i2}G_{i}\indicator(\ddot Y_{i} \leq \hat{ \ddot y}_{\hat \pi_{1}}^{1})}$ is a consistent estimator of $\E[\ddot Y_{i} \mid G_{i} = 1, S_{i2} = 1, \ddot Y_{i} \leq \ddot y_{\pi_{1}}^{1} ]$. Similar steps can be used to prove that $\frac{\sum_{i}S_{i2}(1-G_{i})\indicator(\ddot Y_{i} \geq \hat{\ddot y}_{1 - \hat \pi_{0}}^{0})\ddot Y_{i}}{\sum_{i}S_{i2}(1-G_{i})\indicator(\ddot Y_{i} \geq \hat{\ddot y}_{1 - \hat \pi_{0}}^{0})}$ is a consistent estimator of $ \E[\ddot Y_{i} \mid G_{i} = 0, S_{i2} = 1, \ddot Y_{i} \leq \ddot y_{\pi_{0}}^{0} ]$. Then, by the Continuous Mapping Theorem for p-convergence, it follows that $\widehat{\Delta^{LB}}$ is a consistent estimator for $\Delta^{LB}$. Symmetric arguments can be used to prove that $\widehat{\Delta^{UB}}$ is consistent for $\Delta^{UB}$.

    Consider the case with $J=2$. For $j = 1$ positive monotonicity
    holds while for $j = 2$ negative monotonicity holds. Define the moment function
    \begin{equation*}
       g(z_i,\theta) =  \begin{pmatrix}
            (\ddot Y_i - \mu)S_{i2}G_i\indicator(\ddot Y_i \leq \ddot y_{\pi_{1}}^{1}) \\
            (\indicator(\ddot Y_i \leq \ddot y_{\pi_{1}}^{1} - \pi_{1}^{1})S_{i2}G_i \\
            (s_{i1}^{1}-\alpha_{1,1}^{1})G_i \\
            (s_{i1}^{1} - \alpha_{1,0}^{1})(1-G_i) \\
            (s_{i2}^{1} - \alpha_{2,0}^{1})(1-G_i) \\
            (s_{i2}^{2} - \alpha_{2,1}^{2})G_i \\
            \left(S_{i2} - \frac{\alpha_{1,1}^{1}\frac{\alpha_{2,0}^{1}}{\alpha_{1,0}^{1}} + \alpha_{2,1}^{2}}{\pi_{1}}\right)G_i
        \end{pmatrix},
    \end{equation*}
    where $\theta^{\prime} = (\mu, \ddot y_{\pi_{1}}^{1}, \alpha_{1,1}^{1},\alpha_{1,0}^{1},\alpha_{2,0}^{1},\alpha_{2,1}^{2}, \pi_{1})$, $\theta^{\prime *} = (\mu^{*}, \ddot y_{\pi_{1}}^{1*}, \alpha_{1,1}^{1*},\alpha_{1,0}^{1*},\alpha_{2,0}^{1*},\alpha_{2,1}^{2*}, \pi_{1}^{*})$, $\mu^{*} \equiv \E[\ddot Y_{i} \mid G_{i} = 1, S_{i2} = 1, \ddot Y_{i} \leq \ddot y_{\pi_{1}}^{1} ]$, $\alpha_{t, g}^{j*} \equiv Pr(s_{t}^{j} = 1 \mid G = g)$. The estimators provided in equations (\ref{eq:pi1}), (\ref{est:trimt}) and (\ref{est:lb}) are the solution to
    \begin{equation*}
        \min_{\theta \in \Theta} \left(\sum_i g(z_i,\theta)\right)^{\prime}\left(\sum_i g(z_i,\theta)\right).
    \end{equation*}
    This just-identified system yields only one solution as long as any of the following expectations are equal to zero: $\E[s_{1}^{1} \mid G=1],\E[s_{1}^{1} \mid G=0],\E[s_{2}^{1} \mid G=0],\E[s_{2}^{2} \mid G=1]$. $\Theta$ is compact provided that it consists of the bounds of the support for the trimmed mean and quantiles, and the unit interval for the probabilities. Given that the supports are bounded and $g(z,\theta)$ is continuous at each $\theta \in \Theta$, from Lemma \ref{prop:NFtheo26} it follows that $\frac{\sum_{i} S_{i2}G_{i} \indicator(\ddot Y_{i} \leq \hat{ \ddot y}_{\hat \pi_{1}}^{1})\ddot Y_{i} }{\sum_{i}S_{i2}G_{i}\indicator(\ddot Y_{i} \leq \hat{ \ddot y}_{\hat \pi_{1}}^{1})}$ is a consistent estimator of $\E[\ddot Y_{i} \mid G_{i} = 1, S_{i2} = 1, \ddot Y_{i} \leq \ddot y_{\pi_{1}}^{1} ]$.

    Once consistency is proved, asymptotic normality follows from Lemma \ref{prop:NFtheo72} as shown in \textcite[Proposition 3]{lee_training_2009} and \textcite{kang_robust_2024}.

\end{proof}

\subsection{Proof of Proposition \ref{prop:cross_proportion}} \label{proof:cross_proportion}
\begin{proof}
Let $\pi_{g}$ be the proportion of AO units in group $g$ in the post-treatment period, that is, $\pi_{g} = Pr(V_{i} =AO \mid G_{i}=g, S_{i2} = 1)$. Using the Law of Total Probability, I can write:
\begin{equation}
    Pr(V_{i}= AO \mid G_{i} =g) = Pr(V_{i} =AO, S_{i2} = 1 \mid G_{i} = g) +Pr(V_{i} =AO, S_{i2} = 0 \mid G_{i} = g). \label{eq:ap11bat}
\end{equation}
Because $V_{i} =AO \implies S_{i2} = 1$, it follows that $Pr(V_{i} =AO, S_{i2} = 0 \mid G_{i} = g) = 0$. Therefore, equation \eqref{eq:ap11bat} becomes:
\begin{align}
     Pr(V_{i}= AO \mid G_{i} =g) &= Pr(V_{i} =AO, S_{i2} = 1 \mid G_{i} = g) \\
     & = Pr(V_{i} =AO \mid G_{i}=g,S_{i2} =1) Pr(S_{i2} = 1\mid G_{i}=g)\\
     &=\pi_{g} \E[S_{i2}\mid G_{i}=g] \label{eq:ap11bi}
\end{align}
By assumption \ref{as:rand_samp}, this probability is the same in the pre-treatment period and in the post-treatment period. Recall that $V_{i}$ is defined using the potential selection outcomes in the post-treatment period. Assumption \ref{as:rand_samp} ensures that units sampled in the pre-treatment period have the same distribution of post-treatment selection outcomes as units sampled in the post-treatment period.

Using the same logic for the pre-treatment period, it follows that:
\begin{equation}
    Pr(V_{i}= AO \mid G_{i} =g) = Pr(V_{i} =AO, S_{i1} = 1 \mid G_{i} = g) +Pr(V_{i} =AO, S_{i1} = 0 \mid G_{i} = g). \label{eq:ap11hiru}
\end{equation}
By Assumption \ref{as:abstate}, it follows that $V_{i} = AO \implies S_{i1} = 1$, which implies that $Pr(V_{i} =AO, S_{i1} = 0 \mid G_{i} = g) = 0$. That is, all the units sampled in the pre-treatment period that would belong to the AO stratum in the post-treatment period have to have $S_{i1} = 1$. From this, it follows that
\begin{align}
     Pr(V_{i}= AO \mid G_{i} =g) & = Pr(V_{i} =AO, S_{i1} = 1 \mid G_{i} = g) \\
     & = Pr(V_{i} =AO \mid G_{i}=g,S_{i1} =1) Pr(S_{i1} = 1\mid G_{i}=g)\\
     &=\varpi_{g} \E[S_{i1}\mid G_{i}=g] \label{eq:ap11lau}
\end{align}
Finally, combining equations \eqref{eq:ap11bi} and \label{eq:ap11lau}, it follows that
\begin{equation}
    \pi_{g} \E[S_{i2}\mid G_{i}=g] = \varpi_{g} \E[S_{i1}\mid G_{i}=g]
\end{equation}
\begin{equation} \label{eq:ap_varpi}
    \varpi_{g} = \pi_{g}\frac{\E[S_{i2}\mid G_{i}=g]}{\E[S_{i1}\mid G_{i}=g]}
\end{equation}
\end{proof}

\subsection{Proof of Proposition \ref{prop:cross_proportion_bounds}} \label{proof:cross_proportion_bounds}
If Assumption \ref{as:abstate} does not hold, then $\varpi_{g}$ can only be partially identified.

\begin{proof}
    \begin{equation*}
        \varpi_{g} = Pr(V_i =AO \mid G_i = g, S_{i1} = 1) = \frac{Pr(V_i =AO, S_{i1} =1 \mid G_i = g)}{\E[S_{i1}=1\mid G_{i} =g]}
    \end{equation*}
    From equation \eqref{eq:ap11bi}, 
    \begin{equation*}
         Pr(V_i =AO \mid G_i = g) = \pi_{g} \E[S_{i2} \mid G{i} = g].
    \end{equation*}
    By the Law of Total Probability,
    \begin{align}
         Pr(V_i =AO \mid G_i = g) &= Pr(V_i =AO, S_{i1} =1 \mid G_i = g)  + Pr(V_i =AO, S_{i1} =0 \mid G_i = g) , \\
        \pi_{g} \E[S_{i2} \mid G{i} = g] &= Pr(V_i =AO, S_{i1} =1 \mid G_i = g)  + Pr(V_i =AO, S_{i1} =0 \mid G_i = g)  .
    \end{align}
    So that,
    \begin{equation}
        Pr(V_i =AO, S_{i1} =1 \mid G_i = g) = \pi_{g} \E[S_{i2} \mid G{i} = g] - Pr(V_i =AO, S_{i1} =0 \mid G_i = g)  
    \end{equation}
    Because 
    \begin{equation}
        0 \leq Pr(V_i =AO, S_{i1} =0 \mid G_i = g)   \leq Pr(S_{i1} = 0 \mid G_{i} =g) = 1-\E[S_{i1}\mid G_{i} = g],
    \end{equation}
    It follows that 
    \begin{equation}
        \pi_{g} \E[S_{i2} \mid G{i} = g] - (1-\E[S_{i1}\mid G_{i} = g]) \leq Pr(V_i =AO, S_{i1} =1 \mid G_i = g)
    \end{equation}
    and
    \begin{equation}
        Pr(V_i =AO, S_{i1} =1 \mid G_i = g) \leq \pi_{g} \E[S_{i2} \mid G{i} = g]
    \end{equation}
    All in all, it follows that
    \begin{equation}
        \varpi_{g} = \frac{Pr(V_i =AO, S_{i1} =1 \mid G_i = g)}{\E[S_{i1}=1\mid G_{i} =g]} \leq \min\left\{1,\pi_{g}\frac{\E[S_{i2}\mid G_{i}=g]}{\E[S_{i1}\mid G_{i}=g]}\right\}.
    \end{equation}
    and
    \begin{equation}
         \varpi_{g} = \frac{Pr(V_i =AO, S_{i1} =1 \mid G_i = g)}{\E[S_{i1}=1\mid G_{i} =g]} \geq\max\left\{0, \frac{\pi_{g}\E[S_{i2}\mid G_{i}= g]-\left(1 - \E[S_{i1}\mid G_{i} = g]\right)}{\E[S_{i1}\mid G_{i} = g]}\right\}
    \end{equation}
\end{proof}

The upper bound corresponds to the case where all the units in the AO stratum are observed in the pre-treatment period, and is given by equation \eqref{eq:ap_varpi}. This is always the case under Assumption \ref{as:abstate}. Without this assumption, however, the proportion of units observed in the pre-treatment period may be smaller than the proportion of AO. In that case, the upper bound of $\varpi_{g}$ is equal to 1, implying that all the observed units belong to the AO stratum. The lower bound corresponds to the case in which all the units not observed in the pre-treatment period belong to the $AO$ stratum in the post-treatment period. Consequently, if the proportion of units with not observed outcomes in the pre-treatment period is larger than the proportion of the AO stratum, then the lower bound is just 0.

\renewcommand{\theequation}{C\arabic{equation}}
\renewcommand{\thetable}{C\arabic{table}}
\renewcommand{\thefigure}{C\arabic{figure}}
\setcounter{equation}{0}
\setcounter{table}{0}
\setcounter{figure}{0}

\section{Supplementary material} \label{ap:supmat}
Table \ref{tab:example_mult_sourc} illustrates Lemma \ref{prop:proportions_multiple_sources} under Assumptions \ref{as:cond_mono} and \ref{as:intersec} when multiple sources of attrition are available in the data. Table \ref{tab:map_subp_str} shows how the subpopulations based on the observed selection outcomes identify the proportions of the 4 strata.

\begin{landscape}
    \begin{table}[H]
    \caption{Principal Strata in the post-treatment period with multiple sources of attrition}
    \label{tab:example_mult_sourc}
\vspace{-0.1cm} 
\begin{center} 
\begin{tabular}{|c|ccc|ccc|ccc|c|}
\hline
 & \multicolumn{3}{c|}{Source 1} & \multicolumn{3}{c|}{Source 2} & \multicolumn{3}{c|}{Overall Selection} & \\ 
\textbf{Row} & $\bs{s^{1}(0)}$ &$\bs{s^{1}(1)}$ & $\bs{V^{1}}$ &
 $\bs{s^{2}(0)}$ &$\bs{s^{2}(1)}$ & $\bs{V^{2}}$ &
  $\bs{S^{}(0)}$ &$\bs{S^{}(1)}$ & $\bs{V^{}}$ &
\textbf{Ruled out} \\ \hline
1 & 0 & 0 & $NO_{1}$ & 0 & 0 & $NO_{2}$ & 0 & 0 & $NO$ & Mutual Exclusivity \\
2 & 0 & 0 & $NO_{1}$ & 0 & 1 & $OT_{2}$ & 0 & 0 & $NO$ & Mutual Exclusivity, Monotonicity \\
3 & 0 & 0 & $NO_{1}$ & 1 & 0 & $OC_{2}$ & 0 & 0 & $NO$ & Mutual Exclusivity \\
4 & 0 & 0 & $NO_{1}$ & 1 & 1 & $AO_{2}$ & 0 & 0 & $NO$ & No \\
5 & 1 & 1 & $AO_{1}$ & 0 & 0 & $NO_{2}$ & 0 & 0 & $NO$ & No \\
6 & 1 & 1 & $AO_{1}$ & 0 & 1 & $OT_{2}$ & 0 & 1 & $OT$ & Negative Monotonicity Source 2 \\
7 & 1 & 1 & $AO_{1}$ & 1 & 0 & $OC_{2}$ & 1 & 0 & $OC$ & No \\
8 & 1 & 1 & $AO_{1}$ & 1 & 1 & $AO_{2}$ & 1 & 1 & $AO$ & No \\
9 & 0 & 1 & $OT_{1}$ & 0 & 0 & $NO_{2}$ & 0 & 0 & $NO$ & Mutual Exclusivity \\
10 & 0 & 1 & $OT_{1}$ & 0 & 1 & $OT_{2}$ & 0 & 1 & $OT$ & Mutual Exclusivity, Monotonicity \\
11 & 0 & 1 & $OT_{1}$ & 1 & 0 & $OC_{2}$ & 0 & 0 & $NO$ & No Overlap \\
12 & 0 & 1 & $OT_{1}$ & 1 & 1 & $AO_{2}$ & 0 & 1 & $OT$ & No \\
13 & 1 & 0 & $OC_{1}$ & 0 & 0 & $NO_{2}$ & 0 & 0 & $NO$ & Mutual Exclusivity, Monotonicity \\
14 & 1 & 0 & $OC_{1}$ & 0 & 1 & $OT_{2}$ & 0 & 0 & $NO$ & No Overlap, Monotonicity x 2\\
15 & 1 & 0 & $OC_{1}$ & 1 & 0 & $OC_{2}$ & 1 & 0 & $OC$ & Mutual Exclusivity, Monotonicity \\
16 & 1 & 0 & $OC_{1}$ & 1 & 1 & $AO_{2}$ & 1 & 0 & $OC$ & Positive Monotonicity Source 1 \\
\hline 
\end{tabular}
 
\end{center}
\vspace{-0.2cm}
\begin{minipage}{1\linewidth \setstretch{0.75} } 
{\scriptsize Notes: For Source 1 positive monotonicity is assumed, i.e., $s^{1}(1) \geq s^{1}(0)$. For Source 2 negative monotonicity is assumed, i.e., $s^{2}(1) \leq s^{2}(0)$. All the rows that involve the $OC_{1}$ stratum according to Source 1 or the $OT_{2}$ according to Source 2 are ruled out by these monotonicity assumptions. All the rows such that $s^{1}(0) = s^{2}(0) = 0$ or $s^{1}(1) = s^{2}(1) = 0$ or both are ruled out by the fact that sources are mutually exclusive. Finally, rows 11 and 14, which constitute the cases in which a unit belongs to the OT stratum according to one source and to the OC stratum according to the other, are ruled out by Assumption \ref{as:intersec}. All in all, we are left with 5 rows.
}
 \end{minipage}
\end{table}
\end{landscape}

\begin{table}[H] \centering
\caption{Map between subpopulations defined based on observed outcomes and Principal Strata.}
\begin{tabular}{|c|c|c|}
\hline
\begin{tabular}[c]{@{}c@{}}Subpopulation \\ $(s^{1}, s^{2}, G)$ \end{tabular} & Strata & Rows in Table \ref{tab:example_mult_sourc} \\ \hline
$(1,1,1)$ & AO, OT & 8, 12 \\
$(1,0,1)$ & NO, OC & 5, 7 \\
$(0,1,1)$ & NO & 4 \\
$(1,1,0)$ & OC, AO & 7, 8 \\
$(1,0,0)$ & NO & 5 \\
$(0,1,0)$ & NO, OT & 4, 12 \\
\hline
\end{tabular}
\label{tab:map_subp_str}
\begin{minipage}{0.81\linewidth  \setstretch{0.75}}
{\scriptsize  Notes: The first column represents the partition of the sample into 6 different subpopulations based on the observed outcomes. The second column represents the composition in terms of the principal strata of these subpopulations. The third column is the correspondence between the principal strata in the second column con the rows in Table \ref{tab:example_mult_sourc}. For example, consider the first row. This subpopulation, $(s^{1}, s^{2}, G) = (1,1,1)$ are the treated units that are observed. In terms of potential outcomes, these are the units such that $s^{1}(1) = s^{2}(1) = 1 \implies S(1) = 1$. So that these units are AO or OT. If we look at table \ref{tab:example_mult_sourc}, once the rows that are ruled out by assumption are excluded, the units in this subpopulation can only belong to rows 8 and 12. }
 \end{minipage}
\end{table}

\begin{table}[H]
    \caption{Post-treatment selection outcomes for units observed in the pre-treatment period}
    \label{tab:samplesize}
\vspace{-0.1cm} 
\begin{center}
 \begin{tabular}{l c | cc cc}
\toprule
& \multicolumn{1}{c}{Pre-treatment} & \multicolumn{4}{c}{Post-treatment} \\
\cmidrule(lr){2-2}\cmidrule(lr){3-6}
& \multicolumn{1}{c}{$S_1 = 1$}
& \multicolumn{2}{c}{$S_1 = 1,\; S_2 = 0$}
& \multicolumn{1}{c}{$S_1 = 1,\; S_2 = 1$} \\
\cmidrule(lr){3-4}\cmidrule(lr){5-5}
& & \makecell[c]{Unemployment \\ /Self-employment} & \makecell[c]{Non-response} & \makecell[c]{Paid \\employment} & Total \\
\hline
Control          &       661  & 124&  150 &   387 & 661\\
Treatment &  792 &  166 & 125  &  501  & 792\\

\hline
Total  &       1453 &  290 & 275  &  888 & 1453 \\
\hline
\end{tabular}
\end{center}
\vspace{-0.2cm}
\begin{minipage}{1\linewidth \setstretch{0.75} } 
{\scriptsize Notes: The first column indicates the number of units with salaried earnings observed in the pre-treatment period. Columns 2-4 indicate the post-treatment selection outcome for these units. Specifically, Column 2 shows the number of units observed in the first period who did not have paid employment in the second period. Column 3 reports the number of units that had paid employment in the first period but did not answer the survey in the second period. Column 4 shows the number of units for whom salaried earnings are observed both in the pre-treatment and the post-treatment periods.
}
 \end{minipage}
\end{table}

\end{spacing}
\newrefcontext[sorting=nyt]
\end{document}